\documentclass[acmsmall,screen]{acmart}

\usepackage{wrapfig}
\usepackage{booktabs}   
\usepackage{subcaption} 

\usepackage{color}

\definecolor{DOrange}{RGB}{145,20,20}
\definecolor{Red}{RGB}{255,0,0}
\definecolor{DeepGreen}{RGB}{5,102,8}

\usepackage{wrapfig}

\usepackage{listings} 
\usepackage{semantic}
\usepackage{tikz,pgf}
\usepackage{stmaryrd}
\usepackage{xcolor}
\setpremisesspace{1cm}

\usepackage{authorcomments}
\newcommand{\name}{\textsf{Morpheus}\xspace}
\usepackage{stmaryrd}

\newcommand{\eip}{\mbox{$\mathit{e}$}}
\newcommand{\lp}{\mbox{$\mathit{l}_{p}$}}
\newcommand{\xp}{\mbox{$\mathit{x}$}}
\newcommand{\vp}{\mbox{$\mathit{v}$}}
\newcommand{\cp}{\mbox{$\mathit{c}$}}
\newcommand{\el}{\mbox{$\varepsilon$}}
\newcommand{\qual}{\mbox{$\mathit{Q}$}}
\definecolor{mygreen}{rgb}{0,0.6,0}
\definecolor{mygray}{rgb}{0.5,0.5,0.5}
\definecolor{mymauve}{rgb}{0.58,0,0.82}

\renewcommand{\S}[1]{\mbox{$\mathsf{#1}$}}
\newcommand{\B}[1]{\mbox{$\mathbf{#1}$}}
\newcommand{\rp}{\mbox{$\ell$}}
\newcommand{\M}[1]{\mbox{\textcolor{red}{$\mathsf{#1}_{{\sf m}}$}}}
\newcommand{\h}[1]{\mbox{$\hat{#1}$}}

\tikzoption{base font}{\def\tikz@base@textfont{#1}}
\tikzset{
  base font=\sffamily,
}
\usepackage{amsmath}
\usepackage{supertabular}
\usepackage{geometry}
\usepackage{ifthen}
\usepackage{alltt}
\usepackage{color}

\usepackage{multirow}
\usepackage{multicol}
\usepackage{semantic}
\usepackage{algpseudocode}
\usepackage[noend,noline]{algorithm2e}
\usepackage{pgfplots, pgfplotstable}
\usepackage{xspace}

\begin{document}

\title{\name: Automated Safety Verification of Data-dependent Parser Combinator Programs}


\author{Ashish Mishra}
\affiliation{
  \position{}
  \department{Department of Computer Science}              
  \institution{Purdue University}            
  \country{USA}                    
}
\email{mishr115@purdue.edu}          

\author{Suresh Jagannathan}
\affiliation{
  \position{}
  \department{Department of Computer Science}             
  \institution{Purdue University}           
  \country{USA}                   
}
\email{suresh@cs.purdue.edu}         



\begin{abstract}
  Parser combinators are a well-known mechanism used for the
  compositional construction of parsers, and have shown to be
  particularly useful in writing parsers for rich grammars with
  data-dependencies and global state.  Verifying applications written
  using them, however, has proven to be challenging in large part
  because of the inherently effectful nature of the parsers being
  composed and the difficulty in reasoning about the arbitrarily rich
  data-dependent semantic actions that can be associated with parsing
  actions. In this paper, we address these challenges by defining a
  parser combinator framework called \name\ equipped with abstractions
  for defining composable effects tailored for parsing and semantic
  actions, and a rich specification language used to define safety
  properties over the constituent parsers comprising a program. Even
  though its abstractions yield many of the same expressivity benefits
  as other parser combinator systems, \name\ is carefully engineered
  to yield a substantially more tractable automated verification
  pathway. We demonstrate its utility in verifying a number of
  realistic, challenging parsing applications, including several cases
  that involve non-trivial data-dependent relations.
\end{abstract}




\maketitle
\section{Introduction}
\label{sec:introduction}

Parsers are transformers that decode serialized, unstructured data
into a structured form. Although many parsing problems can be
described using simple context-free grammars (CFGs), numerous
real-world data formats (e.g., pdf~\cite{pdfreference},
dns~\cite{dns}, zip~\cite{zip}, etc.), as well as many programming
language grammars (e.g., Haskell, C, Idris, etc.) require their parser
implementations to maintain additional context information during
parsing. A particularly important class of context-sensitive parsers
are those built from {\it data-dependent grammars}, such as the ones
used in the data formats listed above. Such {\it data-dependent}
parsers allow parsing actions that explicitly depend on earlier parsed
data or semantic actions. Often, such parsers additionally use global
effectful state to maintain and manipulate context information. To
illustrate, consider the implementation of a popular class of {\it
  tag-length-data} parsers; these parsers can be used to parse image
formats like PNG or PPM images, networking packets formats like TCP,
etc., and use a parsed length value to govern the size of the input
payload that should be parsed subsequently.  The following BNF grammar
captures this relation for a simplified PNG image.
\begin{tabbing}
  \small
\hspace{.2in}\= \small {\sf png} ::= {\sf header} \ .  \ $\mathsf{chunk^{*}}$ \\
\> \small $\mathsf{chunk}$ ::= {\sf length } \ .  {\sf typespec} \ . {\sf content}
\end{tabbing}
The grammar defines a {\sf header} field followed by zero or more {\sf
  chunks}, where each {\sf chunk} has a single byte {\sf length} field
parsed as an unsigned integer, followed by a single byte chunk {\sf
  type specifier}. This is followed by zero or more bytes of actual
{\sf content}.
A useful data-dependent safety property that any
parser implementation for this grammar should satisfy is that
``\emph{the length of {\sf content} plus {\sf typespec} is
  equal to the value of {\sf length}}''.

Parser combinator
libraries~\citep{parsingwadler,parsec,parsinghutton,hammer} provide an
elegant framework in which to write parsers that have such
data-dependent features. These libraries simplify the task of writing
parsers because they define the grammar of the input language and
implement the recognizer for it at the same time. Moreover, since
combinator libraries are typically defined in terms of a
shallowly-embedded DSL in an expressive host language like
Haskell~\cite{indentation1,megaparsec} or OCaml~\cite{parsec}, parser
implementations can seamlessly use a myriad of features available in
the host language to express various kinds of data-dependent
relations. This makes them capable of parsing both CFGs as well as
richer grammars that have non-trivial semantic actions.  Consequently,
this style of parser construction has been adopted in many
domains~\cite{AfroozehAli2015Optr,indentation1,hammer}, a fact
exemplified by their support in many widely-used languages like
Haskell, Scala, OCaml, Java, etc.

Although parser combinators provide a way to easily write
data-dependent parsers, verifying the correctness (i.e., ensuring
  that all data dependencies are enforced) of parser
implementations written using them remains a challenging problem.
This is in large part due to the inherently effectful nature of the
parsers being composed, the pervasive use of rich higher-order
abstractions available in the combinators used to build them, and the
difficulty of reasoning about complex data-dependent semantic actions
triggered by these combinators that can be associated with a parsing
action.





This paper directly addresses these challenges.  We do so by imposing
modest constraints on the host language capabilities available to
parser combinator programs; these constraints \emph{enable} motly automated
reasoning and verification, \emph{without} comprising the ability to
specify parsers with rich effectful, data-dependent safety properties. We manifest
these principles in the design of a
deeply-embedded DSL for OCaml
called \name\ that we use to express and verify parsers and the
combinators that compose them.  Our design provides a novel (and, to
the best of our knowledge, first) automated verification pathway for
this important application class.  This paper makes the following
contributions:
\begin{enumerate}
\item It details the design of an OCaml DSL \name\ that allows
  compositional construction of {\it data-dependent} parsers using
  a rich set of primitive parsing combinators along with an expressive
  specification language for describing safety properties relevant to
  parsing applications.
\item It presents an automated refinement type-based verification
  framework that validates the correctness of \name\ programs with
  respect to their specifications and which supports fine-grained
  effect reasoning and inference to help reduce specification
  annotation burden.
\item It justifies its approach through a detailed evaluation study
  over a range of complex real-world parser applications that
  demonstrate the feasibility and effectiveness of the proposed
  methodology.
\end{enumerate}

The remainder of the paper is organized as follows.  The next section
presents a detailed motivating example to illustrate the challenges
with verifying parser combinator applications and presents a detailed overview of \name\ that
builds upon this example.  We formalize \name's
specification language and type system in Secs.~\ref{sec:syntax}
and~\ref{sec:typing}.  Details about \name's implementation and
benchmarks demonstrate the utility of our framework is given in
Sec.~\ref{sec:evaluation}.  Related work and conclusions are given in
Secs.~\ref{sec:related} and~\ref{sec:conc}, respectively.

\section{Motivation and \name Overview}
\label{sec:overview}
To motivate our ideas and give an overview of \name, consider a parser
for a simplified C language {\it declarations, expressions and
  typedefs} grammar. The grammar must handle context-sensitive
disambiguation of \textit{typenames} and
\textit{identifiers}~\footnote{https://web.archive.org/web/20070622120718/http://www.cs.utah.edu/research/projects/mso/goofie/grammar5.txt}.
Traditionally, C-parsers achieve this disambiguation via cumbersome
{\it lexer hacks}\footnote{https://www.lysator.liu.se/c/ANSI-C-grammar-l.html} which use feedback from the symbol table maintained in the parsing into the lexer to distinguish variables from types.  Once the
disambiguation is outsourced to the lexer-hack, the C-decl grammar can
be defined using a context-free-grammar. For instance, the left hand side, Figure~\ref{fig:cgrammars}, presents a simplified context-free grammar production for a C
declaration.
\begin{figure}
\begin{subfigure}{.5\textwidth}
  \begin{lstlisting}[escapechar=\@,basicstyle=\small\sf,numbers=left,breaklines=true]
  decl ::= @\textbf{typedef}@ . type-expr . id=rawident 
      |  @\textbf{extern}@ ...
      | ...
  typename ::= rawident 
  type-exp ::= "int" | "bool" 
  expr ::=  ... | id=rawident   
\end{lstlisting}
\label{fig:cfg-c}
\end{subfigure}%
\begin{subfigure}{.5\textwidth}
\begin{lstlisting}[escapechar=\@,basicstyle=\small\sf,numbers=left,breaklines=true]
  decl ::= @\textbf{typedef}@ . type-expr . id=rawident [@$\lnot$@ id @$\in$@ (!identifiers)]
             {types.add id} 
	  | ...
  typename ::= x = rawident [x @$\in$@ (!types)]{return x}
  type-exp ::= "int" | "bool" 
  expr ::=  ... | id=rawident {identifiers.add id ; return id}
\end{lstlisting}
\label{fig:csg-c}
\end{subfigure}
\caption{Context-free and context-sensitive grammars for C declarations.}
\label{fig:cgrammars}
\end{figure}

Unfortunately, ad-hoc lexer-hacks are both tedious and error prone.
Further, this convoluted integration of the lexing and parsing phases makes
it challenging to validate the correctness of the parser implementation.

A cleaner way to implement such a parser is to disambiguate {\it
  typenames} and {\it identifiers} when parsing by writing an actual
context-sensitive parser. One approach would be to define a shared {\it
  context} of two non-overlapping lists of {\sf types} and {\sf
  identifiers} and a stateful-parser using this context.  The modified
{\it context-sensitive} grammar is shown in right hand side, Figure~\ref{fig:cgrammars}.

The square brackets show context-sensitive
checks e.g. [$\lnot$ {\sf id} $\in$ {\sf (!identifier)}] checks that
the parsed {\sf rawident} token {\sf id} is not in the list of {\sf 
  identifiers}, while the braces show semantic actions associated with
parser reductions, e.g. \{{\sf typed.add id}\}, adds the token {\sf
  id} to {\sf types}, a list of identifiers seen thus far in the parse.

\begin{figure}
    \begin{subfigure}[t]{0.5\textwidth}
    \begin{lstlisting}[escapechar=\@,basicstyle=\small\sf,numbers=left,language=ML,showstringspaces=false]
let ids = ref []
let types = ref []
type decl = 
    Typedecl of {typeexp;string}
    | $\ldots$
type expression = 
    Address of expression 
    | Cast of string * expression
    | $\ldots$
    | Identifier of string 
@\scriptsize\texttt{\textcolor{gray}{
    \begin{tabbing}
  ex\=pression :\\ \label{line:typeexprspec}
  $\mathsf{PE^{stexc}}$\\
  \>\{$\forall$ h, \\
     \> ldisjoint (sel (h, ids), 
      sel (h, types)) = true)\} \\
  \>\quad            $\nu$ : expression result  \\
  \>\{$\forall$ h, $\nu$, h'.\=
    $\nu$ = Inl (v1) => \\
  \>\quad ldisjoint (sel (h', ids), 
       sel (h', types)) = true) \\
  \>\quad $\wedge$ $\nu$ = Inr (Err) =>\ 
   included(inp,h,h') = true \}
 \end{tabbing}}}@
 let expression = @\label{line:expressionparse}@ 
  $\M{do}$ char '(' 
    tn <- typename 
    char ')'
    e <- expression
    return Cast (tn, e))
  <|> $\ldots$
  <|> 
  $\M{do}$  @\label{line:beginidentifierexp}@ 
    id <- identifier
    let b = List.mem id !types @\label{line:memtypecheck}@ 
    if (!b) then     
      ids := id :: (!ids) @\label{line:addidentifiers}@
      return (Identifier id) 
    else 
        fail @\label{line:endidentifierexp}@
    \end{lstlisting}
    \end{subfigure}
\begin{subfigure}[t]{0.44\textwidth}
\begin{lstlisting}[escapechar=\@,basicstyle=\small\sf,numbers=right,breaklines=true,firstnumber=28,language=Haskell,showstringspaces=false]
 @\scriptsize\texttt{\textcolor{blue}{
  \begin{tabbing}
  typ\=edecl :\\ \label{line:typedeclspec}
  $\mathsf{PE^{stexc}}$\\
  \>\{$\forall$ h, \\
     \> ldisjoint (sel (h, ids), 
       sel (h, types)) = true) \}\\
  \>\quad            $\nu$ : tdecl result  \\
  \>\{$\forall$ h, $\nu$, h'.\=
    $\nu$ = Inl (v1) => \\
  \> \quad ldisjoint (sel (h', ids), 
       sel (h', types)) = true) \\
  \> \quad $\wedge$ $\nu$ = Inr (Err) => \
  included(inp,h,h') = true \}  
 \end{tabbing}}}@
 let typedecl = 
   $\M{do}$ 
     td <- keyword "typedef"
     te <- string "bool" <|> string "int" @\label{line:texp}@
     id <- indentifier @\label{line:identifierparse}@
     @\textcolor{gray}{(* incorrect-check: if (not (List.mem id !types)) then*)}@ @\label{line:tdeclcheckincorrect}@ 
     if (not (List.mem id !ids)) then @\label{line:tdeclcheck}@                                  
          types := id :: (!types) @\label{line:addtype}@
          return Tdecl {typeexp; id} 
     else 
          fail @\label{line:identifierparse'}@
@\scriptsize\texttt{\textcolor{gray}{
  \begin{tabbing}
  typ\=ename :\\ \label{line:typenamespec}
  $\mathsf{PE^{stexc}}$\\
  \>\{$\forall$ h. \\
     \> ldisjoint (sel (h, ids), 
      sel (h, types)) = true\} \\
  \>\quad $\nu$ : string result   \\
  \>\{$\forall$ h, $\nu$, h'.\=
    $\nu$ = Inl (v) => \\
   \> \quad mem (sel (h', types), v) = true \\
   \> \quad $\wedge$ $\nu$ = Inr (Err) => \
    included(inp,h,h') = true 
   \}
  \end{tabbing}}}@
let typename = @\label{line:typenameparse}@
    $\M{do}$ 
      x <- identifier 
      if (List.mem x !types) then @\label{line:checktypename}@
          return x
      else 
          fail    
    \end{lstlisting}
    \end{subfigure}
\caption{A simplified C-declaration parser written in \name.  Specifications
  in \textcolor{blue}{blue} are provided by the programmer; specifications
  in \textcolor{gray}{gray} are inferred by \name.}
\label{fig:c-decl}
\end{figure}
\begin{wrapfigure}{r}{.40\textwidth}
  \begin{center}
  \begin{lstlisting}[escapechar=\@,basicstyle=\small\sf,breaklines=true,language=ML, showspaces=false]
    type 'a t
    val eps : unit t
    val bot : 'a t
    val char : char -> char t 
    val (>>=) : 'a t -> (a -> 'b t) -> 'b t 
    val <|> : 'a t -> 'a t -> 'a t 
    val fix : ('b t -> 'b t) -> 'b t
    val return : 'a -> 'a t
  \end{lstlisting}
  \end{center}
  \caption{Signatures of primitive parser combinators supported by \name.}
  \label{fig:api}
  \vspace*{-.2in}
  \end{wrapfigure}

Given this grammar, we can use parser combinator
libraries~\cite{parsec,mparser} in our favorite language to implement
a parser for C language declarations. Unfortunately, although cleaner
than the using unwieldy lexer hacks, it is still not obvious how we
might verify that implementations actually satisfy the desired {\it
  disambiguation} property, i.e. {\it typenames} and {\it identifiers}
do not overlap. In the next section we provide an overview of \name\
that informally presents our solution to this problem.

\subsection{\name Surface Language}

An important design decision we make is to provide a surface syntax
and API very similar to conventional monadic parser combinator
libraries like \S{Parsec}~\cite{parsec} in Haskell or
\S{mParser}~\cite{mparser} in OCaml; the core API that \name\ provides
has the signature shown in Figure~\ref{fig:api}. The library defines
a number of primitive combinators: \S{eps} defines a parser for the
empty language, \S{bot} always fails, and \S{char \ c} defines a
parser for character \S{c}. Beyond these, the library also provides a
bind (\S{>>=}) combinator for monadically composing parsers, a choice
({<|>}) combinator to non-deterministically choose among two parsers,
and a \S{fix} combinator to implement recursive parsers.
The \S{return \ x} is a parser which always succeeds with a value \S{x}.
As we
demonstrate, these combinators are sufficient to derive a number of
other useful parsing actions such as \S{many}, \S{count}, etc. found
in these popular combinator libraries. From the parser writer's
perspective, \name\ programs can be expressed using these combinators
along with a basic collection of other non-parser expression forms
similar to those found in an ML core language, e.g., first-class
functions, \S{let} expressions, references, etc.
For instance a parser for {\sf option p}, which either parses an empty string or anything that 
{\sf p} parses can be written:
\begin{lstlisting}[escapechar=\@,basicstyle=\small\sf,language=ML]
  let option p = (eps >>= @$\lambda$@_. return None) <|> (p >>= @$\lambda$@ x. return Some x) \end{lstlisting}
We can also define more intricate parsers like {\it Kleene-star} and {\it Kleene-plus}:
\begin{lstlisting}[escapechar=\@,basicstyle=\small\sf,language=ML]
  let star p = fix ($\lambda$ p_star. eps  <|> p >>= @$\lambda$@ x. p_star >>= @$\lambda$@ xs . return (x :: xs) ) 
  let plus p = fix ($\lambda$ p_star. p  <|> p >>= @$\lambda$@ x. p_star >>= @$\lambda$@ xs . return (x :: xs) ) 
\end{lstlisting}

Figure~\ref{fig:c-decl} shows a \name
implementation that parses a valid C language \textsf{decl}.\footnote{For now, ignore the specifications given in \textcolor{gray}{gray} and \textcolor{blue}{blue}.}  The parser uses two mutable
lists to keep track of {\sf types} and {\sf identifiers}.
The structure is similar to the original data-dependent grammar, 
even though the program uses ML-style operators for assignment
and dereferencing. For ease of presentation, we have written the
program using {\it do-notation} as syntactic sugar for
\name's monadic bind combinator.

The {\sf typedecl} parser follows the grammar and parses the keyword
{\it typedef} using the {\sf keyword} parser (not shown).\footnote{\name, like
  other parser combinator libraries provides a library of parsers for
  parsing keywords, identifiers, natural numbers, strings, etc.} It 
uses a choice combinator (<|>) (line~\ref{line:texp}), which
has a semantics of a non-deterministic choice between  two
sub-parsers. The interesting case occurs while parsing an {\sf
  identifier} (lines~\ref{line:identifierparse} -
~\ref{line:identifierparse'}), in order to enforce disambiguation
between {\it typenames} and {\it identifiers}, the parser needs to
maintain an invariant that the two lists, {\sf types} for parsed {\it
  typenames} and {\sf ids} for parsed {\it identifiers} are always
{\it disjoint} or {\it non-overlapping}.

In order to maintain the non-overlapping list invariant, a parsed
identifier token (line~\ref{line:identifierparse}) can be a valid
typename only if it is not parsed earlier as an identifier
expression. i.e. it is not in the {\it ids} list. The parser performs
this check at (line~\ref{line:tdeclcheck}). If this check succeeds,
the list of typenames ({\it types}) is updated and a {\sf decl} is
returned, else the parsing fails.

The disambiguation decision is required during the parsing of an {\sf
  expression}. The expression parser defines multiple choices. The
parser for the {\it casting} expression parses a typename followed by
a recursive call to expression. The {\sf typename} parser in turn
(line~\ref{line:typenameparse}) parses an identifier token and checks
that the identifier is indeed a {\sf typename}
(line~\ref{line:checktypename}) and returns it, or fails.

The {\sf ids} list is updated during parsing an {\sf identifier} expression (line~\ref{line:beginidentifierexp}), here again to maintain disambiguation, before adding a string to the {\sf ids} list, its non-membership in the current {\sf types} list is checked (line~\ref{line:memtypecheck}).

Although the above parser program is easy to comprehend given how
closely it hews to the grammar definition, it is still nonetheless
non-trivial to verify that the parser actually satisfies the required
disambiguation safety property.  For example, an implementation
in which line~\ref{line:tdeclcheckincorrect} is replaced with the
commented expression above it would incorrectly check membership on
the wrong list.  We describe how \name\ facilitates verification of
this program in the following section.

\subsection{Specifying Data-dependent Parser Properties}
Intuitively, verifying the above-given parser for the absence of
overlap between the {\it typenames} and {\it identifiers} requires
establishing the following partial correctness property: if the
\textsf{types} and \textsf{identifiers} lists do not overlap when the
\textsf{typedecl} parser is invoked, and the parser terminates without
an error, then they must not overlap in the output state generated by
the parser.  Additionally, it is required that the parser consumes
some prefix of the input list. \name provides an expressive
specification language to specify properties such as these.

\name allows standard ML-style inductive type definitions that can be
refined with {\it qualifiers} similar to other refinement type
systems~\cite{liquidoriginal,VS+14,relref}.  For instance, we can
refine the type of a list of strings to only denote {\it non-empty}
lists as: \textsf{type nonempty = \{ $\nu$ : [string]\ |\ len ($\nu$)
  $>$ 0 \}}.  Here, $\nu$ is a special bound variable representing a
list and ({\sf len $\nu$ > 0}) is a {\it refinement} where {\sf len}
is a {\it qualifier}, a predicate available to the type system that
captures the length property of a list.


\paragraph{Specifying effectful safety properties}
Standard refinement type systems, however, are ill-suited to specify safety properties for
effectful computation of the kind expressible by parser combinators.
Our specification language, therefore, also provides a type for
effectful computations. We use a specification monad (called a \emph{Parsing
  Expression}) of the form
$\mathsf{PE^{\el}}$ \{ $\phi$ \} $\nu$ : $\tau$ \{ $\phi'$ \} that is
parameterized by the \emph{effect} of the computation $\el$ (e.g.,
\S{state}, \S{exc}, \S{nondet}, and their combinations like \S{stexc} for (both \S{state} and \S{exc}), \S{stnon} (for both \S{state} and \S{nondet}), etc.); and Hoare-style pre- and
post-conditions~\cite{htt,fstar,vcc}.  Here, $\phi$ and $\phi'$ are first-order
logical propositions over qualifiers applied to program variables and
variables in the type context.  The precondition $\phi$ is defined
over an abstract input heap {\sf h} while the postcondition $\phi'$ is
defined over input heap {\sf h}, output heap {\sf h'}, and the special
result variable $\nu$ that denotes the result of the computation.
Using this monad, we can specify a safety property for the {\sf
  typedecl} subparser as shown at line~\ref{line:typedeclspec} in
Figure~\ref{fig:c-decl}. 
The type should be understood as follows: The {\it effect} label \S{stexc} defines that the parser may have both \S{state} effect as it reads and updates the context; and \S{exc} effect as the parser may fail. 
The precondition defines a
property over a list of identifiers {\sf ids} and a list of typenames
{\sf types} in the input heap {\sf h} via the use of the built-in
qualifier {\sf sel} that defines a select operation on the
heap~\cite{mccarthy}; here, $\nu$ is bound to the result of the parse. \name
also allows user-defined qualifiers, like the qualifier {\sf\small
  lsdisjoint}. It establishes the {\it disjointness/non-overlapping}
property between two lists.  This qualifier is defined using the
following definition:
\begin{lstlisting}[escapechar=\@,basicstyle=\small\sf,breaklines=true]
   qualifier lsdisjoint [] l2 -> true 
                  | l1 [] -> true
                  | (x :: xs) l2 -> member (x, l2) = false @$\wedge$@ lsdisjoint (xs, l2)
                  | l1 (y :: ys) -> member (y, l1) = false @$\wedge$@ lsdisjoint (l1, ys)
\end{lstlisting}
This definition also uses another qualifier for list membership called
{\sf member}.  \name automatically translates these user-defined
qualifiers to axioms, logical sentences whose validity is assumed by
the underlying theorem prover during verification.  For instance,
given the above qualifier, \name generates axioms like:
\begin{lstlisting}[escapechar=\@,basicstyle=\small\sf,breaklines=true]
   Axiom1: $\forall$ l1, l2 : $\alpha$ list. (empty(l1) $\vee$ empty (l2)) => lsdisjoint (l1, l2) = true
   Axiom2: $\forall$ xs, l2: $\alpha$ list, x : $\alpha$. lsdisjoint (xs, l2) = true @$\wedge$@ member (x, l2) = false => lsdisjoint ((x::xs), l2) = true
   Axiom3: $\forall$ l1, l2: $\alpha$ lsdisjoint (l1, l2) <=> lsdisjoint (l2, l1)
\end{lstlisting}
The specification (at line~\ref{line:typedeclspec}) also uses another
qualifier, {\sf included(inp,h,h')}, which captures the monotonic consumption property of the input list \S{inp}. The qualifier is true when the
remainder \textsf{inp} after parsing in \textsf{h'} is a suffix of the original \S{inp} list in \textsf{h}.

The types for other parsers in the figure can be specified as shown at
lines~\ref{line:typeexprspec},~\ref{line:typenamespec}, etc.; these
types shown in \textcolor{gray}{gray} are automatically inferred by
\name's type inference algorithm. For example, the type for the {\sf
  typename} parser (line ~\ref{line:typenamespec}) returns an optional
string (\textsf{result} is a special option type) and records that
when parsing is successful, the returned string is added to the {\sf
  types} list, and when unsuccessful, the input is still monotonically
consumed.

\paragraph{Verification using \name}
Note that the pre-condition in the specification ({\sf\small
  lsdisjoint (Id, Ty) = true)}) and the type ascribed to the
membership checks in the implementation (line ~\ref{line:tdeclcheck})
are sufficient to conclude that the addition of a typename to the \S{types} list
(line ~\ref{line:addtype}) maintains the {\sf\small lsdisjoint}
invariant as required by the postcondition.

In contrast, an erroneous implementation that omits the membership
check or replaces the check at line ~\ref{line:tdeclcheckincorrect}
with the commented line above it will cause type-checking to fail.
The program will be flagged ill-typed by \name.  For this example,
Morpheus generated {\sf 21} verification conditions (VCs) for the
control-path representing a successful parse and generated {\sf 5} VCs
for the failing branch.  We were able to discharge these VCs to
the SMT solver Z3~\cite{z3}, which took {\sf 6.78} seconds to verify the former and
{\sf 1.90} seconds to verify the latter.

\section{\name\ Syntax and Semantics}
\label{sec:syntax}

\subsection{\name\ Syntax}

Figure ~\ref{fig:syntax} defines the syntax of $\lambda_{sp}$, a core
calculus for \name\ programs.  The language is a call-by-value
polymorphic lambda-calculus with effects, extended with primitive
expressions for common parser combinators and a refinement type-based
specification language. A $\lambda_{sp}$ value is either a constant
drawn from a set of base types (\textsf{int}, \textsf{bool}, etc.), as
well as a special \textsf{Err} value of type {\sf exception}, an
abstraction, or a constructor application. Variables bound to
updateable locations (\rp) are distinguished from variables introduced
via function binding (\xp).  A $\lambda_{sp}$ expression \eip\ is
either a value, an application of a function or type abstraction,
operations to dereference and assign to top-level locations (see
below), polymorphic \B{let} expressions, reference binding
expressions, a \textbf{match} expression to pattern-match over type
constructors, a \textbf{return} expression that lifts a value to an
effect, and various parser primitive expressions that define parsers
for the empty language (\textsf{eps}), a character (\textsf{char})
parser, and $\bot$, a parser that always fails. Additionally, the
language provides combinators to monadically compose parsers (\S{>>=}),
to implement parsers defined in terms of a non-deterministic choice of
its constituents (\S{<|>}), and to express parsers that have recursive
($\mu\,(\textsf{x} : \tau). p$) structure.

We restrict how effects manifest by requiring reference creation to
occur only within \B{let} expressions and not in any other expression
context. Moreover, the variables bound to locations so created
($\ell$) can only be dereferenced or assigned to and cannot be
supplied as arguments to abstractions or returned as results since
they are not treated as ordinary expressions. This stratification,
while arguably restrictive in a general application context, is
consistent with how parser applications, such as our introductory
example are typically written and, as we demonstrate below, do not
hinder our ability to write real-world data-dependent parser
implementations. Enforcing these restrictions, however, provides a
straightforward mechanism to prevent aliasing of effectful components
during evaluation, significantly easing the development of an
automated verification pathway in the presence of parser
combinator-induced computational effects.
\begin{figure}[t]
\begin{tabbing}
  \textbf{Expression\ Language}\\[-4pt]
\end{tabbing}  
\[  
\begin{array}{lclll} 
\cp, \B{unit}, {\sf Err}\ & \in & \mathit{Constants} \\ 
\xp\ & \in & \mathit{Vars}\\
\textsf{inp}, \rp\ & \in & \mathit{RefVars}\\
\vp\ & \in & \mathit{Value} & ::= & \cp\,  \mid\, \lambda\, (\xp:\tau).\,\eip \mid\ \Lambda\, (\alpha).\, \eip\ \mid\
      \mathsf{D_i}\, \overline{t_k}\, \overline{\vp_{j}}\\
\eip\ & \in & \mathit{Exp} & ::= &  \vp\, \mid\ \xp\, \mid\ p\, \mid\ \eip\ \xp\, \mid\ \eip\, [\mathsf{t}]
               \mid\  \mathsf{deref}\, \rp\, \mid\ \rp\ \mathsf{:=}\, \eip \\
&&&&         \mid\ \mathbf{let}\, \xp\ = \vp\, \mathbf{in}\ \eip\, 
             \mid\ \mathbf{let}\, \rp\ = \S{ref}\, \eip\, \mathbf{in}\ \eip\, \\
&&&&         \mid\ \mathbf{match}\, \vp\ \mathbf{with}\ \mathsf{D_i}\, 
               \overline{\alpha}\, \overline{\xp_{j}} \rightarrow \eip\, \mid\ \mathbf{return}\ \eip  \\
\mathit{p} & \in & \mathit{Parsers} & ::= &     
             \mid\ \mathsf{eps}\, \mid\, \bot\, \mid\, \mathsf{char}\, \eip\,
              \mid\ (\mu\, (\xp\, :\, \tau).\, p) 
              \mid\ p\, >>=\, \eip
              \mid\ p\, <\mid>\, p 
\end{array}\]\\[1pt]
\begin{tabbing}
\textbf{Specification Language}\\[-4pt]
\end{tabbing}
\[
\begin{array}{lclll} 
\alpha & \in & \mathit{Type Variables}\\
\mathsf{TN} & \in & \mathit{User\ Defined\ Types } & ::= &  \alpha\, \mathsf{list}, \alpha\ \mathsf{tree}, \ldots\\
\mathsf{t}  & \in & \mathit{Base\ Types} & ::= & \alpha\, \mid\ \mathsf{int}\, \mid\ \mathsf{bool}\, \mid\ \S{unit}\, \mid\ 
                \mathsf{heap}\, \mid\ \mathsf{TN}\, \mid\ \S{t}\ \S{result}\ \mid\ \S{t}\, \S{ref} \mid\ \mathsf{exc} \\
\tau & \in & \mathit{Type} & ::= & \{\nu\, :\, \mathsf{t} \mid\ \phi\} \mid\ (\xp\, :\, \tau) \rightarrow \tau 
                \mid\ \mathsf{PE^{\el}} \{ \phi_1 \} \nu\, :\, \mathsf{t}\, \{ \phi_2 \}  \\
\el & \in & \mathit{Effect\ Labels}  & ::= & \mathsf{pure}\, \mid\ \mathsf{state}\, \mid\ \mathsf{exc}\, \mid\ \mathsf{nondet} \mid \ldots\\
\sigma & \in & \mathit{Type\ Scheme} & := & \tau\, \mid\ \forall \alpha.\, \tau \\
\qual  & \in & \mathit{Qualifiers}   & := &  \mathit{QualifierName}(\overline{x_i}) \\
\phi,P & \in & \mathit{Propositions} & ::=& \mathsf{true}\, \mid\ \mathsf{false} \mid\ \qual \mid\ \qual_1 = \qual_2 \\
&&&&            \mid\ \neg\ \phi\ \mid\ \phi\ \wedge\ \phi\  \mid\ \phi\ \lor \phi  \mid \phi \Rightarrow\ \phi 
                \mid\ \forall (\xp\, :\, \mathsf{t}). \phi\ \\
\Gamma\ & \in & \mathit{Type\ Context} & ::= & \varnothing\, \mid\ \Gamma,\ \xp\, :\, \sigma\ \mid\ \Gamma,\ \rp\, :\, \tau\ \S{ref}\, \mid\ \Gamma, \phi \\
\Sigma  & \in & \mathit{Constructors} & ::= & \varnothing\, \mid\ \Sigma,\ 
                \mathsf{D_i}\ \overline{\alpha_k}\ \overline{\xp_j : \tau_{j}}\, \rightarrow\ \tau\\
\hline
\end{array}\]
\caption{$\lambda_{sp}$ Expressions and Types } 
\label{fig:syntax} 
\end{figure}

\subsection{Semantics}
Figure ~\ref{fig:semantics-parser} presents a big-step operational
semantics for $\lambda_{sp}$ parser expressions; the semantics of
other terms in the language is standard.  The semantics is defined via
an evaluation relation ($\Downarrow$) that is of the form
$(\mathcal{H}; \mathsf{\eip}) \Downarrow (\mathcal{H'}; \vp)$.
The relation defines how a \name\ expression {\sf \eip} evaluates with
respect to a heap $\mathcal{H}$, a store of locations to base-type
values, to yield a value $\vp$, which can be a normal value or an
exceptional one, the latter represented by the exception constant {\sf
  Err}, and a new heap $\mathcal{H'}$.

The empty string parser (rule {\sc P-eps}) always succeeds, returning
a value of type \textbf{unit}, without changing the heap. A ``bottom'' ($\bot$) parser on the other hand always fails, producing an exception value, also without
changing the heap. If the argument \eip\ to a character parser {\sf
  char} yields value (a char `c'), and `c' is the head of the input string
(denoted by \S{inp}) being parsed, the parse succeeds (rule {\sc
  P-char-true}), consuming the input and returning `c', otherwise,
the parse fails, with the input not consumed and the distinguished
\S{Err} value being returned (rule {\sc P-char-false}). The fixpoint
parser $\mu\, x.p$ ({\sc P-fix}) allows the construction of recursive parser
expressions. 
The monadic bind parser primitive (rule {\sc
  P-bind-success}) binds the result of evaluating its parser
expression to the argument of the abstraction denoted by its second
argument, returning the result of the evaluating the abstraction's
body ({\sc P-bind-success}); the {\sc P-bind-err} rule deals with the
case when the first expression fails. Evaluation of ``choice''
expressions, defined by rules {\sc P-choice-l} and {\sc P-choice-r},
 introduce an unbiased choice semantics over two parsers allowing non-deterministic choices in parsers. 

\begin{figure}[h]
\begin{flushleft}
\fbox{
        $(\mathcal{H}; \mathsf{\eip}) \Downarrow (\mathcal{H'}; \vp) $ 
           
} 
\bigskip
\end{flushleft}
\begin{minipage}{0.6\textwidth}
  \inference[{\sc P-eps}]{} {(\mathcal{H}; {\sf eps}) \Downarrow (\mathcal{H}; {\sf ()})} 
\end{minipage}
\hfill
\begin{minipage}{0.6\textwidth}
\inference[{\sc P}-$\bot$]{ }{(\mathcal{H}; {\sf \bot}) \Downarrow (\mathcal{H}; {\sf Err})}  
\end{minipage}
\hfill
\begin{minipage}{0.6\textwidth}
\inference[{\sc P-fix}]{ (\mathcal{H}; [\mu \mathsf{x:\sigma}. p/\S{x}]p)) \Downarrow (\mathcal{H'}; \vp)  }
                                  {(\mathcal{H}; \mu \mathsf{x:\sigma}. p) \Downarrow (\mathcal{H'}; \vp) } 
\end{minipage}\\[10pt]
\begin{minipage}{0.4\textwidth}
  \inference[{\sc P-char-true}]{(\mathcal{H}; \eip) \Downarrow (\mathcal{H}; \emph{`c'}) & {\sf \mathcal{H}(inp) = (\emph{`c'}::s)}\\
                                 \mathcal{H'} = \mathcal{H}[{\sf inp} \mapsto s]}
                             {(\mathcal{H}; {\sf char} \ \eip) \Downarrow (\mathcal{H'}; \emph{`c'})}
\end{minipage}\\[10pt]
\begin{minipage}{0.4\textwidth}
  \inference[{\sc P-char-false}]{(\mathcal{H}; \eip) \Downarrow (\mathcal{H}; \emph{`c'}) & {\sf \mathcal{H} (inp)} \neq  (\emph{`c'} :: s) \\
                                  \mathcal{H'} = \mathcal{H}[{\sf inp} \mapsto {\sf inp}]}
                              {(\mathcal{H}; {\sf char} \ \eip) \Downarrow (\mathcal{H'}; {\sf Err})) } 
\end{minipage}\\[10pt]
\begin{minipage}{0.4\textwidth}
  \inference[{\sc P-bind-success}]{ 
    (\mathcal{H};  p)  \Downarrow (\mathcal{H'}; \vp_1) \quad
    (\mathcal{H'};  \eip) \Downarrow (\mathcal{H'}; (\lambda\, x\, :\, \tau.\, \eip'))\\ 
    (\mathcal{H'}; [\ \vp_1 / x]\eip') \Downarrow  (\mathcal{H''}; \vp_2)}
    {(\mathcal{H}; p \textnormal{>>=} \eip) \Downarrow (\mathcal{H''}; \vp_2)} 
\end{minipage}\\[10pt]
\begin{minipage}{0.6\textwidth}
  \inference[{\sc P-bind-err}]{
    (\mathcal{H}; p) \Downarrow (\mathcal{H'}; {\sf Err})}
  {(\mathcal{H}; p \textnormal{>>=} \eip) \Downarrow (\mathcal{H'}; {\sf Err})) } 
\end{minipage}\\[10pt]
\begin{minipage}{0.5\textwidth}
\inference[{\sc P-choice-l}]{ (\mathcal{H}; p_1) \Downarrow (\mathcal{H'}; {\sf \vp_1}) }
                                  {(\mathcal{H}; (p_1 <\mid> p_2)) \Downarrow 
(\mathcal{H'}; {\sf \vp_1})) } 
\end{minipage}
\hfill
\begin{minipage}{0.4\textwidth}
\inference[{\sc P-choice-r}]{ (\mathcal{H}; p_2) \Downarrow (\mathcal{H''}; {\sf \vp_2}) }
                                  {(\mathcal{H}; (p_1 <\mid> p_2)) \Downarrow 
(\mathcal{H''}; {\sf \vp_2})) } 
\end{minipage}
\bigskip
\caption{Evaluation rules for $\lambda_{sp}$ parser expressions}
\label{fig:semantics-parser}
\end{figure}

\section{Typing $\lambda_{sp}$ Expressions}
\label{sec:typing}

\subsection{Specification Language}
\label{sec:spec}  
The syntax of \name's type system is shown in the bottom of
Figure~\ref{fig:syntax} and permits the expression of {\it base types}
such as integers, booleans, strings, etc., as well as a special
\S{heap} type to denote the type of abstract heap variables like \S{h,
  h'} found in the specifications described below. There are additionally
user-defined datatypes \textsf{TN} (\textsf{list}, \textsf{tree},
etc.), a special sum type ({\sf t} result) to define two options of a
successful and exceptional result respectively, and a special
exception type.

More interestingly, base types can be refined with \emph{propositions}
to yield monomorphic refinement types.  Such
types~\cite{fstar,liquidoriginal,VS+14} are either {\it base refinement
  types}, refining a base typed term with a refinement; {\it
  dependent function types}, in which arguments and return values of
functions can be associated with types that are refined by
propositions; or a {\it computation type} specifying a type for an
effectful computation.  Effectful computations are refined using an
effect specification monad
\[\mathsf{PE}^{\el}\ \{\forall\,\mathsf{h}.\phi_1\}\, \nu : \mathsf{t}\,
\{\forall\, \mathsf{h},\nu,\mathsf{h'}.\phi_2\}\] that encapsulates a
base type \textsf{t}, parameterized by an effect label $\el$, with
Hoare-style pre- ($\{\forall\,\mathsf{h}.\phi_1\}$) and post-
($\{\forall\,\mathsf{h},\nu,\mathsf{h'}.\phi_2\}$) conditions.  This
type captures the behavior of a computation that (a) when executed in
a pre-state with input heap \S{h} satisfies proposition $\phi_1$; (b)
upon termination, returns a value denoted by $\nu$ of base type \S{t}
along with output heap \S{h'}; (c) satisifies a postcondition $\phi_2$
that relates \S{h}, $\nu$, and \S{h'}; and (d) whose effect is
over-approximated by effect label
$\el$~\cite{param-monad-effect,wadler1}. An effect label $\el$ is
either (i) a \S{pure} effect that records an effect-free computation;
(i) a \S{state} effect that signifies a stateful computation over the
program heap; (ii) an exception effect \S{exc} that denotes a
computation that might trigger an exception; (iii) a \S{nondet} effect
that records a computation that may have non-deterministic behavior;
or (iv) a \emph{join} over these effects that reflect composite
effectful actions. The need for the last is due to the fact that
effectful computations are often defined in terms of a composition
of effects, e.g. a parser oftentimes will define a computation that
has a state effect along with a possible exception effect. To capture
these composite effects, base effects can be joined to build a finite
lattice that reflects the behavior of computations which perform
multiple effectful actions, as we describe below.

Propositions ($\phi$) are first-order predicate logic formulae over
base-typed variables.  Propositions also include a set of qualifiers
which are applications of user-defined uninterpreted function symbols
such as \S{mem, size} etc. used to encode properties of program
objects, \textsf{sel} used to model accesses to the heap, and
\textsf{dom} used to model membership of a location in the heap, etc.
Proposition validity is checked by embedding them into a decidable
logic that supports equality of uninterpreted functions and linear
arithmetic (EUFLIA).

A type scheme ($\sigma$) is either a monotype ($\tau$) or a
universally quantified polymorphic type over type variables expressed
in prenex-normal form ($\forall\ \alpha.\sigma$).  A
\name\ specification is given as a type scheme.

There are two environments maintained by the \name\ type-checker: (1)
an environment $\Gamma$ records the type of variables, which can include
variables introduced by function abstraction as well as bindings to
references introduced by \S{let} expressions, along with a set of
propositions relevant to a specific context, and (2) an environment
$\Sigma$ maps datatype constructors to their signatures.  Our typing
judgments are defined with respect to a typing environment\\
\vspace*{-.1in}
\begin{center}\small
        $\Gamma$ ::= . $\mid$ $\Gamma$, \S{x} : $\sigma$ $\mid$ $\Gamma$, \rp : $\tau$\, \S{ref} \\
\end{center}
\noindent that is either empty, or contains a list of bindings of
variables to either type schemes or references.  The rules have two
judgment forms: ($\Gamma \vdash$ $\eip$ : $\sigma$) gives a type
for a \name\ expression $\eip$ in $\Gamma$; and ($\Gamma \vdash$
$\sigma_1$ <: $\sigma_2$) defines a dependent subtyping rule under
$\Gamma$.

Since our type expressions contain refinements, we generalize the
usual notion of type substitution to reflect substitution within
refined types:
\[
\begin{array}{rl}
  [x_a/x]\{ \nu : \S{t} | \phi \} &= \{ \nu : \S{t} | [x_a/x]\phi \}\\
  [x_a/x](y : \tau) \rightarrow \tau' &= (y : [x_a/x]\tau) \rightarrow [x_a/x]\tau', y \neq x\\
  [x_a/x]\mathsf{PE}^{\el} \{ \phi_1 \} \{ \nu : \S{t} \} \{ \phi_2 \} &=
      \mathsf{PE}^{\el} \{ [x_a/x]\phi_1 \} \{ \nu : \S{t} \} \{ [x_a/x]\phi_2 \}
\end{array}
\]

\subsection{Typing Base Expressions}

Figure~\ref{fig:typing-base} presents type rules for non-parser
expressions. The type rules for non-reference variables, functions, and
type abstractions ({\sc T-typ-fun}) are standard. The syntax for function application
restricts its argument to be a variable,  allowing us to record the
argument's (intermediate) effects in the typing environment when
typing the application as a whole.

\begin{figure*}[t]
\begin{flushleft}
\bigskip
{\bf Base Expression Typing}\quad \fbox{\small $\Gamma \vdash$ $\eip$ : $\sigma$}
\end{flushleft}
\begin{minipage}{0.1\textwidth}
\begin{center}
  \inference[{\sc T-var}]{\Gamma (\S{x}) = \sigma}
                         {\Gamma \vdash \S{x} : \sigma}
\end{center}
\end{minipage}
\begin{minipage}{0.3\textwidth}
\begin{center}
\inference[{\sc T-fun}]{\Gamma, \S{x} : \tau_1 \vdash \eip : \tau_2}
		            {\Gamma \vdash \lambda (\S{x : \tau_1}).\eip\, : \tau_1 \rightarrow \tau_2}
\end{center}  
\end{minipage}
\begin{minipage}{0.3\textwidth}
\begin{center}
\inference[{\sc T-typApp}]{\Gamma \vdash \Lambda \alpha. \eip : \forall \alpha. \sigma}
		    {\Gamma \vdash {\Lambda \alpha. \eip}[\S{t}] : [\S{t}/\alpha]\sigma}
\end{center}  
\end{minipage}\\[5pt]
\bigskip
\begin{minipage}{0.4\textwidth}
\inference[{\sc T-App}]{\Gamma \vdash \eip_f : (x\, :\, \{ \nu : \S{t} \mid \phi_x\}) \rightarrow \S{PE}^{\el} \{\phi\}\ \nu\ : \S{t}\ \{\phi'\}\
                                   & \Gamma \vdash \xp_a : \{ \nu : \S{t} \mid \phi_x\} }
                   {\Gamma \vdash \eip_f\ \xp_a\, : [\xp_a/x]\S{PE}^{\el} \{\phi\}\ \nu\ : \S{t}\ \{\phi'\}}
\end{minipage}  
\begin{minipage}{0.4\textwidth}
\inference[{\sc T-typFun}]{\Gamma \vdash \S{e} : \sigma & \alpha \notin FV (\Gamma)}
		    {\Gamma \vdash \Lambda \alpha. \S{e} : \forall \alpha. \sigma}
\end{minipage}\\[5pt]
\bigskip
\begin{minipage}{0.4\textwidth}
\begin{center}  
\inference[{\sc T-return}]{\Gamma \vdash \eip\ : \{ \nu : \S{t} \mid \phi \} }
	                    {\Gamma \vdash \mathbf{return}\, \eip :
                                \mathsf{PE}^{\S{pure}} \{ \forall \S{h}. \S{true} \}\, \nu :\, \S{t} \
                                \{ \forall \S{h}, \nu, \S{h'}. \S{h' = h} \wedge \phi \}}
\end{center}
\end{minipage}\\[5pt]
\bigskip
\begin{minipage}{0.4\textwidth}
\begin{center}
\inference[{\sc T-let}]{\Gamma \vdash \vp\ : \forall\alpha.\sigma & \Gamma, \xp\ : \forall\alpha.\sigma \vdash \eip_2 : \sigma' }
	                 {\Gamma \vdash \mathbf{let}\, \xp\, =\, \vp\, \mathbf{in}\, \eip_2\, :\, \sigma'}
\end{center}
\end{minipage}\\[5pt]
\bigskip
\begin{minipage}{.5\textwidth}
\begin{center}
\inference[{\sc T-capp}]{\Sigma (D_i) = \forall \overline{\alpha_k}. \overline{\S{x_j : \tau_j}} 
\rightarrow \tau & \forall i, j. \Gamma \vdash \S{v}_j : [\overline{\S{t}_k / \alpha_k }][\overline{\S{v_j / x_j}}]\tau_j}
		{\Gamma \vdash D_i\ \overline{\S{t}_k} \overline{\S{v}_j} : [\overline{\S{t / \alpha }}] [\overline{\S{v_j / x_j}}]\tau}
\end{center}
\end{minipage}\\[10pt]
\begin{minipage}{0.5\textwidth}
\begin{center}
\inference[{\sc T-match}]{\Sigma (D_i) = \forall \overline{\alpha_k}. \overline{\S{x_j : \tau_j}} \rightarrow \tau_0 \\  
		      \Gamma \vdash \vp : \tau_0 & 
                    \Gamma_{i} = \Gamma, \overline{\alpha_k}, \overline{\mathsf{x}_{j} :\tau_j} \\ 
                      \Gamma_{i} \vdash (\mathsf{D_i} \
                           \overline{\mathsf{\alpha}_{k}} \overline{\mathsf{x}_{j}}) : \tau_0 & \Gamma_{i} \vdash e_i : \S{PE}^{\S{\el}} \{\phi_i\}\ \nu\ : \S{t}\ \{\phi_{i'}\}\ }
                     {\Gamma \vdash {\bf match} \ \vp \ {\bf with} \ \mathsf{D_i}\ 
                           \overline{\mathsf{\alpha}_{k}} \overline{\mathsf{x}_{j}} \rightarrow e_i : \S{PE}^{\S{\el}} \{ \forall\ \S{h}. 
                            \bigwedge_{i}^{} (\vp\ = \mathsf{D_i} \ \overline{\mathsf{\alpha}_{k}} \overline{\mathsf{x}_{j}}) =>  
                           \phi_i\}\ \nu\ : \S{t}\ \{ 
                            \forall\ \S{h}, \nu', \S{h'}.                           
                           \bigvee_{i}^{} \phi_{i'}\}\ }
\end{center}
\end{minipage}\\[5pt]
\bigskip
\begin{minipage}{0.5\textwidth}
\begin{center}
\inference[{\sc T-deref}]{\Gamma \vdash \rp\ : \S{PE}^{\S{state}} \{\phi_1\}\ \nu\ : \S{t\ ref}\ \{\phi_2\}\                        }
		         {\Gamma \vdash \S{deref} \ \rp : \S{PE}^{\S{state}} \{ \forall\ \S{h}. \S{dom (h,\rp)} \}\ \nu'\ : \S{t}\ 
                                                          \{ \forall\ \S{h}, \nu', \S{h'}. \S{sel(h,\rp)}\ = \nu' \wedge \S{h = h'} \}}
\end{center}
\end{minipage}\\[5pt]
\bigskip
\begin{minipage}{0.5\textwidth}
\begin{center}
  \inference[{\sc T-assign}]{\Gamma \vdash\ e\ :\ \{ \nu\ :\  \S{t} \mid\ \phi \}}
		        {\Gamma \vdash \rp\ \mathsf{:=}\ \eip\  : \S{PE}^{\S{state}} \{ \forall \S{h}.\S{dom(h,\rp)} \}\ \nu'\ :\ \S{t}\ \{ \forall\ \S{h},\nu',\S{h'}.\S{sel(h',\rp)}\ =\ \nu' \wedge\ \phi(\nu') \}}
\end{center}
\end{minipage}\\[5pt]
\bigskip
\begin{minipage}[t]{0.5\textwidth}
\begin{center}
  \inference[{\sc T-ref}]
            {
                \Gamma\ \vdash\ \vp\ :\ \{\ \nu\ :\ \S{t} \mid\ \phi\ \}\\
                \Gamma, \rp\ : \S{PE}^{\S{state}} \{ \forall\ \S{h}. \neg\ \S{dom(h,\rp)} \}\ 
                                    \nu'\ : \S{t\ ref} \,
                                     \{ \forall\ \S{h,\nu',h'}.  \S{sel(h',\rp)}\ =\ \vp\ \wedge
                                    \phi(\vp) \wedge\ \S{dom(h',\rp)} \} \vdash\ \\ 
                                    \S{e}_b\ :  \S{PE}^{\S{\el}} \{\S{dom(h,\rp)}\}\ \nu''\ : \S{t}\ \{\phi_b'\}\ }
            {\begin{array}{cc}
                   \Gamma\ , \S{h_i} : \S{heap} \vdash\ \B{let}\ \rp\ = \S{ref}\ \vp\ \B{in}\ \S{e}_b\ :\\
                   \S{PE}^{\S{\el} \sqcup \mathsf{state}} \{\forall\ \S{h}. \neg\ \S{dom(h,\rp)}\}\,
                   \nu''\ : \S{t}\,
                   \{ \forall\ \S{h, \nu'', h}. \S{dom(h_i,\rp)} \wedge \S{sel(h_i,\rp)}\ =\ \vp\ \wedge \phi(\vp) \wedge\ \phi_b'\}\             
            \end{array}
          }
\end{center}
\end{minipage}\\[5pt]
\bigskip
\caption{Typing Semantics for \name\ Base Expressions}
\label{fig:typing-base}
\end{figure*}

The type rule for the return expression ({\sc T-return}) lifts its
non-effectful expression argument \eip\ to have a computation effect
with label \S{pure}, thereby allowing \eip's value to be used in
contexts where computational effects are required; a particularly
important example of such contexts are bind expressions used to
compose the effects of constituent parsers.

In the constructor application rule ({\sc T-capp}), the expression's
type reflects the instantiation of the type and term variables in the
constructor's type with actual types and terms. A match expression is
typed (rule {\sc T-match}) by typing each of the alternatives in a
corresponding extended environment and returning a {\it unified type}.
The pre-condition of the {\it unified} type is a conjunction of the
pre-conditions for each alternative, while the post-condition
over-approximates the behavior for each alternative by creating a
disjunction of each of the possible alternative's post-conditions.
Location manipulating expressions ({\sc T-deref} and {\sc T-assign})
use qualifiers \S{sel} and \S{dom} to define constraints that reflect
state changes on the underlying heap. The argument \rp\ of a
dereferencing expression (rule {\sc T-deref}) is associated with a
computation type over a \S{t ref} base type. Its pre-condition requires
\rp\ to be in the domain of the input heap, and its post-condition
establishes that \rp's contents is the value returned by the
expression and that the heap state does not change. The assignment
rule ({\sc T-assign}) assigns the contents of a top-level reference
$\rp$ to the non-effectful value yielded by evaluating expression
$e$. The pre-condition of its computation effect type requires that
\rp\ is in the domain of the input heap and that \rp's contents in
the output heap satisfies the refinement ($\phi$) associated with its
r-value. Finally, rule {\sc T-ref} types a \B{let} expression that
introduces a reference initialized to a value \vp. The body is typed
in an environment in which \rp\ is given a computational effect type.
The pre-condition of this type requires that the input heap, i.e., the
heap extant at the point when the binding of \rp\ to
\S{ref}\ \vp\ occurs, not include \rp\ in its domain; its
postcondition constrains \rp's contents to be some value $\nu'$ that
satisfies the refinement $\phi$ associated with $\vp$, its
initialization expression. The body of the \B{let} expression is then
typed in this augmented type environment.

\subsection{Typing Parser Expressions}
\label{sec:parser-typing}

\begin{figure*}[htbp]
\begin{flushleft}
\bigskip
{\bf Parser Expression Typing}\quad \fbox{\small $\Gamma \vdash$ $\eip$ : $\sigma$
           
}
\end{flushleft}
\bigskip
\begin{minipage}[t]{.7\textwidth}
\inference[{\sc T-sub}]{\Gamma \vdash e : \sigma_1 & \Gamma \vdash \sigma_1 <: \sigma_2}
		    {\Gamma \vdash e : \sigma_2}
\end{minipage}\\[10pt]
\begin{minipage}[t]{.7\textwidth}
\inference[{\sc T-p-eps}]{}
		    {\Gamma \vdash \S{eps} : \S{PE^{pure}}\, \{ \forall \S{h}.\, \S{true} \}\, \nu : \S{unit} 
\ \{ \forall \S{h}, \nu, \S{h'}. \S{h' = h} \}}
\end{minipage}\\[10pt]
\begin{minipage}[t]{.7\textwidth}
\inference[{\sc T-p-bot}]{}
		    {{\Gamma \vdash \S{\bot} : \S{PE^{exc}}\, \{ \forall \S{h}.\, \S{true} \}\, \nu  :  \S{exc}
\ \{ \forall \S{h}, \nu, \S{h'}. \S{h' = h} \wedge \nu = \S{Err} \}} }
\end{minipage}\\[10pt]
\begin{minipage}[t]{\textwidth}
  \inference[{\sc T-p-char}]{
    \Gamma \vdash \eip : \{ \nu' :\S{char} \mid \nu' = \lq \S{c} \rq \}\\
    \phi_2 = \forall \S{h, \nu, h'}. \forall \S{x}, \S{y.} \\ 
    (\S{Inl(x) = \nu} \implies \S{x} = \lq \S{c} \rq \wedge \S{upd(h', h, inp, tail (inp))}) \wedge \\
    (\S{Inr(y) = \nu} \implies \S{y = Err}  \wedge  \S{sel(h,inp)} = \S{sel(h',inp)})}
		    {\begin{array}{@{}c@{}}
		        \Gamma \vdash \S{char}\ \eip: \S{PE}^{\S{state}\, \sqcup\, \S{exc}} 
		      \{ \forall \S{h}. \S{true} \}\, \nu\, :\, \mathsf{char\ result}\, \{ \phi_2 \} \\ 
		    \end{array}
		    }		    
\end{minipage}\\[10pt]
\begin{minipage}{0.4\textwidth}
  \inference[{\sc T-p-choice}]{\Gamma \vdash p_1\, :\,
                                  \mathsf{PE}^{\el}\, \{ \phi_1 \}\, \nu_1\, :\, \tau\, \{ \phi_1' \} &  
 		               \Gamma \vdash p_2\, :\,
                                  \mathsf{PE}^{\el}\, \{ \phi_2 \}\, \nu_2\, :\, \tau\, \{ \phi_2' \} \\ 
				 } 
 			      { \begin{array}{@{}c@{}}  
				   \Gamma \vdash (p_1 \textnormal{<|>}  p_2) : 
				   \mathsf{PE}^{\el\, \sqcup\, \mathsf{nondet}} \ \{ (\phi1 \wedge \phi_2) \}
                                        \, \nu : \tau\,
                                        \{ (\phi_1'  \lor \phi_2') \} 
				     \end{array}			  
				 } 
\end{minipage}\\[10pt]
\begin{minipage}[t]{.3\textwidth}
\inference[{\sc T-p-fix}]{\Gamma, \mathsf{x}\, :\, (\mathsf{PE}^{\el}\, \{ \phi \}\, \nu\, :\, \S{t}\, \{ \phi' \})\, \vdash\ p\, : \mathsf{PE}^{\el}\, \{ \phi \}\, \nu\, :\, \S{t}\, \{ \phi' \} & {\sf x} \notin FV(\phi, \phi')}
		    {\Gamma \vdash \mu\, \mathsf{x}\, :\, (\mathsf{PE}^{\el}\, \{ \phi \}\, \nu\, :\, \S{t}\, \{ \phi' \}).\, p\, :\, \mathsf{PE}^{\el}\, \{ \phi \}\, \nu\, :\, \S{t}\, \{ \phi' \} }
\end{minipage}\\[10pt]
\hfill
\begin{minipage}[t]{.5\textwidth}
  \inference[{\sc T-p-bind}]{
      \Gamma\ \vdash\ p\, :\, \mathsf{PE}^{\el} \ \{ \phi_1 \}\, \nu\, :\, \S{t}  \{ \phi_{1'} \}  &
      \Gamma\ \vdash\ \eip\, :\, (\xp : \tau) \rightarrow \mathsf{PE}^{\el} \
              \{ \phi_2 \} \ \nu' : \S{t'} \ \{ \phi_{2'} \}\\
      \Gamma' = \Gamma, \xp : \tau, \mathsf{h_i : heap} \quad\ \mathsf{h_i}\ \mbox{fresh}}
	    { \begin{array}{@{}c@{}}
		  \Gamma' \vdash p\ \textnormal{>>=}\ \eip : \\
		  \mathsf{PE}^{\el} \ \{ \forall \mathsf{h}.\, \phi_1 \, \mathsf{h}\,  \wedge 
                                               \phi_{1'} (\mathsf{h}, x, \mathsf{h_i})  => \phi_2 \ \mathsf{h_i} \} \\
  			          \ \nu'\,:\, \S{t'}\ \S{result} \\  
				  \{ \forall \mathsf{h}, \nu', \mathsf{h'}, \S{y}.\, 
				  ({x} \neq \S{Err} => \nu' = \S{Inl}\ \S{y} \wedge \phi_{1'} (\mathsf{h}, x, \mathsf{h_i})  \wedge 
                                               \phi_{2'} ({\mathsf{h_i}, \S{y}, \mathsf{h'}}))\,
                    \wedge \\
                   \mkern-18mu ({x} = \S{Err} => \nu' = \S{Inr}\ {\S{Err}} \wedge \phi_{1'} (\mathsf{h}, x, \mathsf{h_i}))                           \} 
	     \end{array}}
\end{minipage}

\begin{flushleft}
\bigskip
{\bf Subtyping} \quad\ \fbox{\small $\Gamma \vdash$ $\sigma_1$ <: $\sigma_2$
}
\end{flushleft}
\bigskip
\begin{minipage}{0.40\textwidth}
 \inference[{\sc T-Sub-Base}]{ \Gamma \vdash \{ \nu : \S{t} \mid \phi_1 \} & \Gamma \vdash \{ \nu : \S{t} 
\mid \phi_2 \} \\
			  \Gamma \vDash \phi_1 => \phi_2 }
				{ 
				\Gamma \vdash \{ \nu : \S{t} \mid \phi_1 \} <: \{ \nu : \S{t} \mid \phi_2 
\}
				  } 
\end{minipage}\\[5pt]
\begin{minipage}{0.50\textwidth}
 \inference[{\sc T-Sub-Arrow}]{ \Gamma \vdash \tau_{21} <: \tau_{11} & \Gamma \vdash \tau_{12} <: \tau_{22}}
			  { \Gamma \vdash (\S{x }: \tau_{11}) \rightarrow \tau_{12}  <: (\S{x} : \tau_{21}) 
\rightarrow \tau_{22} } 
\end{minipage}\\[10pt]
\begin{minipage}{0.4\textwidth}
 \inference[{\sc T-Sub-Schema}]{\Gamma \vdash \sigma_1 <: \sigma_2 }
				{\begin{array}{@{}c@{}}  
				    \Gamma  \vdash  
				      \forall \alpha. \sigma_1 <: \forall \alpha. \sigma_2 
				    \end{array}			  
				  } 
\end{minipage}
\quad\quad
\begin{minipage}{0.4\textwidth}
 \inference[{\sc T-Sub-TVar}]{ }
				{\begin{array}{@{}c@{}}  
				    \Gamma  \vdash  
				      \alpha  <: \alpha 
				    \end{array}			  
				  }

\end{minipage}\\[10pt]
\begin{minipage}{\textwidth}
 \inference[{\sc T-Sub-Comp}]{ \Gamma \vDash \phi_2 => \phi_1  &    \Gamma \vdash \tau_1 <: \tau_2 &
				    \Gamma \vdash \el_1 \leq \el_2 &   \Gamma, \phi_2 
\vDash (\phi_{1'} => \phi_{2'})}
				{\begin{array}{@{}c@{}}  
				    \Gamma  \vdash  
				     \mathsf{PE^{\el_1}} \ \{ \phi_1 \} \ \tau_1 \ \{ \phi_{1'} \} 
					 <:
				      \mathsf{PE^{\el_2}} \ \{ \phi_2 \} \ \tau_2 \ \{ \phi_{2'} \} 
				    \end{array}			  
				  } 
				  
\end{minipage}
\caption{Typing semantics for primitive parser expressions and subtyping rules.}
\label{fig:typing-parserss}
\end{figure*}
Figure~\ref{fig:typing-parserss} presents the type rules for
\name\ parser expressions. ({\sc T-sub}) rule defines the standard type subsumption rule. The empty string parser typing rule ({\sc
  T-p-eps}) assigns a type with \S{pure} effect and \S{unit} return
type, while the postcondition establishes the equivalence of the input
and the output heaps. The {\sc T-p-bot} rule captures the always
failing semantics of $\bot$ with an exception effect \S{exc} and
corresponding return types and return values while maintaining the
stability of the input heap.

The type rules governing a character parser ({\sc T-p-char}) is more
interesting because it captures the semantics of the success and the
failure conditions of the parser. We use a sum type ($\alpha$
\S{result}) to define two options representing a successful and
exceptional result, resp. (with the \S{Err} exception value in the
latter case), using standard injection functions to differentiate
among these alternatives. In the successful case, the returned value
is equal to the consumed character, captured by an equality constraint
over characters. In the successful case, the structure of the output heap
with respect to the parse string \S{inp} must be the same as the input
heap except for the absence of the {\sf 'c'}, the now consumed head-of-string
character. In the failing case, the input remains unconsumed. Note that we also join the effect labels ($\S{state} \sqcup \S{exc}$), highlighting the state and exception effect. These effect labels form a standard join semi-lattice with an ordering relation ($\leq$)~\footnote{Details of the effect-labels and their join semi-lattice is provided in the supplementary material.}.

Rule {\sc T-p-choice} defines the static semantics for a
non-deterministic choice parser. It introduces a non-determinism
effect to the parser's composite type. The effect's precondition
requires that either of the choices can occur; we achieve this by
restricting it to the conjunction of the two preconditions for the
sub-parsers. The disjunctive post-condition requires that both the
choices must imply the desired goal postcondition for a composite
parser to be well-typed. The effect for the choice expression takes a
join over the effects of the choices and the non-deterministic effect.

Rule ({\sc T-P-Fix}) defines the semantics for the terminating recursive fix-point combinator.
Given an annotated type $\tau$ for the parameter \S{x}, if the type of the body $p$ in an extended environment which has \S{x} mapping to $\tau$, is $\tau$, then $\tau$ is also
a valid type for a recursive fixpoint parser expression. 
The {\sc T-p-bind} rule defines a typing judgement for the exceptional
monadic composition of a parser expression $p$ with an abstraction
\eip.  The composite parser is typed in an extended environment
($\Gamma$') containing a binding for the abstraction's parameter $x$
and an intermediate heap $\mathsf{h_{i}}$ that acts as the
output/post-state heap for the first parser and the input/pre-state
for the second.  The relation between these heaps is captured by the
inferred pre-and post-conditions for the composite parser.  There are
two possible scenarios depending upon whether the first parser $p$
results in a success (i.e. x $\neq$ \S{Err}) or a failure (x =
\S{Err}).  In the successful case, the inferred conditions capture the
following properties: a) the output of the combined parser is a
success; b) the post-condition for the first expression over the
intermediate heap $\mathsf{h_{i}}$ and the output variable x should
imply the precondition of the second expression (required for the
evaluation of the second expression); and, c) the overall
post-condition relates the post-condition of the first with the
precondition of the second using the intermediate heap $\mathsf{h_i}$.
The case when $p$ fails causes the combined parser to fail as well,
with the post-condition after the failure of the first as the overall
post-condition.
Note that the core calculus is sub-optimal in size since $\lambda_{sp}$ supports both \S{return} and \S{eps}, even though the latter could be modeled using \S{return}. However, this design choice enables decidable typechecking by limiting the combination of higher-order functions, combinators and states. This is achieved using a limited bind \S{p >>= e}, rather than the general \S{e >>= e}, allowing for the definition of semantic actions \S{e} that only perform limited state manipulation, i.e., reading and updating locations. Thus \S{>>=} and \S{<|>} only take parser arguments; thus, eps <|> p is not equivalent to (return () <|> p), in fact the latter is disallowed.
Another such design restriction shows up in the typing rules, e.g., the typing rule for function application (T-APP) restricts the arguments to be of {\it basetype}, thus disallowing expressions returning abstractions or computations, like return ($\lambda$x. e) or return (x := e)
A more general definition for \S{>>=} will allow valid HO arguments, like $\lambda$x. e >>= e1, but translating such general HO stateful programs to decidable logic fragments is not always feasible, as is discussed in other fully dependent type systems~\cite{fstar}.

The subtyping rules enable the propagation of refinement type information
and relate the subtyping judgments to logical entailment. The
subtyping rule for a base refinement ({\sc T-Sub-Base}) relates
subtyping to the logical implication between the refinement of the
subtype and the supertype. The ({\sc T-Sub-Arrow}) rule defines
subtyping between two function refinement types. The ({\sc
  T-Sub-Comp}) rule for subtyping between computation types follows
the standard Floyd-Hoare rule for {\it consequence}, coupled with the
subtyping relation between result types and an ordering relation
between effects($\leq$). The subtyping rule for type variables (T-Sub-TVar)
relates each type variable to itself in a reflexive way, while the subtyping
for a type-schema lifts the subtyping relation from a schema to
another schema.

\subsection{Example}  
\label{example:4p4}
To illustrate the application of these typing rules, consider how 
we might type-check a simple \S{consume} parser, a parser that
successfully consumes the next character in an input stream
(\S{inp}).  An intuitive specification capturing a safety property
related to how inputs are consumed might be:
\begin{lstlisting}[escapechar=\@,basicstyle=\small\sf,breaklines=true,language=ML]
  consume : @$\mathsf{PE^{state}}$@ { @$\forall$@ h. true }  $\nu$ :  char { @$\forall$@ h $\nu$ h'. $\nu$ = hd (sel h inp) @$\wedge$@ len (sel h' inp) = len (sel h inp) @$\text{-}$@ 1) }
\end{lstlisting}
\noindent that simply establishes that the parser's output is a character and that the
length of the input stream after the character has been consumed is
one less than its length before the consumption.

Using this parser, we can define a parser for consuming {\sf
  k} elements, called {\sf k-consume}, which is defined in terms of 
\S{count}, a derived parser available in the \name\ library.  Thus,
  {\sf k-consume} $\equiv$ \S{count\ k\ consume}, and translates to the
following definition, in which specifications in gray are inferred by \name:
\vspace*{-.2in}
\begin{lstlisting}[escapechar=\@,basicstyle=\small\sf,breaklines=true,language=ML]
@\begin{tabbing}
let k-consume = \\
fix ($\lambda$\= k-consume : \textcolor{gray}{(k : int) $\rightarrow$ \{$\forall$ h. true\} $\nu$ : char list 
\{$\forall$ h $\nu$ h'. len ($\nu$) = k $\wedge$ len (sel h inp) $\text{-}$ len (sel h' inp) = k\}}.\\
\> if (k <= 0)\ \=then (eps >>= ($\lambda$\_. return []))  \\
\>\>             else (consume >>= $\lambda$ x : char. k-consume (k-1) >>= $\lambda$ xs : char list. return (x :: xs))
\end{tabbing}@
\end{lstlisting}
Now, applying  rule {\sc T-p-fix}, we need to prove the following requirement:
\vspace*{-.2in}
\begin{lstlisting}[escapechar=\@,basicstyle=\small\sf,breaklines=true,language=ML]
@\begin{tabbing}  
 $\Gamma$,  (k-consume : \textcolor{gray}{(k : int) $\rightarrow$ \{true\} v : char list \{ len (v) = k $\wedge$ len (sel h inp) $\text{-}$ len (sel h' inp) = k\}}) $\vdash$\\ 
\quad\quad \=if (k <= 0)\ \= then (eps >>= ($\lambda$\_. return [])) \\
\>\>                     else (consume >>= $\lambda$ x : char. k-consume (k $\text{-}$ 1) >>= $\lambda$ xs : char list. return (x :: xs))) : \\
\>           \textcolor{gray}{(k : int) $\rightarrow$ \{ true \} v : char list \{ len (v) = k $\wedge$ len (sel h inp) $\text{-}$ len (sel h' inp) = k \}}
\end{tabbing}@
\end{lstlisting}
\noindent i.e., we need to prove that, in an extended environment, with a
type-mapping for the fixpoint combinator's argument ({\sf k-consume}),
the combinator's body also satisfies the type.   Using
the type for {\sf consume} and the typing rule for {\sc T-p-bind},
we can infer the type for the {\sf else} branch in the body:
\vspace*{-.2in}
\begin{lstlisting}[escapechar=\@,basicstyle=\small\sf,breaklines=true,language=ML]
@\begin{tabbing}
  $\Gamma$,\= (x : char), (hi : heap), (xs : char list), (hi' : heap),\\
  (k-consume : \textcolor{gray}{(k : int) $\rightarrow$ \{ true \} v : char list \{ len (v) = k $\wedge$ len (sel h inp) $\text{-}$ len (sel h' inp) = k \}}) $\vdash$ \\
\> (consume >>=\= $\lambda$ x : char. k-consume (k $\text{-}$ 1) >>= $\lambda$ xs\= : char list. return (x :: xs)) :\\
\>\quad                \textcolor{gray}{(k : int) $\rightarrow$ \{ true \} v : char list \{ }\= \textcolor{gray}{len (xs) = k - 1 $\wedge$ len (sel hi inp) $\text{-}$ len (sel h' inp) = (k - 1) $\wedge$}\\
\>\>                                                                      \textcolor{gray}{len (sel h inp) $\text{-}$ len (sel hi inp) = 1 $\wedge$} \\ \> \>
\textcolor{gray}{ hi = hi' $\wedge$ len (v) = len (xs) + 1 \}}
\end{tabbing}@
\end{lstlisting}
\noindent The \S{then} branch is relatively simpler and uses the semantics of the derived combinator {\sf map}\footnote{Definitions for these derived combinators are provided in the supplementary material.} and primitive combinator {\sf eps}:
\begin{lstlisting}[escapechar=\@,basicstyle=\small\sf,breaklines=true,language=ML]
 @$\Gamma$@, (x : unit), (hi : heap), (k@-@consume : @\textcolor{gray}{(k : int) $\rightarrow$ \{ true \} v : char list \{ len (v) = k $\wedge$ len (sel h inp) $\text{-}$ len (sel h' inp) = k \}}@) @$\vdash$@ 
   (eps >>= (@$\lambda$\_@. return []))  :  @\textcolor{gray}{(k : int) $\rightarrow$ \{ k=0 \} v : char list \{ len (v) = 0 $\wedge$ hi = h $\wedge $ len (sel hi inp) $\text{-}$ len (sel h' inp) = 0 \}}@
\end{lstlisting}
\noindent Finally, using the standard rule for {\sf if-then-else} (implemented using {\sf match}), and 
simplifying the conclusion in the post-condition for the \S{else} branch
shown earlier, we can infer that the type for the body agrees with
the type for the fixpoint combinator's argument, thus proving that the
{\sf k-consume} is correct with respect to the given specification.

However, consider a scenario where we change the definition of say,
{\sf k-consume}'s \S{else} branch, as follows:
\vspace*{-.05in}
\begin{lstlisting}[escapechar=\@,basicstyle=\small\sf,breaklines=true,language=ML]
@\quad\quad@ (consume >>= @$\lambda$@ x : char. k@-@consume (0) >>= @$\lambda$@ xs : char list. return (x :: xs))
\end{lstlisting}
\noindent Now, this definition of {\sf k-consume} does not run k-successive \S{consume}
parsers, but instead  only runs the \S{consume} parser once;  type-checking as
above fails.



 
 
 
 
 
 


\subsection{Properties of the Type System}

\begin{definition}[Environment Entailment $\Gamma \models \phi$]
  Given $\Gamma$ = \dots , $\overline{\phi_i}$, the entailment of a formula $\phi$ under $\Gamma$ is defined as 
   ($\bigwedge_{i}^{} \phi_i$)  $ \implies \phi$ 
\end{definition}

In the following, $\Gamma \models$ $\phi$($\mathcal{H}$) extends the notion of semantic entailment of a formula over an
abstract heap $\Gamma \models$ $\phi$ ({\sf h}) to a concrete heap
using an interpretation of concrete heap $\mathcal{H}$ to an abstract
heap {\sf h} and the
standard notion of well-typed {\it stores} ($\Gamma \vdash
\mathcal{H}$).\footnote{Details are provided in the supplemental material.}

To prove soundness of \name typing, we first state a soundness lemma
for pure expressions (i.e. expressions with non-computation type).
\begin{lemma}[Soundness Pure-terms]
 If\ $\Gamma \vdash \eip : \{ \nu\ :\ \mathsf{t}\ |\ \phi\ \}$ then:
 \begin{itemize}
   \item \textnormal{Either \eip\ is a value with $\Gamma \models$ $\phi$ ( \eip )} 
   \item \textnormal{OR Given there exists a $\S{v}$, such that ($\mathcal{H}$; \eip) $\Downarrow$ ($\mathcal{H}$; \S{v}) then $\Gamma \vdash$ \S{v} : t and $\Gamma \models$ $\phi$ ($\S{v}$) }
\end{itemize}   
\end{lemma}


\begin{theorem}[Soundness \name]
\label{thm:soundness}
Given a specification $\sigma$ = $\forall \overline{\alpha}$. $\mathsf{PE^{\el}}$
$\{\phi_1\}$ $\nu$ : {\sf t} $\{\phi_2\}$ and a \name\ expression \eip, such
that under some $\Gamma$, $\Gamma \vdash$ \eip : $\sigma$, then if
there  exists $\mathcal{H}$ such that $\Gamma \models \phi_1 (\mathcal{H})$
then:
\begin{enumerate}
  \item \textnormal{Either \eip \ is a value, and:} 
    $\Gamma , \phi_1$ $\models$ $\phi_2$ ($\mathcal{H},\ \eip$,\ $\mathcal{H}$)
  \item \textnormal{Or, if there exists an $\mathcal{H'}$ and \S{v} such that  ($\mathcal{H}$; \eip) $\Downarrow$ ($\mathcal{H'}$; \S{v}), then\\
  $\exists$ $\Gamma'$, $\Gamma \subseteq \Gamma'$ and ({\sf consistent $\Gamma$ $\Gamma'$}), such that:} 
  \begin{enumerate}
    \item $\Gamma'$ $\vdash$ $\S{v} : \S{t}$. 
    \item $\Gamma', \phi_1$ ($\mathcal{H}$) $\models$$\phi_2$ ($\mathcal{H},\ \S{v}$,\ $\mathcal{H'}$) 
     \end{enumerate}
\end{enumerate}
\end{theorem}


where ({\sf consistent $\Gamma$ $\Gamma'$}) is a Boolean-valued
function that ensures that $\forall$ x $\in$ (dom ($\Gamma$) $\cap$
dom ($\Gamma'$)). $\Gamma \vdash$ x : $\sigma$ $\implies$ $\Gamma'
\vdash$ x : $\sigma$. Additionally, $\forall \phi$. $\Gamma \models
\phi$ $\implies$ $\Gamma' \models \phi$.
\begin{proof}
The soundness proof is by induction on typing rules in
Figures~\ref{fig:typing-base} and ~\ref{fig:typing-parserss}, proving
the soundness statement against the evaluation rules in
Figures~\ref{fig:semantics-parser}.\footnote{Proofs for all theorems
  are provided in the supplemental material.}
\end{proof}

\paragraph{Decidability of Typechecking in \name}
Propositions in our specification language are first-order formulas in
the theory of EUFLIA~\cite{eufa}, a theory of equality of uninterpreted
functions and linear integer arithmetic.

The subtyping judgment in $\lambda_{sp}$ relies on the semantic
entailment judgment in this theory. Thus, decidability of type
checking in $\lambda_{sp}$ reduces to decidability of semantic
entailment in EUFLIA. The following lemma argues that the
verification conditions generated by \name typing rules always
produces a logical formula in the Effectively Propositional
(EPR)~\cite{epr,smt-eps} fragment of this theory consisting of formulae with
prenex quantified propositions of the forms $\exists^{*}$
$\forall^{*}$ $\phi$.  Off-the-shelf SMT solvers (e.g., Z3) are
equipped with efficient decision procedures for EPR
logic~\cite{smt-eps}, thus making typechecking decidable in \name.

\begin{definition}
  We define two judgments:
  \begin{itemize}
    \item $\vdash$ $\Gamma$ {\sf EPR} asserting that all
    propositions in $\Gamma$ are of the form $\exists^{*}$ $\forall^{*}$
    $\phi$ where $\phi$ is a quantifier free formula in {\sf EUFLIA}.
    \item $\Gamma$ $\vdash$ $\phi$ {\sf EPR}, asserting that under a given $\Gamma$,
    semantic entailment of $\phi$ is always of the form $\exists^{*}$
    $\forall^{*}$ $\phi'$.  
  \end{itemize}
\end{definition}

\begin{lemma}[Grounding]{\label{lem:grounding}}
If $\Gamma$ $\vdash$ {\sf e} : $\tau$,
then $\vdash$ $\Gamma$ {\sf EPR} and if\ $\Gamma$ $\vDash$ $\phi$ then
$\Gamma$ $\vdash$ $\phi$ {\sf EPR}
\end{lemma}

\begin{theorem}[Decidability \name]
  \label{thm:decidability}
  Typechecking in \name\ is decidable.
  \end{theorem}

\section{Evaluation}
\label{sec:evaluation}

\subsection{Implementation}

\name\ is implemented as a deeply-embedded DSL in OCaml\footnote{An anonymized repository link is provided in the supplemental material.} equipped with
a refinement-type based verification system encoding the typing rules
given in Section~\ref{sec:typing} and a parser translating an
OCaml-based surface language of the kind presented in our motivating
example to the \name\ core, described in Section~\ref{sec:syntax}.  To
allow \name\ programs to be easily used in an OCaml development, its
specifications can be safely erased once the program has been
type-checked.  Note that a \name program, verified against a safety
specification is guaranteed to be safe when erased since verification
takes place against a stricter memory abstraction; in particular,
since \name\ programs are free of aliasing by construction and thus remain so
when evaluated as an ML program.  This obviates the need for a
separate interpreter/compilation phase and gives \name-verified
parsers efficiency comparable to the parsers written using OCaml
parser-combinator libraries~\cite{mparser,angstrom}.

\name\ specifications typically require meaningful qualifiers over
inductive data-types, beyond those discussed in our core language; in
addition to the qualifiers discussed previously, typical examples
include qualifiers to capture properties such as the length of a list,
membership in a list, etc.  \name provides a way for users to write
simple inductive propositions over inductive data types, translating
them to axioms useful for the solver, in a manner similar to the use
of {\it measures} and {\it predicates} in other refinement type
works~\cite{liquidoriginal,liquidextended}.  For example, a qualifier
for capturing the length property of a list can be written as:
\begin{lstlisting}[escapechar=\@,basicstyle=\small\sf,breaklines=true]
   qualifier len [] $\rightarrow$ 0 | len (x :: xs) $\rightarrow$ len (xs) + 1.
\end{lstlisting}
\name\ generates the following axiom from this qualifier:
\begin{lstlisting}[escapechar=\@,basicstyle=\small\sf,breaklines=true]
   $\forall$ xs : $\alpha$ list, x : $\alpha$. len (x :: xs) = len (xs) + 1 $\wedge$ len [] = 0
\end{lstlisting}
\name\ is implemented in approximately 9K lines of OCaml code.  The
input to the verifier is a \name\ program definition, correctness
specifications, and any required qualifier definitions.  Given this,
\name\ infers types for other expressions and component parsers,
generates first-order verification conditions using the typing
semantics discussed earlier, and checks the validity of these
conditions.

%

%
%

\subsection{Results and Discussions}

We have implemented and verified the examples given in the paper,
along with a set of benchmarks capturing interesting, real-world
safety properties relevant to data-dependent parsing tasks. The goal
of our evaluation is to consider the effectiveness of \name\ with
respect to generality, expressiveness and practicality.
Table~\ref{table:results} shows a summary of the benchmark programs
considered. Each benchmark is a \name\ parser program affixed with a
meaningful safety property (last column). The first column gives the
name of the benchmark. The second column of the table describes
benchmark size in terms of the number of lines of \name\ code, without
the specifications. The third column gives a pair D/P, showing the
number of unique derived (D) combinators (like \textsf{count}, \textsf{many}, etc.) used in the
benchmark from the \name\ library, and the number of primitive (P)
parsers (like \textsf{string}, \textsf{number}, etc.) from the \name library used in the benchmark; the former provides some insight on the
usability of our design choices in realizing extensibility. The fourth
column lists the size of the grammar along with the number of
production rules in the grammar. The fifth column gives the number of
verification conditions generated, followed by the time taken to verify
them (sixth column). The overall verification time is the time taken
for generating verification conditions plus the time Z3 takes to solve
these VCs. All examples were executed on a 2.7GHz, 64 bit Ubuntu
The seventh column quantifies the annotation effort for verification. It gives a ratio (\#A/\#Q) of required user-provided specifications (in terms of the number of conjuncts in the specification) to the total specification size (annotated + inferred). User-provided specifications are required to specify a top-level
safety property and to specify invariants for \textsf{fix}
expressions akin to loop invariants that would be provided in a
typical verification task.

Finally, the last column gives a
high-level description of the data-dependent safety property being
verified.

%
%
%

Our benchmarks explore data-dependent parsers from several interesting
categories.\footnote{ The grammar for each of our implementations is
  given in the supplemental material.}
The first category, represented by {\sf Idris do-block},
\textsf{Haskell case-exp} and \textsf{Python while-statement}, capture
parsing activities concerned with layout and indentation introduced
earlier. Languages in which layout is used in the definition of their
syntax require context-sensitive parser
implementations~\cite{indentation1, indentation2}. We encode a
Morpheus parser for a sub-grammar for these languages whose
specifications capture the layout-sensitivity property.
  
The second category, represented by \textsf{png} and \textsf{ppm}
consider data-dependent image formats like PNG or PPM.
Verifying data-dependence is non-trivial as it requires verifying an
invariant over a monadic composition of the output of one parser
component with that of a downstream parser component, interleaved with
internal parsing logic.

The next category, captured by \textsf{xauction}, \textsf{xprotein},
and \textsf{health}, represent data-dependent parsing in
data-processing pipelines over XML and CSV databases.  For {\sf
  xauction} and {\sf xprotein}, we extend  XPath expressions over
XML to {\it dependent} XPath expressions.  Given that XPath
expressions are analogous to regular-expressions over structured XML
data, {\it dependent} XPath expressions are analogous to dependent
regular-expressions over XML.  We use these expressions to encode a
property of the XPath query over XML data for an online auction and
protein database, resp. Note that verifying such properties over XPath
queries is traditionally performed manually or through testing. In the
case of \textsf{health}, we extend regular custom pattern-matching
over CSV files to stateful custom pattern-matching, writing a
data-dependent custom pattern matcher. We verify that the parser
correctly checks relational properties between different columns in
the database.

The next two categories have one example each: we introduced the {\sf
  c typedef} parser in Section~\ref{sec:overview} that uses data
dependence and effectful data structures to disambiguate syntactic
categories (e.g., \emph{typenames} and \emph{identifiers}) in a
language definition.  Benchmark {\sf streams} defines a parser over
streams (i.e. input list indexed with natural numbers).
\paragraph{Annotation overhead vs inference.}
There are some interesting things to note in the second to last column; First, as the benchmarks (grammars) become more complex, i.e., have a greater number of functions (sub-parsers), the ratio decreases (small is better). In other words, the gains of type-inference become more visible (e.g., Haskell, Idris, C typedef). This is because Morpheus easily infers the types of these functions (sub-parsers).
The worst (highest) ratio is for the PPM parser. This parser is interesting because, even though the grammar is small, it makes multiple calls to fixpoint combinators. Thus, the user must provide specifications for the top-level parser and each fix-point combinator. Additionally, given a small number of functions (sub-parsers) due to small grammar size, the gains due to inference are also low. In summary, these trends show that the efforts needed for verification are at par with other Refinement typed languages (like., Liquid Types~\cite{liquidoriginal}, FStar~\cite{fstar}, etc, and as the parsers become bigger, the benefits of inference become more prominent.

\begin{table}[t]
\centering 
\small\begin{tabular}{| l | c | c | c | c | c | c | l |} 
\hline 
Name	& \# Loc & D/P & G(\#prod) & \# VCs& T (s) & (\#A/\#Q) & data-dependence \\
\hline

haskell case-exp   & 110 & 5/4 & 20 (7) 	&  17  	&  8.11 &  9/39 & \scriptsize{layout-sensitivity}   \\
idris do-block 	   & 115  & 5/5   &  22(8)  		&  33	&  10.46  & 7/26  & \scriptsize {layout-sensitivity}  \\
python while-block & 47   & 3/3   & 25 (7)	&  23	&  7.44 & 6/20  &  \scriptsize {layout-sensitivity} \\
ppm 		   & 46   & 5/2   & 21 (7)	&  20	&  5.33   & 4/9 &  \scriptsize {tag-length-data}  \\
png chunk 	   & 30   & 3/4	  & 10 (2) 	&  12     & 3.38  & 2/7 & \scriptsize {tag-length-data} \\
xauction	   & 54   & 4/4   & 31 (10) 	&  19 	&  6.70 & 2/8 & \scriptsize {data-dependent XPath expression} \\
xprotein	&  45  & 3/3   & 24(6)		&  22	& 6.23 & 4/10  & \scriptsize {data-dependent XPath expression}\\
health 		&  40  & 4/3	  & 15(5)	& 13 	& 4.56	& 2/8  & \scriptsize {data-dependent CSV pattern-matching} \\

c typedef  	   & 60   & 4/4	  & 14 (5)	&  21	& 6.78	& 4/16 & \scriptsize {context-sensitive disambiguation} \\
streams		&  51  & 4/2   &  12 (4)   	& 16	&  5.21	& 2/9 & \scriptsize {safe stream manipulation} \\

\hline 
\end{tabular}
\caption{Summary of Benchmarks : {\small \#Loc Loc defines the size of the parser implementation in \name ; 
D/P gives the number of derived/primitive combinator uses in the parser implementation;
grammar size G(\# prod) defines size of the grammar along with the number of production rules in the grammar; \#VCs defines number of VCs generated; T(s) is the time for discharging these VCs in seconds; (\#A/\#Q) defines the ratio of number of conjuncts used in the specification provided by the user (\#A) to the total number of conjuncts (\#Q) across all files in the implementation; Property gives a high-level description of the data-dependent safety property. }} 
\label{table:results} 
\end{table}

\subsection{Case Study: Indentation Sensitive Parsers}
As a case study to illustrate \name's capabilities, we consider a
particular class of stateful parsers that are {\it
  indentation-sensitive}, and which are widely used in many functional
language implementations. These parsers are characterized by having
indentation or layout as an essential part of their grammar. Because
indentation sensitivity cannot be specified using a context-free
grammar, their specification is often specified via an orthogonal set
of rules, for example, the offside rule in Haskell.\footnote{https://www.haskell.org/onlinereport/haskell2010/haskellch10.html}
Haskell language specifications define these rules in a complex routine found
in the lexing phase of the compiler~\cite{marlow}. Other indentation-sensitive
languages like Idris~\cite{idris} use parsers written using a parser
combinator libraries like \S{Parsec} or its
variants~\cite{parsec,megaparsec} to enforce indentation constraints.

\begin{figure}
\begin{subfigure}[t]{0.48\textwidth}
\begin{lstlisting}[escapechar=\@,basicstyle=\small\sf,breaklines=true,showspaces=false,
  showstringspaces=false, tabsize=2, breaklines=true,
  xleftmargin=5.0ex]
 DoBlock ::= 'do' OpenBlock Do* CloseBlock;
 Do ::=
    'let' Name  TypeSig'      '=' Expr
   | 'let' Expr'                '=' Expr
   |  Name  '<-' Expr
   |  Expr' '<-' Expr
   |  Ext Expr
   |  Expr  
\end{lstlisting}
\caption{An Idris grammar rule for a {\sf do} block}
\label{fig:idris-grammar}
\end{subfigure}
\begin{subfigure}[t]{0.48\textwidth}  
\begin{lstlisting}[escapechar=\@,basicstyle=\small\sf,breaklines=true,showspaces=false,
  showstringspaces=false, tabsize=2, breaklines=true,
  xleftmargin=5.0ex]
  expr =  do 
            t <- term 
  	        symbol "+" 
	          e <- expr 
	          pure t + e
         symbol '*'      
\end{lstlisting}
\caption{An input to the parser.}
\label{fig:idris-grammar-example}
\end{subfigure}
\caption{An Idris grammar rule for a {\sf do} block and an example input.}
\end{figure}
Consider the Idris grammar fragment shown in
Figure~\ref{fig:idris-grammar}. The grammar defines the rule to parse
a \S{do}-block. Such a block begins with the \S{do} keyword, and is
followed by zero or more \S{do} statements that can be \S{let}
expressions, a binding operation ($\leftarrow$) over names and
expressions, an external expression, etc. The Idris documentation
specifies the indentation rule in English governing where these
statements must appear, saying that the ``\emph{indentation of each
  \S{do} statement in a \S{do}-block {\sf Do*} must be greater than
  the current indentation from which the rule is
  invoked}~\cite{idrisgrammar}.'' Thus, in the Idris code fragment
shown in Figure~\ref{fig:idris-grammar-example}, indentation sensitivity
constraints require that the last statement is not a part
of the do-block, while the inner four statements are. A correct Idris
parser must ensure that such indentation rules are preserved.

Figure~\ref{fig:idris-original-do} presents a fragment of the parser
implementation in Haskell for the above grammar, taken from the Idris
language implementation source, and simplified for ease of
explanation. The implementation uses Haskell's \S{Parsec} library, and
since the grammar is not context-free, it implements indentation rules
using a state abstraction (called {\sf IState}) that stores the
current indentation level as parsing proceeds. The parser then manually performs reads and
updates to this state and performs indentation checks at appropriate
points in the code (e.g. line~\ref{line:lo1},~\ref{line:indentGreater2}).

\begin{figure}
\begin{subfigure}[t]{0.48\textwidth}
\begin{lstlisting}[escapechar=\@,basicstyle=\small\sf,numbers=left,language=Haskell,showstringspaces=false]
data IState = IState {
    ist :: Int
    @\dots@
} deriving (Show)
data PTerm =  PDoBlock [PDo]                    data PDo t =  DoExp t | DoExt t | DoLet  t t | @\ldots@            
type IdrisParser a = Parser IState a @\label{line:idrisp}@

getIst :: IdrisParser IState
getIst = get 
putIst :: (i : Int) -> IdrisParser ()
pustIst i = put {ist = i}

doBlock :: IdrisParser PTerm
doBlock = do  
            reserved "do"
            ds <- indentedDoBlock
            return (PDoBlock ds)
indentedDo ::  IdrisParser (PDo PTerm)
indentedDo = do 
               allowed <- ist getIst @\label{line:getind}@
               i <- indent
               if (i <= allowed) @\label{line:lo1}@
                  then fail ("end of block")
               else do_ @\label{line:originaldo}@ 
indent :: IdrisParser Int                     
indent =  @\label{line:indentdef}@
    do    
     if (lookAheadMatches (operator)) then 
        do 
          operator 
          return (sourceColumn.getSourcePos)
      else 
          return (sourceColumn.getSourcePos)
                    
\end{lstlisting}
\end{subfigure}
\begin{subfigure}[t]{0.48\textwidth}
\begin{lstlisting}[escapechar=\@,basicstyle=\small\sf,numbers=left,breaklines=true,firstnumber=34,language=Haskell,showstringspaces=false]
do_ :: IdrisParser (PDo PTerm)
do_  =  do 
          reserved "let"
          i <- name
          reservedOp "="
          e <- expr 
          return (DoLet i e)
    <|> do 
          e <- expr 
          return (DoExt  i e)
    <|> do e <- expr 
           return (DoExp  e)
indentedDoBlock :: IdrisParser [PDo PTerm]
indentedDoBlock  = @\label{line:indentdostart}@
    do 
        allowed <- ist getIst @\label{line:indentdobegin}@  
        lvl' <- indent 
        if (lvl' > allowed) then @\label{line:indentGreater2}@
            do 
                putIst lvl' @\label{line:setind}@ 
                res <- many (indentedDo) @\label{line:many}@ 
                putIst allowed @\label{line:reset}@
                return res 
        else fail "Indentation error"

lookAheadMatches :: IdrisParser a -> IdrisParser Bool 
lookAheadMatches p = 
    do 
        match <- lookAhead (optional p)
        return (isJust match)
\end{lstlisting}
\end{subfigure}
\caption{A fragment of a \S{Parsec} implementation for Idris \S{do}-blocks with indentation checks.}
\label{fig:idris-original-do}
\end{figure}
The IdrisParser (line~\ref{line:idrisp}) is defined in terms of
\S{Parsec}'s parser monad over an Idris state (here, {\sf IState}), which
along with other fields has an integer field (\S{ist}) storing the
current indentation value.  A typical indentation check (e.g. see
lines~\ref{line:getind} - ~\ref{line:lo1}) fetches the current value
of \S{ist} using \S{getIst}, fetches the indentation of the next
lexeme using the \S{Parsec} library function \S{indent}, and compares
these values.

The structure of the implementation follows the grammar
(Figure~\ref{fig:idris-grammar}): the {\sf doBlock} parser parses a
reserved keyword ``\S{do}'' followed by a block of \S{do\_} statement
lists.  The indentation is enforced using the parser
\S{indentedDoBlock} (defined at line~\ref{line:indentdostart}) that
gets the current indentation value (\S{allowed}) and the indentation for
the next lexeme using \S{indent}, checks that the indentation is
greater than the current indentation (line~\ref{line:indentGreater2})
and updates the current indentation so that each \S{do} statement is
indented with respect to this new value. 

It then calls a parser
combinator {\sf many} (line~\ref{line:many}), which is the \S{Parsec}
combinator for the Kleene-star operation, over the result of
\S{indentedDo}, i.e., $\mathsf{indentedDo}^{*}$.  The {\sf
  indentedDo} parser again performs a manual indentation check,
comparing the indentation value for the next lexeme against the
block-start indentation (set earlier by {\sf indentedDoBlock} at line
~\ref{line:setind}) and, if successful, runs the actual \S{do\_}
parser (line~\ref{line:originaldo}).  Finally, {\sf indentedDoBlock} resets
the indentation value to the value before the block
(line~\ref{line:reset}).

Unfortunately, it is non-trivial to reason that these manual checks
suffice to enforce the indentation sensitivity property we desire.
Since they are sprinkled throughout the implementation, it is easy to
imagine missing or misplacing a check, causing the parser to
misbehave. More significantly, the implementation make incorrect
assumptions about the effectful actions performed by the library that
are reflected in API signatures. In fact, the logic in the above code
has a subtle bug~\cite{indentation1} that manifests in the input
example shown in Figure~\ref{fig:inputexample}.

Note that the indentation of the token `mplus' is such that it is
not a part of either \S{do} block; the implementation, however, parses
the last statement as a part of the inner do-block, thereby violating the
indentation rule, leading to the program being incorrectly parsed.

The problem lies in a mismatch between the contract provided by the
library's \S{indent} function and the assumptions made about its
behavior at the check at line~\ref{line:lo1} in the {\sf indentedDo}
parser (or similarly at line~\ref{line:indentGreater2}). 
Since checking indentation levels for each character is costly, \S{indent}
  is implemented (line~\ref{line:indentdef}) in a way that causes
  certain lexemes (user defined operators like `mplus') to be ignored
  during the process of computing the next indentation level. It uses
  a \textsf{lookAdheadMatches} parser to skip all lexemes that are
  defined as operators.  In this example, {\sf indent} does not check
the indentation of lexeme `mplus', returning the indentation of the
token \S{pure} instead.  Thus, the indentation of the last statement
is considered to start at \S{pure}, which incorrectly satisfies the
checks at line~\ref{line:lo1} or line~\ref{line:indentGreater2}, and
thus causes this statement to be accepted as part of
\S{indentedDoBlock}.

\begin{wrapfigure}{r}{.5\textwidth}
  \vspace*{-.2in}
\begin{lstlisting}[escapechar=\@,basicstyle=\small\sf,breaklines=true,language=Haskell, numbers=left, showspaces=false,
  showstringspaces=false, tabsize=2, breaklines=true,
  xleftmargin=5.0ex]
  expr = do 
            t <- term 
	          do 
              symbol "+" 
	            e <- expr 
	            pure  t + e
	       `mplus` pure t
\end{lstlisting}
\caption{An input expression that is incorrectly parsed by the implementation shown in Figure~\ref{fig:idris-original-do}.}
\label{fig:inputexample}
\vspace*{-.20in}
\end{wrapfigure}
Unfortunately, unearthing and preventing such bugs is challenging. We show how implementing the same parser in \name allows us to catch the bug and verify a correct version of the parser.

Figure~\ref{fig:idris-do} shows a \name implementation for a portion
of the Idris {\sf doBlock} parser from
Figure~\ref{fig:idris-original-do} showing the implementation of three
parsers for brevity, {\sf doBlock}, {\sf indentedDo}, and {\sf
  indent}, along with other helper functions. The structure is similar to the original Haskell
implementation, even though the program uses ML-style operators for assignment
and dereferencing. For ease of presentation, we have written the
program using {\it do-notation} (\M{do}) as syntactic sugar for
\name's monadic bind combinator.

\begin{figure}
  \begin{subfigure}[t]{0.4\textwidth}
  \begin{lstlisting}[escapechar=\@,basicstyle=\footnotesize\sf,numbers=left,language=ML,showspaces=false,showstringspaces=false,tabsize=1]
type @$\alpha$@ pdo  =  DoExp of @$\alpha$@ | DoExt of @$\alpha$@ | @\ldots@    
type pterm =  PDoBlock of ((pterm pdo) list)                        
let ist = ref 0
@$\ldots$@
@\scriptsize\texttt{\textcolor{blue}{
\begin{tabbing}
doBlock :\\ \label{line:doBlockSpec}
$\mathsf{PE^{stexc}}$\=\\
\>\{$\forall$ h, I. sel\= (h, ist) = I\} \\
\>\quad            $\nu$ : (offsideTree I) result  \\
\>\{$\forall$ h, $\nu$, h', I, I'.\=\\
\>\,\,\,\,( $\nu$ = Inl (\_) => (sel (h, ist) = I $\wedge$ \\
\>\,\,\,\,\,sel (h', ist) = I') => I' = I) \\ 
\>\,\,\,\,\,$\wedge$ $\nu$ = Inr (Err) => \\
\>\,\,\,\,\,(sel (h', inp) $\subseteq$ sel (h, inp)) \}
 \end{tabbing}}}@
let doBlock @\label{line:doblock}@ = 
    @\M{do}@ 
       dot <- reserved "do" @\label{line:do}@
       ds  <- indentedDoBlock @\label{line:ds}@
       return Tree {term = PDoBlock ds; @\label{line:l1}@
                   indentT = indentT (dot);
                   children = (dot :: ds) @\label{line:lcons}@}
@\scriptsize\texttt{\textcolor{gray}{
   \begin{tabbing}
     do\_ : $\mathsf{PE^{stexc}}$\=
\,\{$\forall$ h, I. sel\= (h, inp) = I\} \label{line:dospec}\\  
\>\quad              $\nu$ : tree result    \\
\>\{$\forall$ h, $\nu$, h', I, I'.\=\\
($\nu$ = Inl(\_) => \\
 \>\ indentT($\nu$)\= = pos (sel (h, inp)) \\
\>\>children ($\nu$) = nil ) \\
 $\wedge$ $\nu$ = Inr (Err) => \\
\>\,\,\,\,\,(sel (h', inp) $\subseteq$ sel (h, inp)) \}
   \end{tabbing}}}@                     
let do_ = @$\ldots$@
@\scriptsize\texttt{\textcolor{gray}{
   \begin{tabbing}
     lookAheadMatches : $\mathsf{PE^{pure}}$\=
\,\{true\} 
$\nu$ : bool\
\{[h'=h]\}
   \end{tabbing}}}@ 
lookAheadMatches p = 
        $\M{do}$ 
            match <- lookAhead (optional p)
            return (isJust match)
\end{lstlisting}
  \end{subfigure}
  \begin{subfigure}[t]{0.5\textwidth}
  \begin{lstlisting}[escapechar=\@,basicstyle=\footnotesize\sf,numbers=right,breaklines=true,firstnumber=20,language=ML,showstringspaces=false]
@\label{line:inferredindentDo}\scriptsize\texttt{\textcolor{gray}{
  \begin{tabbing}                        
    ind\=entedDo :\\
  \> $\mathsf{PE^{stexc}}$\,\= \{$\forall$ h, I.sel\= (h, ist) = I \}\\ 
  \>\>\quad   $\nu$ : tree result  \\
  \>\{$\forall$ h, $\nu$, h', I, I'.\=\\
 $\forall$ i :\= int.\=
   (i <= I $\Rightarrow$ sel (h', inp) $\subseteq$ sel (h, inp)) $\wedge$ \\
  \>\>\>(i > I $\Rightarrow$ \=indentT ($\nu$) = pos (sel (h, inp) $\wedge$ \\
  \>\>\>\>children ($\nu$) = nil\} 
  \end{tabbing}}}@
let indentedDo = @\label{line:idostart}@ 
    @\M{do}@ 
       allowed <- !ist
       i <- indent @\label{line:indentcall}@ 
       if (i <= allowed ) @\label{line:l3}@ then 
           fail ("end of block")
       else    
            do_ @\label{line:l4}@ 
@\scriptsize\texttt{\textcolor{gray}{  
\begin{tabbing}
sourceColumn : (char * int) list -> int \label{line:sc}
\end{tabbing}
}}@
let sourceColumn = @$\ldots$@
@\scriptsize\texttt{\textcolor{gray}{  
\begin{tabbing}
indent : \label{line:indentspec}
$\mathsf{PE^{state}}$\= \{true\}\,\\
  \>\quad      $\nu$ : int\, \\
  \>\{$\forall$ h, $\nu$, h'.\=\\

  \> sel (h', inp) $\subseteq$ sel (h, inp) \}
\end{tabbing}
}}@ 
let indent = 
    @\M{do}@
      if (lookAheadMatches (operator)) then 
        @\M{do}@ 
          operator
          return (sourceColumn !inp)
      else 
          return (sourceColumn !inp)
 \end{lstlisting}
  \end{subfigure}
  \caption{\name\ implementation and specifications for a portion of
    an Idris \S{Do}-block with indentation checks, \M{do} is a
    syntactic sugar for \name's monadic bind.  Specifications given in
    \textcolor{blue}{Blue} are provided by the parser writer;
    \textcolor{gray}{Gray} specifications are inferred by \name.}
  \label{fig:idris-do}
  \end{figure}

\paragraph{Specifying Data-dependent Parser Properties}
To specify an \emph{indentation-sensitivity} safety property, we first
define an inductive type for a parse-tree (\S{tree}) and refine this
type using a dependent function type, ({\sf offsideTree i}), that
specifies an indentation value for each parsed result.
\begin{lstlisting}[escapechar=\@,basicstyle=\small\sf,language=ML]
  type tree = Tree {term : pterm; indentT : int; children : tree list}
  type offsideTree i = Tree {term : pterm; @\label{line:tree}@ indentT : { v : int | v > i }; children : (offsideTree i) list}
 \end{lstlisting}
This type defines a tree with three fields:
\begin{itemize}
\item A \S{term} of type {\sf pterm}. 
\item The indentation (\S{indentT}) of a returned parse tree, the refinement constraints on {\sf indentT}
  requires its value to be greater than \S{i}.
\item A list of sub-parse trees (\S{children}) for each of the
  terminals and non-terminals in the current grammar rule's right-hand
  side, each of which must also satisfy this refinement.
\end{itemize}
\name\ additionally automatically generates {\it qualifiers} like, \S{indentT}, \S{children}, etc, for each of the
datatype's constructors and fields with the same name that can be used
in type refinements.
The type {\sf offsideTree i} is sufficient to specify pure functions
that return an indented tree, e.g.,
\begin{lstlisting}[escapechar=\@,basicstyle=\small\sf,language=ML]
  goodTree : (i : int) -> offsideTree i
\end{lstlisting}
However, such types are not sufficiently expressive to specify
stateful properties of the kind exploited in our example program. For
example, using this type, we cannot specify the required safety
property for {\sf doBlock} that requires ``{\it the indentation of the
  parse tree returned by {\sf doBlock} must be greater than the current
  value of \S{ist}}'' because \S{ist} is an effectful heap variable.

We can specify a safety property for a {\sf doBlock} parser as shown
on line~\ref{line:doBlockSpec} in Figure~\ref{fig:idris-do}.  The type
specification in \textcolor{blue}{blue} are provided by the
programmer.  The type should be understood as follows: The effect label (\S{stexc}) defines that the possible effects produced by the parser include \S{state} and \S{exc}. The
precondition binds the value of the mutable state variable {\sf ist},
a reference to the current indentation level, to {\sf I} via the use
of the built-in qualifier \S{sel} that defines a select operation on
the heap~\cite{mccarthy}.  This binding is needed even though
\textsf{I} is never used in the precondition because the type for the
return variable (\S{offsideTree\ I}) is dependent on I.  The return
  type (\S{offsideTree\ I} result) obligates the computation to return a
  parse tree (or a failure) whose indentation must be greater than {\sf I}. The
  postcondition constraints that the final value of the indentation is
  to be reset to its value prior to the parse (a \emph{reset}
  property) when the parser succeeds (case $\nu$ = Inl (\_)) or that the input stream \S{inp} is monotonically consumed when the parser fails (case $\nu$ = Inr (Err)).  The types for other parsers in the figure can be
  specified as shown at lines ~\ref{line:dospec},
  ~\ref{line:inferredindentDo}, ~\ref{line:indentspec}, etc.; these
  types shown in \textcolor{gray}{gray} are automatically inferred by
   \name's type inference algorithm. 

\paragraph{Revisiting the Bug in the Example}
The bug described in the previous paragraph is unearthed while
typechecking the {\sf indentedDo} implementation or the {\sf indentedDoBlock} implementation. We discuss the case for {\sf indentedDo} case here. To verify that
\S{doBlock} satisfies its specification, \name needs to prove
that the type inferred for the body of {\sf indentedDo}
(lines~\ref{line:idostart}-~\ref{line:l4}):

\begin{enumerate}
\item has a return type that is of the form, {\sf offsideTree I}.
  Concretely, the indentation of the returned tree must be greater
  than the initial value of {\sf ist} (i.e. indentT ($\nu$) > I).
  
\item asserts that the final value of {\sf ist} is equal to the
  initial value.
\end{enumerate}
Goal (1) is required because \textsf{indentedDo} is used by
\textsf{indentedDoBlock} (see Figure~\ref{fig:idris-original-do}),
which is then invoked by \textsf{doBlock}, where its result constructs
the value for \textsf{children}, whose type is \textsf{offsideTree I}
list.  Goal (2) is required because \S{doBlock}'s specified
post-condition demands it.
Type-checking the body for {\sf indentedDo} yields the type
shown at line~\ref{line:inferredindentDo}.
The two conjuncts in the post-condition correspond to the {\it then}
(failure case) and {\it else} (success case) branch in the parser's
body. 


The failure conjunct asserts that the input stream is consumed
monotonically if the indentation level is greater than
\textsf{ist}. The success conjunct is the post-condition of the
\textsf{do\_} parser. This inferred type is, however, too weak to
prove goal (1) given above, which requires the combinator to return a
parse tree that respects the offside rule.  The problem is that {\sf
  indent}'s type (line~\ref{line:indentspec}), inferred as:
\begin{lstlisting}[escapechar=\@,mathescape=true,basicstyle=\small\sf,language=ML]
@\texttt{\textcolor{black}{  
\quad\quad{\sf indent} : $\mathsf{PE^{state}}$
  \{true\} $\nu$\, :\, int\, \{ $\forall$ h, $\nu$, h'. sel (h', inp) $\subseteq$ sel(h, inp)
  \}}}@
\end{lstlisting}
does not allow us to conclude that
\S{indentedDo} satisfies the indentation condition demanded by
\S{doBlock}, i.e., that it returns a well-typed (\S{offsideTree}
I). This is because the type imposes no constraint between the integer
\textsf{indent} returns and the function's input heap, and thus offers no
guarantees that its result gives the position of the first lexeme of
the input list.

We can revise \textsf{indent}'s implementation such that it does not
skip any reserved operators and always returns the position of
the first element of the input list, allowing us to track the
indentation of every lexeme:
\begin{lstlisting}[escapechar=\@,basicstyle=\small\sf,language=ML]
@\scriptsize\texttt{\textcolor{black}{  
indent : $\mathsf{PE^{state}}$ \{true\} $\nu$ : int  
\{$\forall$ h, $\nu$, h'.$\nu$ = pos (sel (h, inp)) $\wedge$ sel (h', inp) $\subseteq$ sel (h, inp)\}}}@ 
let indent = @\M{do}@    
                    s <- !inp
                    return (sourceColumn s)
\end{lstlisting}
This type defines a stronger constraint, sufficient to type-check the revised implementation and raise a type error for the original.
For this example, Morpheus generated {\sf 33}
Verification Conditions (VCs) for the revised successful case and {\sf
  6} VCs for the failing case. We were able to discharge these VCs to the SMT Solver
Z3~\cite{z3}, yielding a total overall verification time of {\sf
  10.46} seconds in the successful case, and {\sf 2.06} seconds in the
case when type-checking failed.


This example highlights several key properties of \name verification:
The specification language and the type system allows verifying
interesting properties over inductive data types (e.g., the {\sf
  offsideTree} property over the parse trees).  It also allows
verifying properties dependent on state and other effects such as the
{\it input consumption} property over input streams ({\sf inp}).
Secondly, the annotation burden on the programmer is proportional to
the complexity of the top-level safety property that needs to be
checked. Finally, the similarities between the Haskell implementation
and the \name implementation minimize the idiomatic burden placed on
\name\ users.

\section{Related Work}
\label{sec:related}

\textbf{Parser Verification.}  Traditional approaches to parser
verification involve mechanization in theorem provers like Coq or
Agda~\cite{rocksalt,totalparsercomb,jason,TRX,llparser,costar,compcertParser}. These
approaches trade-off both automation and expressiveness of the grammar
they verify to prove full correctness. Consequently, these approaches
cannot verify safety properties of data-dependent parsers, the subject
of study in this paper. For instance, RockSalt~\cite{rocksalt} focuses
on regular grammars, while ~\cite{TRX,jason} present interpreters for
parsing expression grammars (without nondeterminism) and limited
semantic actions without data dependence.  Jourdan et al.~\cite{lrcoq}
gives a certifying compiler for LR(1) grammars, which translates the
grammar into a pushdown automaton and a certificate of language
equivalence between the grammar and the automaton.  More recently {\sf
  CoStar}~\cite{costar} presents a fully verified parser for the {\sf
  ALL(*)} fragment mitigating some of the limitations of the above
approaches. However, unlike \name, {\sf CoStar} does not handle
data-dependent grammars or user-defined semantic actions.

Deductive synthesis techniques for parsers like
Narcissus~\cite{narcissus} and ~\cite{everparse} focus mainly on
tag-length-payload, binary data formats.  Narcissus~\cite{narcissus}
provides a Coq framework (an \textsf{encode\_decode} tactic) that can
automatically generate correct-by-construction encoders and decoders
from a given user format input, albeit for a restricted class of
parsers.  Notably, the system is not easily extensible to complex
user-defined data-dependent formats such as the examples we discuss in
\name.  This can be attributed to the fact that the underlying
\textsf{encode\_decode} Coq tactic is complex and brittle and may require manual proofs to verify a new format.  
In contrast, \name enables useful verification capabilities for a larger class of parsers,
albeit at the expense of automatic code generation and full correctness. Writing a safe parser implementation for a user-defined
format in \name is no more difficult than manually building the parser
in any combinator framework with the user only having to provide an additional safety specification. 
EverParse~\cite{everparse} likewise focuses mainly on
binary data formats, guaranteeing full-parser correctness, albeit with
some expressivity limitations. For example, it does not support user-defined semantic actions or global data-dependences for general
data formats. Compared to these efforts, the properties \name\ can
validate are more high-level and general. E.g., ``non-overlapping of two lists of
strings'' in a C-decl parser; ``layout-sensitivity
properties'', etc,. Verifying these properties requires
reasoning over a challenging combination of rich algebraic data types,
mutable states, and higher-order functions.

~\cite{krishnaswami} also explore types for parsing, defining a core
type-system for context-free expressions. However, their goals are
orthogonal to \name\ and are targeted towards identifying expressions
that can be parsed unambiguously.

\textbf{Data-dependent and Stateful Parsers} \name allows writing
parsers for data-dependent and stateful parsers. There is a long line
of work aimed at writing such
parsers~\cite{yakker,indentation1,AfroozehAli2015Optr,stateful-parser}.
None of these efforts, however, provide a mechanism to reason
about the parsers they can express. Further, many of these
systems are specialized for a particular class/domain of problems, such
as ~\cite{yakker} for data-dependent grammars with trivial semantic
actions, or ~\cite{indentation1} for indentation sensitive grammars,
etc.  \name\ is sufficiently expressive to both write parsers and
grammars discussed in many of these approaches, as well as verifying
interesting safety properties.  Indeed, several of our benchmarks are
selected from these works.  In contrast, systems such as~\cite{yakker}
argue about the correctness of the input parsed against the underlying
CFG, a property challenging to define and verify as a \name safety
property, beyond simple string-patterns and regular expressions.  We
leave the expression of such grammar-related properties in \name\ as a
subject for future work.


\textbf{Refinement Types.} Our specification language and type system
builds over a refinement type system developed for functional
languages like Liquid Types~\cite{liquidoriginal} or Liquid
Haskell~\cite{VS+14}.  Extending Liquid Types with {\it bounds}
~\cite{liquidextended} provides some of the capabilities required to
realize data-dependent parsing actions, but it is non-trivial to
generalize such an abstraction to complex parser combinators found in
\name with multiple effects and local reasoning over states and
effects.

\textbf{Effectful Verification}  Our work is
also closely related to dependent-type-based verification approaches
for effectful programs based on monads indexed with either pre- and
post-conditions~\cite{htt,ynot} or more recently, predicate monads
capturing the weakest pre-condition semantics for effectful
computations~\cite{fstar}. As we have illustrated
earlier, the use of expressive and general dependent types, while
enabling the ability to write rich specifications (certainly richer
than what can be expressed in \name), complicates the ability to
realize a fully automated verification pathway. 

Verification using natural proofs~\cite{natural} is based on a
mechanism in which a fixed set of proof tactics are used to reason
about a set of safety properties; automation is achieved via a search
procedure over in this set.  This idea is orthogonal to our approach
where we rather utilize the restricted domain of parsers to remain in
a decidable realm.  Both our effort and these are obviously
incomplete.  Another line of work verifying effectful specifications
use characteristic formulae~\cite{characteristic}; although more
expressive than \name types, these techniques do not lend themselves
to automation.

%

\textbf{Local Reasoning over Heaps} Our approach to controlling
aliasing is distinguished from substructural typing techniques such as
the ownership type system found in Rust~\cite{rustbelt}.  Such type
systems provide a much richer and more expressive framework to reason
about memory and effects, and can provide useful guarantees like
memory safety and data-race freedom etc.  Since our DSL is targeted at
parser combinator programs which generally operate over a
much simplified memory abstraction, we found it unnecessary to
incorporate the additional complexity such systems introduce.  The
integration of these richer systems within a refinement type framework
system of the kind provided in \name is a subject we leave for future
work.



\textbf{Parser Combinators} There is a long line of work implementing
Parser Combinator Libraries and DSLs in different
languages~\cite{parsecclones}. These also include those which provide
a principled way for writing stateful parsers using these
libraries~\cite{indentation1,stateful-parser}.  As we have discussed,
none of these libraries provide an automated verification machinery to
reason about safety properties of the parsers.  However, since they
allow the full expressive power of the host language, they may, in
some instances, be more expressive than \name.  For example,
\name\ does not allow arbitrary user-defined higher-order functions
and builds only on the core API discussed earlier. This may require a
more intricate definition for some parsers compared to traditional
libraries.  For example, traditional parser combinator libraries
typically define a higher-order combinator like {\sf
  many\_fold\_apply} with the following signature and use this
combinator to concisely define a {\it Kleene-star} parser:
\begin{lstlisting}[escapechar=\@,basicstyle=\small\sf,language=ML]
   many_fold_apply : f : ('b -> 'a -> 'b) -> (a : 'a) -> (g : 'a -> 'a) -> p : ('a, 's) t -> ('b, 's) t
   let many p = many_fold_apply (fun xs x -> x :: xs) [] List.rev p
\end{lstlisting}
Contrary to this, in \name, we need to define Kleene-star using
a more complex, lower-level fixpoint combinator.

 
\section{Conclusions}
\label{sec:conc}

This paper presents \name, a deeply-embedded DSL in OCaml that offers
a restricted language of composable effectful computations tailored
for parsing and semantic actions and a rich specification language
used to define safety properties over the constituent parsers
comprising a program.  \name\ is equipped with a rich refinement
type-based automated verification pathway.  We demonstrate \name's
utility by using it to implement a number of challenging parsing
applications, validating its ability to verify non-trivial correctness
properties in these benchmarks.

\bibliographystyle{ACM-Reference-Format}
\citestyle{acmauthoryear}
\bibliography{paper}


\begin{thebibliography}{51}


\ifx \showCODEN    \undefined \def \showCODEN     #1{\unskip}     \fi
\ifx \showDOI      \undefined \def \showDOI       #1{#1}\fi
\ifx \showISBNx    \undefined \def \showISBNx     #1{\unskip}     \fi
\ifx \showISBNxiii \undefined \def \showISBNxiii  #1{\unskip}     \fi
\ifx \showISSN     \undefined \def \showISSN      #1{\unskip}     \fi
\ifx \showLCCN     \undefined \def \showLCCN      #1{\unskip}     \fi
\ifx \shownote     \undefined \def \shownote      #1{#1}          \fi
\ifx \showarticletitle \undefined \def \showarticletitle #1{#1}   \fi
\ifx \showURL      \undefined \def \showURL       {\relax}        \fi
\providecommand\bibfield[2]{#2}
\providecommand\bibinfo[2]{#2}
\providecommand\natexlab[1]{#1}
\providecommand\showeprint[2][]{arXiv:#2}

\bibitem[Adams and A\u{g}acan(2014)]%
        {indentation1}
\bibfield{author}{\bibinfo{person}{Michael~D. Adams} {and}
  \bibinfo{person}{\"{O}mer~S. A\u{g}acan}.} \bibinfo{year}{2014}\natexlab{}.
\newblock \showarticletitle{Indentation-Sensitive Parsing for Parsec}.
  \bibinfo{howpublished}{https://doi.org/10.1145/2775050.2633369}, In
  \bibinfo{booktitle}{SIGPLAN Notices}.
\newblock \bibinfo{journal}{\emph{SIGPLAN Not.}} \bibinfo{volume}{49},
  \bibinfo{number}{12}, \bibinfo{pages}{121--132}.
\newblock
\showISSN{0362-1340}
\urldef\tempurl%
\url{https://doi.org/10.1145/2775050.2633369}
\showDOI{\tempurl}


\bibitem[Afroozeh and Izmaylova(2015a)]%
        {AfroozehAli2015Optr}
\bibfield{author}{\bibinfo{person}{Ali Afroozeh} {and}
  \bibinfo{person}{Anastasia Izmaylova}.} \bibinfo{year}{2015}\natexlab{a}.
\newblock \showarticletitle{One parser to rule them all}. In
  \bibinfo{booktitle}{\emph{2015 ACM International Symposium on new ideas, new
  paradigms, and reflections on programming and software (onward!)}}
  \emph{(\bibinfo{series}{Onward! 2015})}. \bibinfo{publisher}{ACM},
  \bibinfo{pages}{151--170}.
\newblock
\showISBNx{9781450336888}


\bibitem[Afroozeh and Izmaylova(2015b)]%
        {indentation2}
\bibfield{author}{\bibinfo{person}{Ali Afroozeh} {and}
  \bibinfo{person}{Anastasia Izmaylova}.} \bibinfo{year}{2015}\natexlab{b}.
\newblock \showarticletitle{One Parser to Rule Them All}.
  \bibinfo{howpublished}{https://doi.org/10.1145/2814228.2814242}. In
  \bibinfo{booktitle}{\emph{2015 ACM International Symposium on New Ideas, New
  Paradigms, and Reflections on Programming and Software (Onward!)}}
  (Pittsburgh, PA, USA) \emph{(\bibinfo{series}{Onward! 2015})}.
  \bibinfo{publisher}{Association for Computing Machinery},
  \bibinfo{address}{New York, NY, USA}, \bibinfo{pages}{151--170}.
\newblock
\showISBNx{9781450336888}
\urldef\tempurl%
\url{https://doi.org/10.1145/2814228.2814242}
\showDOI{\tempurl}


\bibitem[Angstrom(2021)]%
        {angstrom}
\bibfield{author}{\bibinfo{person}{Angstrom}.} \bibinfo{year}{2021}\natexlab{}.
\newblock \bibinfo{title}{Angstrom parser-combinator library}.
\newblock \bibinfo{howpublished}{https://github.com/inhabitedtype/angstrom}.
\newblock


\bibitem[Brady(2014)]%
        {idris}
\bibfield{author}{\bibinfo{person}{Edwin Brady}.}
  \bibinfo{year}{2014}\natexlab{}.
\newblock \showarticletitle{Idris: Implementing a Dependently Typed Programming
  Language}. \bibinfo{howpublished}{https://doi.org/10.1145/2631172.2631174}.
  In \bibinfo{booktitle}{\emph{Proceedings of the 2014 International Workshop
  on Logical Frameworks and Meta-Languages: Theory and Practice}} (Vienna,
  Austria) \emph{(\bibinfo{series}{LFMTP '14})}.
  \bibinfo{publisher}{Association for Computing Machinery},
  \bibinfo{address}{New York, NY, USA}, Article \bibinfo{articleno}{2},
  \bibinfo{numpages}{1}~pages.
\newblock
\showISBNx{9781450328173}
\urldef\tempurl%
\url{https://doi.org/10.1145/2631172.2631174}
\showDOI{\tempurl}


\bibitem[Chargu\'{e}raud(2011)]%
        {characteristic}
\bibfield{author}{\bibinfo{person}{Arthur Chargu\'{e}raud}.}
  \bibinfo{year}{2011}\natexlab{}.
\newblock \showarticletitle{Characteristic Formulae for the Verification of
  Imperative Programs}.
\newblock \bibinfo{howpublished}{https://doi.org/10.1145/2034574.2034828}.
\newblock \bibinfo{journal}{\emph{SIGPLAN Not.}} \bibinfo{volume}{46},
  \bibinfo{number}{9} (\bibinfo{date}{sep} \bibinfo{year}{2011}),
  \bibinfo{pages}{418–430}.
\newblock
\showISSN{0362-1340}
\urldef\tempurl%
\url{https://doi.org/10.1145/2034574.2034828}
\showDOI{\tempurl}


\bibitem[Danielsson(2010)]%
        {totalparsercomb}
\bibfield{author}{\bibinfo{person}{Nils~Anders Danielsson}.}
  \bibinfo{year}{2010}\natexlab{}.
\newblock \showarticletitle{Total Parser Combinators}. In
  \bibinfo{booktitle}{\emph{Proceedings of the 15th ACM SIGPLAN International
  Conference on Functional Programming}} (Baltimore, Maryland, USA)
  \emph{(\bibinfo{series}{ICFP '10})}. \bibinfo{publisher}{Association for
  Computing Machinery}, \bibinfo{address}{New York, NY, USA},
  \bibinfo{pages}{285–296}.
\newblock
\showISBNx{9781605587943}
\urldef\tempurl%
\url{https://doi.org/10.1145/1863543.1863585}
\showDOI{\tempurl}


\bibitem[de~Moura and Bj{\o}rner(2008)]%
        {z3}
\bibfield{author}{\bibinfo{person}{Leonardo de Moura} {and}
  \bibinfo{person}{Nikolaj Bj{\o}rner}.} \bibinfo{year}{2008}\natexlab{}.
\newblock \showarticletitle{Z3: An Efficient SMT Solver}. In
  \bibinfo{booktitle}{\emph{Tools and Algorithms for the Construction and
  Analysis of Systems}}, \bibfield{editor}{\bibinfo{person}{C.~R. Ramakrishnan}
  {and} \bibinfo{person}{Jakob Rehof}} (Eds.). \bibinfo{publisher}{Springer
  Berlin Heidelberg}, \bibinfo{address}{Berlin, Heidelberg},
  \bibinfo{pages}{337--340}.
\newblock
\showISBNx{978-3-540-78800-3}


\bibitem[Delaware et~al\mbox{.}(2019)]%
        {narcissus}
\bibfield{author}{\bibinfo{person}{Benjamin Delaware}, \bibinfo{person}{Sorawit
  Suriyakarn}, \bibinfo{person}{Cl\'{e}ment Pit-Claudel},
  \bibinfo{person}{Qianchuan Ye}, {and} \bibinfo{person}{Adam Chlipala}.}
  \bibinfo{year}{2019}\natexlab{}.
\newblock \showarticletitle{Narcissus: Correct-by-Construction Derivation of
  Decoders and Encoders from Binary Formats}.
\newblock \bibinfo{howpublished}{https://doi.org/10.1145/3341686}.
\newblock \bibinfo{journal}{\emph{Proc. ACM Program. Lang.}}
  \bibinfo{volume}{3}, \bibinfo{number}{ICFP}, Article \bibinfo{articleno}{82}
  (\bibinfo{date}{July} \bibinfo{year}{2019}), \bibinfo{numpages}{29}~pages.
\newblock
\urldef\tempurl%
\url{https://doi.org/10.1145/3341686}
\showDOI{\tempurl}


\bibitem[DNS(1987)]%
        {dns}
\bibfield{author}{\bibinfo{person}{DNS}.} \bibinfo{year}{1987}\natexlab{}.
\newblock \bibinfo{title}{Domain Names - Implementation and Specification}.
\newblock \bibinfo{howpublished}{https://www.rfc-editor.org/rfc/rfc1035}.
\newblock
\newblock
\shownote{Network Working Group}.


\bibitem[Gross and Chlipala(2015)]%
        {jason}
\bibfield{author}{\bibinfo{person}{J. Gross} {and} \bibinfo{person}{Adam
  Chlipala}.} \bibinfo{year}{2015}\natexlab{}.
\newblock \showarticletitle{Parsing Parsers A Pearl of ( Dependently Typed )
  Programming and Proof}.
\newblock


\bibitem[HaskellWiki(2021)]%
        {parsecclones}
\bibfield{author}{\bibinfo{person}{HaskellWiki}.}
  \bibinfo{year}{2021}\natexlab{}.
\newblock \bibinfo{title}{Parsec --- HaskellWiki{,}}.
\newblock
\newblock
\urldef\tempurl%
\url{https://wiki.haskell.org/index.php?title=Parsec&oldid=64649}
\showURL{%
\tempurl}
\newblock
\shownote{[Online; accessed 7-July-2022]}.


\bibitem[Hutton and Meijer(1999)]%
        {parsinghutton}
\bibfield{author}{\bibinfo{person}{Graham Hutton} {and} \bibinfo{person}{Erik
  Meijer}.} \bibinfo{year}{1999}\natexlab{}.
\newblock \showarticletitle{Monadic Parser Combinators}.
\newblock  (\bibinfo{date}{09} \bibinfo{year}{1999}).
\newblock


\bibitem[Idris(2017)]%
        {idrisgrammar}
Idris \bibinfo{year}{2017}\natexlab{}.
\newblock \bibinfo{booktitle}{\emph{Documentation for the Idris Language}}.
\newblock
\urldef\tempurl%
\url{https://docs.idris-lang.org/en/latest/index.html}
\showURL{%
\tempurl}


\bibitem[Jim et~al\mbox{.}(2010)]%
        {yakker}
\bibfield{author}{\bibinfo{person}{Trevor Jim}, \bibinfo{person}{Yitzhak
  Mandelbaum}, {and} \bibinfo{person}{David Walker}.}
  \bibinfo{year}{2010}\natexlab{}.
\newblock \showarticletitle{Semantics and Algorithms for Data-Dependent
  Grammars}. \bibinfo{howpublished}{https://doi.org/10.1145/1706299.1706347}.
  In \bibinfo{booktitle}{\emph{Proceedings of the 37th Annual ACM
  SIGPLAN-SIGACT Symposium on Principles of Programming Languages}} (Madrid,
  Spain) \emph{(\bibinfo{series}{POPL '10})}. \bibinfo{publisher}{Association
  for Computing Machinery}, \bibinfo{address}{New York, NY, USA},
  \bibinfo{pages}{417--430}.
\newblock
\showISBNx{9781605584799}
\urldef\tempurl%
\url{https://doi.org/10.1145/1706299.1706347}
\showDOI{\tempurl}


\bibitem[Jourdan et~al\mbox{.}(2012a)]%
        {compcertParser}
\bibfield{author}{\bibinfo{person}{Jacques-Henri Jourdan},
  \bibinfo{person}{Fran{\c{c}}ois Pottier}, {and} \bibinfo{person}{Xavier
  Leroy}.} \bibinfo{year}{2012}\natexlab{a}.
\newblock \showarticletitle{Validating LR(1) Parsers}. In
  \bibinfo{booktitle}{\emph{Programming Languages and Systems}},
  \bibfield{editor}{\bibinfo{person}{Helmut Seidl}} (Ed.).
  \bibinfo{publisher}{Springer Berlin Heidelberg}, \bibinfo{address}{Berlin,
  Heidelberg}, \bibinfo{pages}{397--416}.
\newblock
\showISBNx{978-3-642-28869-2}


\bibitem[Jourdan et~al\mbox{.}(2012b)]%
        {lrcoq}
\bibfield{author}{\bibinfo{person}{Jacques-Henri Jourdan},
  \bibinfo{person}{Fran{\c{c}}ois Pottier}, {and} \bibinfo{person}{Xavier
  Leroy}.} \bibinfo{year}{2012}\natexlab{b}.
\newblock \showarticletitle{Validating LR(1) Parsers}. In
  \bibinfo{booktitle}{\emph{Programming Languages and Systems}},
  \bibfield{editor}{\bibinfo{person}{Helmut Seidl}} (Ed.).
  \bibinfo{publisher}{Springer Berlin Heidelberg}, \bibinfo{address}{Berlin,
  Heidelberg}, \bibinfo{pages}{397--416}.
\newblock
\showISBNx{978-3-642-28869-2}


\bibitem[Jung et~al\mbox{.}(2017)]%
        {rustbelt}
\bibfield{author}{\bibinfo{person}{Ralf Jung}, \bibinfo{person}{Jacques-Henri
  Jourdan}, \bibinfo{person}{Robbert Krebbers}, {and} \bibinfo{person}{Derek
  Dreyer}.} \bibinfo{year}{2017}\natexlab{}.
\newblock \showarticletitle{RustBelt: Securing the Foundations of the Rust
  Programming Language}.
\newblock \bibinfo{howpublished}{https://doi.org/10.1145/3158154}.
\newblock \bibinfo{journal}{\emph{Proc. ACM Program. Lang.}}
  \bibinfo{volume}{2}, \bibinfo{number}{POPL}, Article \bibinfo{articleno}{66}
  (\bibinfo{date}{dec} \bibinfo{year}{2017}), \bibinfo{numpages}{34}~pages.
\newblock
\urldef\tempurl%
\url{https://doi.org/10.1145/3158154}
\showDOI{\tempurl}


\bibitem[Kaki and Jagannathan(2014)]%
        {relref}
\bibfield{author}{\bibinfo{person}{Gowtham Kaki} {and} \bibinfo{person}{Suresh
  Jagannathan}.} \bibinfo{year}{2014}\natexlab{}.
\newblock \showarticletitle{A Relational Framework for Higher-Order Shape
  Analysis}. \bibinfo{howpublished}{https://doi.org/10.1145/2628136.2628159}.
  In \bibinfo{booktitle}{\emph{Proceedings of the 19th ACM SIGPLAN
  International Conference on Functional Programming}} (Gothenburg, Sweden)
  \emph{(\bibinfo{series}{ICFP '14})}. \bibinfo{publisher}{Association for
  Computing Machinery}, \bibinfo{address}{New York, NY, USA},
  \bibinfo{pages}{311--324}.
\newblock
\showISBNx{9781450328739}
\urldef\tempurl%
\url{https://doi.org/10.1145/2628136.2628159}
\showDOI{\tempurl}


\bibitem[Karpov(2022)]%
        {megaparsec}
\bibfield{author}{\bibinfo{person}{Mark Karpov}.}
  \bibinfo{year}{2022}\natexlab{}.
\newblock \bibinfo{title}{{Megaparsec: Monadic Parser Combinators}}.
\newblock \bibinfo{howpublished}{https://github.com/mrkkrp/megaparsec}.
\newblock


\bibitem[Katsumata(2014)]%
        {param-monad-effect}
\bibfield{author}{\bibinfo{person}{Shin-ya Katsumata}.}
  \bibinfo{year}{2014}\natexlab{}.
\newblock \showarticletitle{Parametric Effect Monads and Semantics of Effect
  Systems}. \bibinfo{howpublished}{https://doi.org/10.1145/2535838.2535846}. In
  \bibinfo{booktitle}{\emph{Proceedings of the 41st ACM SIGPLAN-SIGACT
  Symposium on Principles of Programming Languages}} (San Diego, California,
  USA) \emph{(\bibinfo{series}{POPL '14})}. \bibinfo{publisher}{Association for
  Computing Machinery}, \bibinfo{address}{New York, NY, USA},
  \bibinfo{pages}{633–645}.
\newblock
\showISBNx{9781450325448}
\urldef\tempurl%
\url{https://doi.org/10.1145/2535838.2535846}
\showDOI{\tempurl}


\bibitem[Koprowski and Binsztok(2010)]%
        {TRX}
\bibfield{author}{\bibinfo{person}{Adam Koprowski} {and} \bibinfo{person}{Henri
  Binsztok}.} \bibinfo{year}{2010}\natexlab{}.
\newblock \showarticletitle{TRX: A Formally Verified Parser editor={Gordon,
  Andrew D.}, Interpreter}. In \bibinfo{booktitle}{\emph{Programming Languages
  and Systems}}. \bibinfo{publisher}{Springer Berlin Heidelberg},
  \bibinfo{address}{Berlin, Heidelberg}, \bibinfo{pages}{345--365}.
\newblock
\showISBNx{978-3-642-11957-6}


\bibitem[Krishnaswami and Yallop(2019)]%
        {krishnaswami}
\bibfield{author}{\bibinfo{person}{Neelakantan Krishnaswami} {and}
  \bibinfo{person}{Jeremy Yallop}.} \bibinfo{year}{2019}\natexlab{}.
\newblock \showarticletitle{A typed, algebraic approach to parsing}. In
  \bibinfo{booktitle}{\emph{Proceedings of the 40th ACM SIGPLAN Conference on
  Programming Language Design and Implementation}}. \bibinfo{pages}{379--393}.
\newblock
\urldef\tempurl%
\url{https://doi.org/10.1145/3314221.3314625}
\showDOI{\tempurl}


\bibitem[Lasser et~al\mbox{.}(2021)]%
        {costar}
\bibfield{author}{\bibinfo{person}{Sam Lasser}, \bibinfo{person}{Chris
  Casinghino}, \bibinfo{person}{Kathleen Fisher}, {and} \bibinfo{person}{Cody
  Roux}.} \bibinfo{year}{2021}\natexlab{}.
\newblock \showarticletitle{CoStar: A Verified ALL(*) Parser}. In
  \bibinfo{booktitle}{\emph{Proceedings of the 42nd ACM SIGPLAN International
  Conference on Programming Language Design and Implementation}} (Virtual,
  Canada) \emph{(\bibinfo{series}{PLDI 2021})}. \bibinfo{publisher}{Association
  for Computing Machinery}, \bibinfo{address}{New York, NY, USA},
  \bibinfo{pages}{420–434}.
\newblock
\showISBNx{9781450383912}
\urldef\tempurl%
\url{https://doi.org/10.1145/3453483.3454053}
\showDOI{\tempurl}


\bibitem[Laurent and Mens(2016)]%
        {stateful-parser}
\bibfield{author}{\bibinfo{person}{Nicolas Laurent} {and} \bibinfo{person}{Kim
  Mens}.} \bibinfo{year}{2016}\natexlab{}.
\newblock \showarticletitle{Taming Context-Sensitive Languages with Principled
  Stateful Parsing}.
  \bibinfo{howpublished}{https://doi.org/10.1145/2997364.2997370}. In
  \bibinfo{booktitle}{\emph{Proceedings of the 2016 ACM SIGPLAN International
  Conference on Software Language Engineering}} (Amsterdam, Netherlands)
  \emph{(\bibinfo{series}{SLE 2016})}. \bibinfo{publisher}{Association for
  Computing Machinery}, \bibinfo{address}{New York, NY, USA},
  \bibinfo{pages}{15–27}.
\newblock
\showISBNx{9781450344470}
\urldef\tempurl%
\url{https://doi.org/10.1145/2997364.2997370}
\showDOI{\tempurl}


\bibitem[Leijen and Meijer(2001)]%
        {parsec}
\bibfield{author}{\bibinfo{person}{Daan Leijen} {and} \bibinfo{person}{Erik
  Meijer}.} \bibinfo{year}{2001}\natexlab{}.
\newblock \bibinfo{booktitle}{\emph{Parsec: Direct Style Monadic Parser
  Combinators for the Real World} (\bibinfo{edition}{technical report
  uu-cs-2001-35, departement of computer science, universiteit utrecht} ed.)}.
\newblock \bibinfo{type}{{T}echnical {R}eport} UU-CS-2001-27.
\newblock
\urldef\tempurl%
\url{https://www.microsoft.com/en-us/research/publication/parsec-direct-style-monadic-parser-combinators-for-the-real-world/}
\showURL{%
\tempurl}
\newblock
\shownote{User Modeling 2007, 11th International Conference, UM 2007, Corfu,
  Greece, June 25-29, 2007}.


\bibitem[Liang et~al\mbox{.}(1995)]%
        {monadx}
\bibfield{author}{\bibinfo{person}{Sheng Liang}, \bibinfo{person}{Paul Hudak},
  {and} \bibinfo{person}{Mark Jones}.} \bibinfo{year}{1995}\natexlab{}.
\newblock \showarticletitle{Monad Transformers and Modular Interpreters}.
  \bibinfo{howpublished}{https://doi.org/10.1145/199448.199528}. In
  \bibinfo{booktitle}{\emph{Proceedings of the 22nd ACM SIGPLAN-SIGACT
  Symposium on Principles of Programming Languages}} (San Francisco,
  California, USA) \emph{(\bibinfo{series}{POPL '95})}.
  \bibinfo{publisher}{Association for Computing Machinery},
  \bibinfo{address}{New York, NY, USA}, \bibinfo{pages}{333–343}.
\newblock
\showISBNx{0897916921}
\urldef\tempurl%
\url{https://doi.org/10.1145/199448.199528}
\showDOI{\tempurl}


\bibitem[Marlow(2010)]%
        {marlow}
\bibfield{author}{\bibinfo{person}{Simon Marlow}.}
  \bibinfo{year}{2010}\natexlab{}.
\newblock \bibinfo{title}{Haskell 2010 Language Report}.
\newblock
  \bibinfo{howpublished}{https://www.haskell.org/onlinereport/haskell2010/}.
\newblock


\bibitem[McCarthy(1993)]%
        {mccarthy}
\bibfield{author}{\bibinfo{person}{J. McCarthy}.}
  \bibinfo{year}{1993}\natexlab{}.
\newblock \bibinfo{booktitle}{\emph{Towards a Mathematical Science of
  Computation}}.
\newblock \bibinfo{publisher}{Springer Netherlands},
  \bibinfo{address}{Dordrecht}, \bibinfo{pages}{35--56}.
\newblock
\showISBNx{978-94-011-1793-7}
\urldef\tempurl%
\url{https://doi.org/10.1007/978-94-011-1793-7_2}
\showDOI{\tempurl}


\bibitem[Morrisett et~al\mbox{.}(2012)]%
        {rocksalt}
\bibfield{author}{\bibinfo{person}{Greg Morrisett}, \bibinfo{person}{Gang Tan},
  \bibinfo{person}{Joseph Tassarotti}, \bibinfo{person}{Jean-Baptiste Tristan},
  {and} \bibinfo{person}{Edward Gan}.} \bibinfo{year}{2012}\natexlab{}.
\newblock \showarticletitle{RockSalt: Better, Faster, Stronger SFI for the
  X86}. \bibinfo{howpublished}{https://doi.org/10.1145/2254064.2254111}. In
  \bibinfo{booktitle}{\emph{Proceedings of the 33rd ACM SIGPLAN Conference on
  Programming Language Design and Implementation}} (Beijing, China)
  \emph{(\bibinfo{series}{PLDI '12})}. \bibinfo{publisher}{Association for
  Computing Machinery}, \bibinfo{address}{New York, NY, USA},
  \bibinfo{pages}{395--404}.
\newblock
\showISBNx{9781450312059}
\urldef\tempurl%
\url{https://doi.org/10.1145/2254064.2254111}
\showDOI{\tempurl}


\bibitem[Murato(2021)]%
        {mparser}
\bibfield{author}{\bibinfo{person}{Max Murato}.}
  \bibinfo{year}{2021}\natexlab{}.
\newblock \bibinfo{title}{{MParser, A Simple Monadic Parser Combinator
  Library}}.
\newblock \bibinfo{howpublished}{https://github.com/murmour/mparser}.
\newblock


\bibitem[Nanevski et~al\mbox{.}(2006)]%
        {htt}
\bibfield{author}{\bibinfo{person}{Aleksandar Nanevski}, \bibinfo{person}{Greg
  Morrisett}, {and} \bibinfo{person}{Lars Birkedal}.}
  \bibinfo{year}{2006}\natexlab{}.
\newblock \showarticletitle{Polymorphism and Separation in Hoare Type Theory}.
\newblock \bibinfo{howpublished}{https://doi.org/10.1145/1160074.1159812}.
\newblock \bibinfo{journal}{\emph{SIGPLAN Not.}} \bibinfo{volume}{41},
  \bibinfo{number}{9} (\bibinfo{date}{Sept.} \bibinfo{year}{2006}),
  \bibinfo{pages}{62–73}.
\newblock
\showISSN{0362-1340}
\urldef\tempurl%
\url{https://doi.org/10.1145/1160074.1159812}
\showDOI{\tempurl}


\bibitem[Nanevski et~al\mbox{.}(2008)]%
        {ynot}
\bibfield{author}{\bibinfo{person}{Aleksandar Nanevski}, \bibinfo{person}{Greg
  Morrisett}, \bibinfo{person}{Avraham Shinnar}, \bibinfo{person}{Paul
  Govereau}, {and} \bibinfo{person}{Lars Birkedal}.}
  \bibinfo{year}{2008}\natexlab{}.
\newblock \showarticletitle{Ynot: Dependent Types for Imperative Programs}.
\newblock \bibinfo{howpublished}{https://doi.org/10.1145/1411203.1411237}.
\newblock \bibinfo{journal}{\emph{SIGPLAN Not.}} \bibinfo{volume}{43},
  \bibinfo{number}{9} (\bibinfo{date}{Sept.} \bibinfo{year}{2008}),
  \bibinfo{pages}{229--240}.
\newblock
\showISSN{0362-1340}
\urldef\tempurl%
\url{https://doi.org/10.1145/1411203.1411237}
\showDOI{\tempurl}


\bibitem[Nelson(1980)]%
        {eufa}
\bibfield{author}{\bibinfo{person}{Charles~Gregory Nelson}.}
  \bibinfo{year}{1980}\natexlab{}.
\newblock \emph{\bibinfo{title}{Techniques for Program Verification}}.
\newblock \bibinfo{thesistype}{Ph.\,D. Dissertation}.
  \bibinfo{address}{Stanford, CA, USA}.
\newblock
\newblock
\shownote{AAI8011683}.


\bibitem[Patterson(2015)]%
        {hammer}
\bibfield{author}{\bibinfo{person}{Meredith~L. Patterson}.}
  \bibinfo{year}{2015}\natexlab{}.
\newblock \bibinfo{title}{Hammer Primer}.
\newblock \bibinfo{howpublished}{https://github.com/sergeybratus/HammerPrimer}.
\newblock


\bibitem[PDF(2008)]%
        {pdfreference}
\bibfield{author}{\bibinfo{person}{PDF}.} \bibinfo{year}{2008}\natexlab{}.
\newblock \bibinfo{title}{ISO 32000 (PDF)}.
\newblock
  \bibinfo{howpublished}{https://www.pdfa.org/resource/iso-32000-pdf/pdf-2}.
\newblock
\newblock
\shownote{PDF Association}.


\bibitem[Piskac et~al\mbox{.}(2008)]%
        {smt-eps}
\bibfield{author}{\bibinfo{person}{Ruzica Piskac}, \bibinfo{person}{Leonardo de
  Moura}, {and} \bibinfo{person}{Nikolaj Bjørner}.}
  \bibinfo{year}{2008}\natexlab{}.
\newblock \bibinfo{booktitle}{\emph{Deciding Effectively Propositional Logic
  with Equality}}.
\newblock \bibinfo{type}{{T}echnical {R}eport} MSR-TR-2008-181.
  \bibinfo{pages}{25} pages.
\newblock


\bibitem[PKWare(2020)]%
        {zip}
\bibfield{author}{\bibinfo{person}{PKWare}.} \bibinfo{year}{2020}\natexlab{}.
\newblock \bibinfo{title}{\.ZIP File Format Specification}.
\newblock
  \bibinfo{howpublished}{https://pkware.cachefly.net/webdocs/casestudies/APPNOTE.TXT}.
\newblock


\bibitem[Qiu et~al\mbox{.}(2013)]%
        {natural}
\bibfield{author}{\bibinfo{person}{Xiaokang Qiu}, \bibinfo{person}{Pranav
  Garg}, \bibinfo{person}{Andrei \c{S}tef\u{a}nescu}, {and}
  \bibinfo{person}{Parthasarathy Madhusudan}.} \bibinfo{year}{2013}\natexlab{}.
\newblock \showarticletitle{Natural Proofs for Structure, Data, and
  Separation}.
\newblock \bibinfo{howpublished}{https://doi.org/10.1145/2499370.2462169}.
\newblock \bibinfo{journal}{\emph{SIGPLAN Not.}} \bibinfo{volume}{48},
  \bibinfo{number}{6} (\bibinfo{date}{jun} \bibinfo{year}{2013}),
  \bibinfo{pages}{231–242}.
\newblock
\showISSN{0362-1340}
\urldef\tempurl%
\url{https://doi.org/10.1145/2499370.2462169}
\showDOI{\tempurl}


\bibitem[Ramananandro et~al\mbox{.}(2019)]%
        {everparse}
\bibfield{author}{\bibinfo{person}{Tahina Ramananandro},
  \bibinfo{person}{Antoine Delignat-Lavaud}, \bibinfo{person}{C\'{e}dric
  Fournet}, \bibinfo{person}{Nikhil Swamy}, \bibinfo{person}{Tej Chajed},
  \bibinfo{person}{Nadim Kobeissi}, {and} \bibinfo{person}{Jonathan
  Protzenko}.} \bibinfo{year}{2019}\natexlab{}.
\newblock \showarticletitle{Everparse: Verified Secure Zero-Copy Parsers for
  Authenticated Message Formats}. In \bibinfo{booktitle}{\emph{Proceedings of
  the 28th USENIX Conference on Security Symposium}} (Santa Clara, CA, USA)
  \emph{(\bibinfo{series}{SEC'19})}. \bibinfo{publisher}{USENIX Association},
  \bibinfo{address}{USA}, \bibinfo{pages}{1465--1482}.
\newblock
\showISBNx{9781939133069}


\bibitem[Ramsey(1930)]%
        {epr}
\bibfield{author}{\bibinfo{person}{F.~P. Ramsey}.}
  \bibinfo{year}{1930}\natexlab{}.
\newblock \showarticletitle{On a Problem of Formal Logic}.
\newblock
  \bibinfo{howpublished}{https://londmathsoc.onlinelibrary.wiley.com/doi/abs/10.1112/plms/s2-30.1.264}.
\newblock \bibinfo{journal}{\emph{Proceedings of the London Mathematical
  Society}} \bibinfo{volume}{s2-30}, \bibinfo{number}{1}
  (\bibinfo{year}{1930}), \bibinfo{pages}{264--286}.
\newblock
\urldef\tempurl%
\url{https://doi.org/10.1112/plms/s2-30.1.264}
\showDOI{\tempurl}
\showeprint{https://londmathsoc.onlinelibrary.wiley.com/doi/pdf/10.1112/plms/s2-30.1.264}


\bibitem[Rondon et~al\mbox{.}(2008)]%
        {liquidoriginal}
\bibfield{author}{\bibinfo{person}{Patrick~M. Rondon}, \bibinfo{person}{Ming
  Kawaguci}, {and} \bibinfo{person}{Ranjit Jhala}.}
  \bibinfo{year}{2008}\natexlab{}.
\newblock \showarticletitle{Liquid Types}.
  \bibinfo{howpublished}{https://doi.org/10.1145/1375581.1375602}. In
  \bibinfo{booktitle}{\emph{Proceedings of the 29th ACM SIGPLAN Conference on
  Programming Language Design and Implementation}} (Tucson, AZ, USA)
  \emph{(\bibinfo{series}{PLDI '08})}. \bibinfo{publisher}{Association for
  Computing Machinery}, \bibinfo{address}{New York, NY, USA},
  \bibinfo{pages}{159--169}.
\newblock
\showISBNx{9781595938602}
\urldef\tempurl%
\url{https://doi.org/10.1145/1375581.1375602}
\showDOI{\tempurl}


\bibitem[Sam~Lasser and Roux(2019)]%
        {llparser}
\bibfield{author}{\bibinfo{person}{Kathleen~Fisher Sam~Lasser,
  Chris~Casinghino} {and} \bibinfo{person}{Cody Roux}.}
  \bibinfo{year}{2019}\natexlab{}.
\newblock \showarticletitle{A Verified LL(1) Parser Generator}. In
  \bibinfo{booktitle}{\emph{ITP}}.
\newblock


\bibitem[Schulte(2008)]%
        {vcc}
\bibfield{author}{\bibinfo{person}{Wolfram Schulte}.}
  \bibinfo{year}{2008}\natexlab{}.
\newblock \showarticletitle{VCC: Contract-based Modular Verification of
  Concurrent C}.
  \bibinfo{howpublished}{https://www.microsoft.com/en-us/research/publication/vcc-contract-based-modular-verification-of-concurrent-c/}.
  In \bibinfo{booktitle}{\emph{31st International Conference on Software
  Engineering, ICSE 2009} (\bibinfo{edition}{31st international conference on
  software engineering, icse 2009} ed.)}. \bibinfo{publisher}{IEEE Computer
  Society}.
\newblock


\bibitem[Swamy et~al\mbox{.}(2011)]%
        {lightweightmonadicML}
\bibfield{author}{\bibinfo{person}{Nikhil Swamy}, \bibinfo{person}{Nataliya
  Guts}, \bibinfo{person}{Daan Leijen}, {and} \bibinfo{person}{Michael Hicks}.}
  \bibinfo{year}{2011}\natexlab{}.
\newblock \showarticletitle{Lightweight Monadic Programming in ML}.
  \bibinfo{howpublished}{https://doi.org/10.1145/2034773.2034778}. In
  \bibinfo{booktitle}{\emph{Proceedings of the 16th ACM SIGPLAN International
  Conference on Functional Programming}} (Tokyo, Japan)
  \emph{(\bibinfo{series}{ICFP '11})}. \bibinfo{publisher}{Association for
  Computing Machinery}, \bibinfo{address}{New York, NY, USA},
  \bibinfo{pages}{15–27}.
\newblock
\showISBNx{9781450308656}
\urldef\tempurl%
\url{https://doi.org/10.1145/2034773.2034778}
\showDOI{\tempurl}


\bibitem[Swamy et~al\mbox{.}(2016)]%
        {multimonfstar}
\bibfield{author}{\bibinfo{person}{Nikhil Swamy},
  \bibinfo{person}{C\u{a}t\u{a}lin Hri\c{t}cu}, \bibinfo{person}{Chantal
  Keller}, \bibinfo{person}{Aseem Rastogi}, \bibinfo{person}{Antoine
  Delignat-Lavaud}, \bibinfo{person}{Simon Forest},
  \bibinfo{person}{Karthikeyan Bhargavan}, \bibinfo{person}{C\'{e}dric
  Fournet}, \bibinfo{person}{Pierre-Yves Strub}, \bibinfo{person}{Markulf
  Kohlweiss}, \bibinfo{person}{Jean-Karim Zinzindohoue}, {and}
  \bibinfo{person}{Santiago Zanella-B\'{e}guelin}.}
  \bibinfo{year}{2016}\natexlab{}.
\newblock \showarticletitle{Dependent Types and Multi-Monadic Effects in F*}.
  In \bibinfo{booktitle}{\emph{Proceedings of the 43rd Annual ACM
  SIGPLAN-SIGACT Symposium on Principles of Programming Languages}} (St.
  Petersburg, FL, USA) \emph{(\bibinfo{series}{POPL '16})}.
  \bibinfo{publisher}{Association for Computing Machinery},
  \bibinfo{address}{New York, NY, USA}, \bibinfo{pages}{256–270}.
\newblock
\showISBNx{9781450335492}
\urldef\tempurl%
\url{https://doi.org/10.1145/2837614.2837655}
\showDOI{\tempurl}


\bibitem[Swamy et~al\mbox{.}(2013)]%
        {fstar}
\bibfield{author}{\bibinfo{person}{Nikhil Swamy}, \bibinfo{person}{Joel
  Weinberger}, \bibinfo{person}{Cole Schlesinger}, \bibinfo{person}{Juan Chen},
  {and} \bibinfo{person}{Benjamin Livshits}.} \bibinfo{year}{2013}\natexlab{}.
\newblock \showarticletitle{Verifying Higher-Order Programs with the Dijkstra
  Monad}. \bibinfo{howpublished}{https://doi.org/10.1145/2491956.2491978}. In
  \bibinfo{booktitle}{\emph{Proceedings of the 34th ACM SIGPLAN Conference on
  Programming Language Design and Implementation}} (Seattle, Washington, USA)
  \emph{(\bibinfo{series}{PLDI '13})}. \bibinfo{publisher}{Association for
  Computing Machinery}, \bibinfo{address}{New York, NY, USA},
  \bibinfo{pages}{387--398}.
\newblock
\showISBNx{9781450320146}
\urldef\tempurl%
\url{https://doi.org/10.1145/2491956.2491978}
\showDOI{\tempurl}


\bibitem[Vazou et~al\mbox{.}(2015)]%
        {liquidextended}
\bibfield{author}{\bibinfo{person}{Niki Vazou}, \bibinfo{person}{Alexander
  Bakst}, {and} \bibinfo{person}{Ranjit Jhala}.}
  \bibinfo{year}{2015}\natexlab{}.
\newblock \showarticletitle{Bounded Refinement Types}.
  \bibinfo{howpublished}{https://doi.org/10.1145/2784731.2784745}. In
  \bibinfo{booktitle}{\emph{Proceedings of the 20th ACM SIGPLAN International
  Conference on Functional Programming}} (Vancouver, BC, Canada)
  \emph{(\bibinfo{series}{ICFP 2015})}. \bibinfo{publisher}{Association for
  Computing Machinery}, \bibinfo{address}{New York, NY, USA},
  \bibinfo{pages}{48--61}.
\newblock
\showISBNx{9781450336697}
\urldef\tempurl%
\url{https://doi.org/10.1145/2784731.2784745}
\showDOI{\tempurl}


\bibitem[Vazou et~al\mbox{.}(2014)]%
        {VS+14}
\bibfield{author}{\bibinfo{person}{Niki Vazou}, \bibinfo{person}{Eric~L.
  Seidel}, \bibinfo{person}{Ranjit Jhala}, \bibinfo{person}{Dimitrios
  Vytiniotis}, {and} \bibinfo{person}{Simon L.~Peyton Jones}.}
  \bibinfo{year}{2014}\natexlab{}.
\newblock \showarticletitle{{Refinement types for Haskell}}. In
  \bibinfo{booktitle}{\emph{Proceedings of the 19th {ACM} {SIGPLAN}
  international conference on Functional programming, Gothenburg, Sweden,
  September 1-3, 2014}}, \bibfield{editor}{\bibinfo{person}{Johan Jeuring}
  {and} \bibinfo{person}{Manuel M.~T. Chakravarty}} (Eds.).
  \bibinfo{publisher}{{ACM}}, \bibinfo{pages}{269--282}.
\newblock
\urldef\tempurl%
\url{https://doi.org/10.1145/2628136.2628161}
\showDOI{\tempurl}


\bibitem[Wadler(1993)]%
        {parsingwadler}
\bibfield{author}{\bibinfo{person}{Philip Wadler}.}
  \bibinfo{year}{1993}\natexlab{}.
\newblock \showarticletitle{Monads for functional programming}. In
  \bibinfo{booktitle}{\emph{Program Design Calculi}},
  \bibfield{editor}{\bibinfo{person}{Manfred Broy}} (Ed.).
  \bibinfo{publisher}{Springer Berlin Heidelberg}, \bibinfo{address}{Berlin,
  Heidelberg}, \bibinfo{pages}{233--264}.
\newblock
\showISBNx{978-3-662-02880-3}


\bibitem[Wadler and Thiemann(2003)]%
        {wadler1}
\bibfield{author}{\bibinfo{person}{Philip Wadler} {and} \bibinfo{person}{Peter
  Thiemann}.} \bibinfo{year}{2003}\natexlab{}.
\newblock \showarticletitle{The Marriage of Effects and Monads}.
\newblock \bibinfo{howpublished}{https://doi.org/10.1145/601775.601776}.
\newblock \bibinfo{journal}{\emph{ACM Trans. Comput. Logic}}
  \bibinfo{volume}{4}, \bibinfo{number}{1} (\bibinfo{date}{Jan.}
  \bibinfo{year}{2003}), \bibinfo{pages}{1--32}.
\newblock
\showISSN{1529-3785}
\urldef\tempurl%
\url{https://doi.org/10.1145/601775.601776}
\showDOI{\tempurl}


\end{thebibliography}



\appendix
\section{Supplemental Material for the Main Paper.}
\section{Evaluation Rules for base-expressions}
\begin{figure}[h]
\begin{flushleft}
\fbox{
        $(\mathcal{H}; \mathsf{\eip}) \Downarrow (\mathcal{H'}; \vp) $ 
           
} 
\bigskip
\end{flushleft}
\begin{minipage}{0.6\textwidth}
  \inference[{\sc P-deref}]{\mathcal{H} (\rp) = {\sf v}} {(\mathcal{H}; \S{deref} \ \rp) \Downarrow (\mathcal{H}; \S{\vp})} 
\end{minipage}
\hfill
\begin{minipage}{0.6\textwidth}
\inference[{\sf P}-ref]{(\mathcal{H}; \mathsf{\eip}) \Downarrow (\mathcal{H}; v) \\
			(\mathcal{H}; \B{let}\ \rp\ = \S{ref}\ \eip\ ) \Downarrow (\mathcal{H}[\rp \mapsto v]; \rp) \\
			(\mathcal{H}[\rp \mapsto v]; \eip_b) \Downarrow (\mathcal{H'}; {\sf v'})
			}{(\mathcal{H}; \B{let}\ \rp\ = \S{ref}\ \eip\ \B{in}\ \S{e}_b \ ) \Downarrow (\mathcal{H'}; {\sf v'})}  
\end{minipage}
\bigskip
\begin{minipage}{0.6\textwidth}
\inference[{\sf P-Assign}]{ (\mathcal{H}; \mathsf{\eip}) \Downarrow (\mathcal{H}; v)}
                                  {(\mathcal{H}; \rp\ := \eip) \Downarrow (\mathcal{H}[\rp \mapsto v]; v) } 
\end{minipage}\\[10pt]
\begin{minipage}{0.4\textwidth}
  \inference[{\sf P-App}]{\eip_f = \lambda (\S{x : \tau_1}).\eip\} \\
			  \\
			   (\mathcal{H}; [{\sf v}/{\sf x}]\eip \Downarrow (\mathcal{H'}; {\sf v'})}
                             {(\mathcal{H}; \eip_f\ \xp_a \ ) \Downarrow (\mathcal{H'}; {\sf v'})}
\end{minipage}
\begin{minipage}{0.4\textwidth}
  \inference[{\sf P-return}]{(\mathcal{H}; \eip) \Downarrow (\mathcal{H}; {\sf v})}
                              {(\mathcal{H}; \mathbf{return}\, \eip) \Downarrow (\mathcal{H}; {\sf v}) } 
\end{minipage}\\[10pt]
\begin{minipage}{0.4\textwidth}
  \inference[{\sf P-let}]{ (\mathcal{H}; \eip_1) \Downarrow (\mathcal{H}; {\sf v}) \\
			  (\mathcal{H}; [{\sf v }/ \xp]\eip_2) \Downarrow (\mathcal{H}; {\sf v'})}
		{(\mathcal{H};\mathbf{let}\, \xp\, =\, \eip_1\, \mathbf{in}\, \eip_2) \Downarrow (\mathcal{H}; {\sf v'})} 
\end{minipage}
\begin{minipage}{0.4\textwidth}
  \inference[{\sf P-TypApp}]{}
		{(\mathcal{H};{\Lambda \alpha. \eip}[\S{t}]) \Downarrow (\mathcal{H}; [\S{t}/ \alpha]\eip)} 
\end{minipage}\\[10pt]

\begin{minipage}{0.4\textwidth}
  \inference[{\sf P-match}]{(\mathcal{H};\eip) \Downarrow (\mathcal{H}; \S{v}) & \S{v} = \mathsf{D_i}\ 
                           \overline{\mathsf{\alpha}_{k}} \overline{\mathsf{x}_{j}} \\
                           (\mathcal{H}; \S{e_i}) \Downarrow (\mathcal{H'}; \S{v_i})}
		{(\mathcal{H};{\bf match} \ \eip \ {\bf with} \ \mathsf{D_i}\ 
                           \overline{\mathsf{\alpha}_{k}} \overline{\mathsf{x}_{j}} \rightarrow e_i) \Downarrow (\mathcal{H'}; \S{v_i})} 
\end{minipage}\\[10pt]

\begin{flushleft}
  {\bf Frame Typing Rule}\quad \fbox{\small $\Gamma \vdash$ $\eip$ : $\sigma$}

  \bigskip
  \end{flushleft}
  
\begin{minipage}{0.4\textwidth}
  \inference[{\sc T-frame}]{ \Gamma \vdash \eip : \S{PE^{\el}}\, \{ \phi \}\, \nu  :  \S{t}
  \ \{ \phi' \} & {\bf Locs (\phi_r)} \cap ({\bf Locs (\phi)} \cup {\bf Locs (\phi')}) = \emptyset}
  { \Gamma \vdash \eip : \S{PE^{\el}}\, \{ \phi_r \wedge \phi \}\, \nu  :  \S{t}\ \{ \phi_r \wedge \phi' \} }
  \end{minipage}
  
\caption{Evaluation rules for $\lambda_{sp}$ base expressions, a few trivial cases are skipped and the {\sc T-frame} rule.}
\label{fig:semantics-base}
\end{figure}

\section{Properties of Type System}

\paragraph{Soundness}
Informally the soundness theorem argues that if for some \name\ expression \eip, our type system
associates a type schema $\forall \overline{\alpha}$. $\mathsf{PE^{\el}}$ \{$\phi_1$\} $\nu$ : t \{$\phi_2$\}, then 
evaluating \eip\ in some heap $\mathcal{H}$ satisfying $\phi_1$ upon termination
produces a result of type t and a new heap $\mathcal{H'}$ satisfying $\phi_2$($\mathcal{H}$, $\nu$, $\mathcal{H'}$). 

\begin{definition}[Heap and Heap Interepretation]
A heap $\mathcal{H}$ is a concrete store mapping locations \rp\ to values.  
In order to relate it to the logical heaps ({\sf h, h'}) we use in our specification language, we define the following 
heap interinterpretation function:
  \[ [.]  = \forall h. empty (h)  \]
  \[  [\mathcal{H}, (\rp \mapsto {\sf v})]  = {\sf update}  [\mathcal{H'}]  [\mathcal{H}]  \rp \ {\sf v}  \]
  \[ [\dots (\rp \mapsto {\sf v})] = {\sf sel} [\mathcal{H}] \ \rp \ {\sf v}  \]
\end{definition}

\begin{definition}[Environment Entailment $\Gamma \models \phi$]
  Given $\Gamma$ = \dots , $\overline{\phi_i}$, the entailment of a formula $\phi$ under $\Gamma$ is defined as 
   ($\bigwedge_{i}^{} \phi_i$)  $ \implies \phi$ 
\end{definition}

Using the above definitions for the Heap interpretation and environment entailment, 
we define the following notion of {\it Well-typed} Heap analogous to the standard notion of {\it well-typed stores}.

\begin{definition}[Well typed Heap]
  A concrete heap $\mathcal{H}$ is well-typed under a $\Gamma$, written as $\Gamma \vdash$ $\mathcal{H}$
  if following two conditions hold. 
  \begin{itemize}
    \item $\forall$ \rp, (\rp\ $\mapsto$ {\sf v}) $\in$ $\mathcal{H}$ $\implies$ 
      $\Gamma$ $\models$ ({\sf dom} [$\mathcal{H}$] \ \rp $\wedge$ {\sf sel} [$\mathcal{H}$] \rp\ {\sf v})  
    \item $\forall$ \rp, (\rp\ $\mapsto$ {\sf v}) $\in$ $\mathcal{H}$, $\Gamma \vdash$ \rp\ : {\sf ref t} 
    $\implies$ $\Gamma \vdash$ {\sf v} : \{ $\nu$ : t | $\phi$ \} for some $\phi$.
  \end{itemize}
  
\end{definition}

In the theorems below, we write $\Gamma \models$ $\phi$($\mathcal{H}$) which extends the notion of semantic entailment of a formula over an abstract heap $\Gamma \models$ $\phi$ ({\sf h}) to a concrete heap using the Heap Interepretation function and the well-typed {\it Heap} ($\Gamma \vdash \mathcal{H}$).

\begin{lemma}[Preservation of Types under Substition]
  \label{lem:substitution}
  If $\Gamma$, {\sf x :} $\tau$  $\vdash \eip $ : $\forall \alpha. \sigma$, and $\Gamma \vdash \S{s} : \tau$ then 
    $\Gamma \vdash$ $[\S{s}$/$\S{x}]\eip$ : $[\S{s}$/$\S{x}]\forall \alpha. \sigma$ 
\end{lemma}
\begin{proof}
  Following is the definition of substitution for refined type
  given in Section 5.1 in the main paper.
  \[
\begin{array}{rl}
  [x_a/x]\{ \nu : \S{t} | \phi \} &= \{ \nu : \S{t} | [x_a/x]\phi \}\\
  [x_a/x](y : \tau) \rightarrow \tau' &= (y : [x_a/x]\tau) \rightarrow [x_a/x]\tau', y \neq x\\
  [x_a/x]\mathsf{PE}^{\el} \{ \phi_1 \} \{ \nu : \S{t} \} \{ \phi_2 \} &=
      \mathsf{PE}^{\el} \{ [x_a/x]\phi_1 \} \{ \nu : \S{t} \} \{ [x_a/x]\phi_2 \}
\end{array}
\]
The proof for the lemma is by induction on the Typing derivations using the above definition of substitution.
 
\end{proof}

\begin{lemma}[Cannonical Form for Values]
  \label{lem:cannonical}
  \begin{itemize}
    \item If \S{v} is a value of type bool then \S{v} is either {\sf true} or {\sf false}
    \item If \S{vs} is a value of type exc then \S{v} is {\sf Err}
    \item If  \S{v} is a value of type unit then \S{v} is ().
    \item If  \S{v} is a value of type (x : {$\tau$}) -> $\tau$ then \S{v} = $\lambda$ (x:$\tau$). \eip  
  \end{itemize}
\end{lemma}
\begin{proof}
  Proof is by case analysis of grammar rules in $\lambda_{sp}$ definition.
\end{proof}
\begin{definition}[{\sf consistent} $\Gamma$ $\Gamma$']
  $\forall$ x $\in$ dom ($\Gamma$), such that if x $\in$ dom ($\Gamma'$) then $\Gamma \vdash$ x : $\sigma$ $\implies$ $\Gamma' \vdash$ x : $\sigma$ $\wedge$ 
  $\forall \phi$. $\Gamma \vDash \phi$ $\implies$ $\Gamma' \vDash \phi$.
\end{definition}

\begin{lemma}[$\Gamma$ Weakening]
  \label{lem:weakening}
  If $\Gamma \vdash \eip : \forall \alpha. \sigma$ and $\exists \Gamma'$ such that $\Gamma \subseteq \Gamma'$ 
  and ({\sf consistent} $\Gamma$ $\Gamma'$) then $\Gamma' \vdash \eip : \forall \alpha. \sigma$
\end{lemma}
\begin{proof}
  Follows from the definition of {\sf consistent} $\Gamma$ $\Gamma'$
\end{proof}

\begin{lemma}[Inversion of the Typing Relations]
  \label{lem:inversion}
  Given a Typing judgement of the form- 
  \[\inference[]{\Gamma \vdash \eip_1 : \tau_1 & \dots \Gamma_n \vdash \eip_n : \tau_n}
    {\Gamma \vdash \eip : \tau}. 
  \] 
  Given $\Gamma \vdash \eip : \tau$, the following holds: ($\Gamma \vdash$ $\eip : \tau$ $\Leftrightarrow$ $\Gamma_i \vdash  \eip_1 : \tau_1$, \dots $\Gamma_n \vdash \eip_n : \tau_n$).
\end{lemma}
\begin{proof}
  The proof immediately follows from the definition of typing rules.
\end{proof}
\bigskip

\begin{lemma}[Uniqueness]
\label{lem:uniqueness}
  Forall all well-typed term $\eip$, $\Gamma \vdash \eip : \sigma$ and well-typed heaps $\mathcal{H}$ and $\mathcal{H'}$, $\Gamma \vdash \mathcal{H}$ and $\Gamma \vdash \mathcal{H'}$, such that $(\mathcal{H}; \mathsf{\eip}) \Downarrow (\mathcal{H'}; v)$, for all $\lp$ we have:
  \begin{itemize}
      \item $\lp$ $\mapsto$ $\eip \in \mathcal{H}$ $\implies$ $\lp$ $\mapsto$ $v \in \mathcal{H'}$
      \item $\lp \notin \S{dom} (\mathcal{H})$ $\implies$ $\exists \eip'. \lp$ $\mapsto$ $\eip' \in \mathcal{H'}$
  \end{itemize}
  where $\lp$ $\mapsto$ $\eip$ denotes that $\lp$ is a {\it unique reference} to $\eip$.
  \\
  i.e. $\forall$ $\mathcal{H}$, any well typed \name\ expression in the initial heap pointed to by a unique reference $\lp$, upon evaluation to a value $v$ and an output heap $\mathcal{H'}$ stays uniquely pointing in the final heap, and any new reference created in the final heap uniquely points to a \name expression.
\end{lemma}

To prove soundness of \name typing, we first prove a soundness lemma for pure expressions (i.e. expressions with non computation type).
\begin{lemma}[Soundness Pure-terms]
 If $\Gamma$ $\vdash$ \eip : \{ $\nu$ : t | $\phi$ \} then:
 \begin{itemize}
   \item \textnormal{Either \eip\ is a value with $\Gamma \models$ $\phi$ ( \eip )} 
   \item \textnormal{OR Given there exists a $\S{v}$ and $\mathcal{H'}$, such that ($\mathcal{H}$; \eip) $\Downarrow$ ($\mathcal{H}$; \S{v}) and $\Gamma \vdash \mathcal{H}$ then $\Gamma \vdash$ \S{v} : t and $\Gamma \models$ $\phi$ ($\S{v}$) }
\end{itemize}   
\end{lemma}
\begin{proof}
  The proof proceeds by induction on the derivation of Typing rules $\Gamma \vdash$ \eip\ : \{ $\nu$ : t | $\phi$ \}:
  \begin{itemize}
    \item The case for constants like T-True, T-false, T-zero, etc. is trivially true as these rules are axioms.
    \item Case {\sc T-capp } : Given $\Gamma \vdash D_i\ \overline{\S{t}_k} \overline{\S{v}_j} : [\overline{\S{t / \alpha }}] [\overline{\S{v_j / x_j}}] \tau$
      \begin{enumerate}
        \item \eip is a value, thus we proove (PG1).
        \item Using {\it Substition Lemma}, we get $\Gamma \models$ $\phi$ ($D_i\ \overline{\S{t}_k} \overline{\S{v}_j} : [\overline{\S{t / \alpha }}] [\overline{\S{v_j / x_j}}]$) 
      \end{enumerate}
    \item Case {\sc T-fun} is not the form (\eip : \{ $\nu$ : t | $\phi$ \}) thus the requirement for the theorem is vacuosly satisfied.
    \item Case {\sc T-typApp} : ${\Gamma \vdash {\Lambda \alpha. \eip}[\S{t}] : [\S{t}/\alpha]\{ \nu : t | \phi \}}$
          \begin{enumerate}
            \item Using Inversion Lemma ${\Gamma \vdash \Lambda \alpha. \eip : \forall \alpha. \{ \nu : t | \phi \}}$
            \item Using {\sc T-fun} we have $\Gamma, \alpha \vdash \eip : \{ \nu : t | \phi \}$.
            \item Using Using IH, we must have $(\mathcal{H};{\eip} \Downarrow (\mathcal{H}; \S{v})$ and $\Gamma, \alpha \vdash \S{v} : t$ and $\Gamma, \alpha \models \phi (\S{v})$
            \item Applying Substition Lemma thus we have $\Gamma, \S{t}, \models [\S{t}/\alpha]\phi (\S{v})$ \dots (Proof-body) 
            \item Using the {\sc P-typApp} we have $(\mathcal{H};{\Lambda \alpha. \eip}[\S{t}]) \Downarrow (\mathcal{H}; [\S{t}/ \alpha]\eip)$ and using the above (Proof-body)
            we have $\Gamma \S{t}, \models [\S{t}/\alpha]\phi (\S{v})$ giving us PG2
          \end{enumerate}
    \item Case {\sc T-let} : $\Gamma \vdash \mathbf{let}\, \xp\, =\, \eip_1\, \mathbf{in}\, \eip_2\, :\, \sigma'$
          \begin{enumerate}
            \item By IL $\Gamma \vdash \eip_1\ : \forall\alpha.\sigma$ and $\Gamma, \xp\ : \forall\alpha.\sigma \vdash \eip_2 : \sigma'$.
            \item Using IH $\exists$ a transition $(\mathcal{H}; \eip_1) \Downarrow (\mathcal{H}; {\sf v})$ and $\Gamma \vdash {\sf v}$ : $\forall\alpha.\sigma$
            \item Since the basetype for {\sf v} is same as {\sf x}, thus the substitution operation $[{\sf v }/ \xp]\eip_2$ is valid.
            \item Again using IH on the secodn judgement in IL above we have $(\mathcal{H}; [{\sf v }/ \xp]\eip_2) \Downarrow (\mathcal{H}; {\sf v'})$ and $\Gamma \vdash {\sf v'} : \sigma'$. \dots (IH2)
            \item Using above arguments, the preconditions for {\sc P-let} are valid, thus {\sc P-let} is applicable, giving us i.e. $(\mathcal{H};\mathbf{let}\, \xp\, =\, \eip_1\, \mathbf{in}\, \eip_2) \Downarrow (\mathcal{H}; {\sf v'})$
            \item Finally IH2 directly gives us PG2, i.e $\Gamma \vdash {\sf v'} : \sigma'$
          \end{enumerate}
    \item Case {\sc T-var} : $\Gamma \vdash \S{x} : \sigma$.
          \begin{enumerate}
            \item Using IL $\Gamma (\S{x}) = \sigma$
            \item PG1 and PG2 directly hold using the IH for {\sf x} already in $\Gamma$.
          \end{enumerate}
        
   \end{itemize}
  
\end{proof}

\begin{theorem}[Soundness \name]
  \label{thm:soundness}
  Given a specification $\sigma$ = $\forall \overline{\alpha}$. $\mathsf{PE^{\el}}$
  $\{\phi_1\}$ $\nu$ : {\sf t} $\{\phi_2\}$ and a \name\ expression \eip, such
  that under some $\Gamma$, $\Gamma \vdash$ \eip : $\sigma$, then if
  there  exists some well-typed heap $\mathcal{H}$ such that $\Gamma \models \phi_1 (\mathcal{H})$
  then:
  \begin{itemize}
    \item \textnormal{Either \eip \ is a value, and:} 
    \begin{enumerate}
      \item $\Gamma , \phi_1$ $\models$ $\phi_2$ ($\mathcal{H},\ \eip$,\ $\mathcal{H}$)
    \end{enumerate}  
    \item \textnormal{OR Given there exists a $\mathcal{H'}$ and \S{v} such that $\Gamma \vdash \mathcal{H}$ ($\mathcal{H}$; \eip) $\Downarrow$ ($\mathcal{H'}$; \S{v}), then\\
    $\exists$ $\Gamma'$, $\Gamma \subseteq \Gamma'$ and ({\sf consistent $\Gamma$ $\Gamma'$}), such that:} 
    \begin{enumerate}
      \item $\Gamma'$ $\vdash$ $\S{v} : \S{t}$. 
      \item $\Gamma', \phi_1$ ($\mathcal{H}$) $\models$$\phi_2$ ($\mathcal{H},\ \S{v}$,\ $\mathcal{H'}$) 
       \end{enumerate}
  \end{itemize}
  \end{theorem}

\begin{proof}
  The proof proceeds by induction on the derivation of Typing rules. Forall typing rules other than the {\sc T-p-fix}, we prove a much stronger argument, where we also show the progress, i.e. we prove that $\exists$ $\mathcal{H'}$ and \S{v} such that  ($\mathcal{H}$; \eip) $\Downarrow$ ($\mathcal{H'}$; \S{v}). For {\sc T-p-fix}, we assume such a $\mathcal{H'}$ and \S{v} to be given, obviating the need to reason about non-terminating programs. \\
  
  $\Gamma \vdash$ \eip\ : $\forall \alpha$. $\mathsf{PE^{\el}}$ $\{\phi_1\}$ $\nu$ : {\sf t} $\{\phi_2\}$:
  \begin{itemize}
    \item Case {\sc T-eps } : Given $\Gamma \vdash \S{eps} : \S{PE^{pure}}\, \{ \forall \S{h}.\, \S{true} \}\, \nu : \S{unit} 
      \ \{ \forall \S{h}, \nu, \S{h'}. \S{h' = h} \}$.
      \begin{enumerate}
        \item The evaluation rule {\sc P-eps} is applicable, thus $\exists \mathcal{H}$, such that $(\mathcal{H}; {\sf eps}) \Downarrow (\mathcal{H}; {\sf unit})$, hence 
        proving (G4).
        \item Using Cannonical lemma (~\ref{lem:cannonical}), we have {\sf ()} : {\sf unit} and Weakening Lemma 
        (~\ref{lem:weakening}), $\Gamma \vdash$ () : {\sf unit}, thus proving (G4.1).
        \item From {\sc P-eps} $\mathcal{H'}$ is same as $\mathcal{H}$, Thus using {\it Heap Interepretation} function 
        [$\mathcal{H'}$] = [$\mathcal{H}$] thus satisfying the post-condition \{ \S{h' = h} \} giving us (G4.2).
      \end{enumerate} 
    \item Case {\sc T-bot} : $\Gamma \vdash \S{\bot} : \S{PE^{exc}}\, \{ \forall \S{h}.\, \S{true} \}\, \nu  :  \S{exc}
     \ \{ \forall \S{h}, \nu, \S{h'}. \S{h' = h} \wedge \nu = \S{Err} \}$.
      \begin{enumerate}
        \item The evaluation rule {\sc P-bot} is applicable, thus $\exists \mathcal{H}$, such that $(\mathcal{H}; {\sf \bot}) \Downarrow (\mathcal{H}; {\sf Err})$, hence 
        proving (G4).
        \item Using Cannonical lemma we have {\sf Err} : {\sf exc} and using Weakening Lemma(~\ref{lem:weakening}), $\Gamma \vdash$ Err : {\sf exc}, thus proving (G4.1).
        \item From {\sc P-bot} $\mathcal{H'}$ is same as $\mathcal{H}$, Thus using {\it Heap Interepretation} function 
        [$\mathcal{H'}$] = [$\mathcal{H}$] thus satisfying the first conjunct of the post-condition (i.e. \{ \S{h' = h} \}) further
        using the Soundness of heap typing and {\sc P-BOT} we get us ({\sf G4.2}). 
      \end{enumerate}  
    \item Case {\sc T-P-char} : Given ${\begin{array}{@{}c@{}}
        \Gamma \vdash \S{char}\ \eip: \S{PE}^{\S{state}\, \sqcup\, \S{exc}} 
         \{ \forall \S{h}. \S{true} \}\, \nu\, :\, \mathsf{char\ result}\, \{ \phi_2 \} \\ 
          \end{array}}$  where $\phi_2 = \forall \S{h, \nu, h'}. \forall \S{x.} \\ 
          (\S{Inl(v) = x} \implies \S{x} = `c' \wedge \S{upd(h', h, inp, tail (inp))}) \wedge \\
          (\S{Inr(v) = x} \implies \S{x = Err}  \wedge  \S{sel(h,inp)} = \S{sel(h',inp)})$
         \begin{enumerate}
           \item By Inversion Lemma(~\ref{lem:inversion}) we have $\Gamma \vdash$ \eip\ : $\{ \nu' :\S{char} \mid \nu' = `c' \}$.
           \item Using the soundness result for Pure-terms(~\ref{lem:pure-soundness}) we have $\Gamma \models$ [$\nu'$ = `c'] and $\mathcal{H'}$ = $\mathcal{H}$.
           \item Doing a Case split on the two evaluation rules {\sc P-CHAR-True} and {\sc P-CHAR-False} :
                \begin{itemize}
                  \item Case {\sc P-char-true}:
                   $(\mathcal{H}; {\sf char} \ \eip) \Downarrow (\mathcal{H}[{\sf inp} \mapsto]; \emph{`c'})$
                   a) Using Cannonical Form Lemma, $\Gamma \vdash$ $\emph{`c'}$ : {\sf char}, \\
                   b) Using Soundness of Heap typing and definition of list constructor \\
                     $\Gamma \models$ $\S{upd([\mathcal{H'}], \mathcal{[H]}, inp, tail (inp))}$
                  \item Case {\sc P-char-false} :
                  $(\mathcal{H}; \eip) \Downarrow (\mathcal{H'}; {\sf Err}))$
                  a) Using Cannonical Form Lemma, $\Gamma \vdash$  {\sf Err} : {\sf exc}, \\
                  b) Using Soundness of Heap typing \\
                  $\Gamma \models$ $\S{sel([\mathcal{H}], inp) = sel(\mathcal{[H]}, inp)}$
               
                \end{itemize}
           \item Using Previous two cases and the definition of Sum type {\sf t result}, we get the required Goals (G4.1 and G4.2)      
         \end{enumerate} 
      \item Case {\sc T-P-choice} : $\begin{array}{@{}c@{}}  
                  \Gamma \vdash (p_1 \textnormal{<|>}  p_2) : 
                   \mathsf{PE}^{\el\, \sqcup\, \mathsf{nondet}} \ \{ (\phi1 \wedge \phi_2) \}
                                     \, \nu : \tau\,
                                     \{ (\phi_1'  \lor \phi_2') \} 
                   \end{array}$
              \begin{enumerate}
                \item By Inversion Lemma on the conclusion we have: 
                \item $\Gamma \vdash p_1\, :\,\mathsf{PE}^{\el}\, \{ \phi_1 \}\, \nu_1\, :\, \tau\, \{ \phi_1' \}$
                \item $\Gamma \vdash p_2\, :\,\mathsf{PE}^{\el}\, \{ \phi_2 \}\, \nu_2\, :\, \tau\, \{ \phi_2' \}$
                \item By Induction Hypothesis on the above two entailment rules we get the following 
                
                \item Using (G4 and G4.2) on Choice 1,  $\exists \mathcal{H'}_l$,  $\Gamma \phi_1 (\mathcal{H})$ $\models \{ \phi_1' \} (\mathcal{H}, \nu_1, \mathcal{H'}_l)$
                \item Similarly Using (G4 and G4.2) on Choice 2,  $\exists \mathcal{H'}_l$,  $\Gamma \phi_1 (\mathcal{H})$ $\models \{ \phi_1' \} (\mathcal{H}, \nu_1, \mathcal{H'}_r)$
                \item Using previous two points $\Gamma, (\phi_1 \wedge \phi_2) (\mathcal{H})$ $\models$ $\{ \phi_1' \} (\mathcal{H}, \nu_1, \mathcal{H'}_l)$ $\vee$
                $\{ \phi_2' \} (\mathcal{H}, \nu_2, \mathcal{H'}_r)$.
                \item Equivalently using distribution over disjunctions we get $\{ (\phi1 \wedge \phi_2) \}
                \, \nu : \tau\, \{ (\phi_1'  \lor \phi_2') \}$ giving us (G4.2)
                \item (G4.1) holds directly from the Induction Hypothesis. 
              \end{enumerate} 
      \item Case {\sc T-p-bind} : $\begin{array}{@{}c@{}}
        \Gamma' \vdash p\ \textnormal{>>=}\ \eip : 
        \mathsf{PE}^{\el} \ \{ \forall \mathsf{h}.\, \phi_1 \, \mathsf{h}\,  \wedge 
                                                 \phi_{1'} (\mathsf{h}, x, \mathsf{h_i})  => \phi_2 \ \mathsf{h_i} \} 
                    \ \nu'\,:\, \tau' \  \\
            \{ \forall \mathsf{h}, \nu', \mathsf{h'}. \phi_{1'} (\mathsf{h}, x, \mathsf{h_i}) \wedge 
                                                 \phi_{2'} ({\mathsf{h_i}, \nu', \mathsf{h'}}) \} 
         \end{array}$
         \begin{enumerate}
           \item By Inversion Lemma we have:
           \item $\Gamma\ \vdash\ p\, :\, \mathsf{PE}^{\el} \ \{ \phi_1 \}\, \nu\, :\, \S{t}  \{ \phi_{1'} \}$
           \item $\Gamma\ \vdash\ \eip\, :\, (\xp : \tau) \rightarrow \mathsf{PE}^{\el} \
           \{ \phi_2 \} \ \nu' : \S{t'} \ \{ \phi_{2'} \}$
           \item By IH (G4 and G4.1, G4.2) hold for the first judgement, thus $\exists \mathcal{H}_i$, 
                    $\Gamma, \phi_1 (\mathcal{H}) \models$ $\phi_{1'} (\mathcal{H}, \nu, \mathcal{H}_i)$
                  and $\Gamma \vdash$ $\nu$ : t \dots (IH1)
           \item Using (T-fun) and IH on the second judgement, $\Gamma, \S{x} : \tau$, if there exists some heap $\mathcal{H}_j$ such that $\phi_2$ ($\mathcal{H}_j$) then
           $\exists \mathcal{H'}$ such that $\Gamma, \S{x} : \tau, \phi_2 (\mathcal{H}_j) \models$ $\phi_{2'} (\mathcal{H}_j, \nu', \mathcal{H'})$ \dots (IH2)
           \item Using two IH above, we have the sufficient conditions to apply Bind evaluation rules {\sc P-bind-success} and {\sc P-bind-err}, we prove goals G4, G4.1 and G4.2 for each 
           of these cases:
              \begin{itemize}
                \item Case {\sc P-bind-err}, $(\mathcal{H}; p \textnormal{>>=} \eip) \Downarrow (\mathcal{H'}; {\sf Err})$, giving us G4
                  \begin{enumerate}
                    \item Using the definition of sum type {\sf t result}, the post condition for this case is handled in the second conjunct in the post-condition for {\sc T-p-bind}.
                    \item Using IH again for the first judgement in the antecedent of {\sc T-p-bind} $\Gamma, \phi_1 (\mathcal{H}) \models$ $\phi_{1'} (\mathcal{H}, \nu, \mathcal{H}_i)$ thus 
                    we have {\sf x = Err} => $\phi_{1'} (\mathcal{H}, \nu, \mathcal{H}_i)$ \dots {Proof-Err} 

                  \end{enumerate} 
                \item Case {\sc P-bind-succes}, $(\mathcal{H}; p \textnormal{>>=} \eip) \Downarrow (\mathcal{H''}; v_2)$, giving us a post heap $\mathcal{H''}$
                  \begin{enumerate}
                    \item Using the definition of sum type {\sf t result}, the post condition for this case is handled in the first conjunct in the post-condition for {\sc T-p-bind}.
                    \item Using the (IH2) argument, we need a an intermediate heap $\mathcal{H}_j$ such that $\phi_2$ ($\mathcal{H}_j$).
                    \item Given the pre-condition for {\sc T-p-bind} we have $\phi_{1'} (\mathsf{h}, x, \mathsf{h_i})  => \phi_2 \ \mathsf{h_i}$, thus we can use 
                    $\mathcal{H}_i$ as required $\mathcal{H}_j$, consequently, (IH2) implies $\Gamma, \S{x} : \tau, \phi_2 (\mathcal{H}_i) \models$ $\phi_{2'} (\mathcal{H}_j, \nu', \mathcal{H''})$ \dots (Proof-Succ)
                    \item Using IH2 we also get the G4.1 for the success branch evaluation.
                  \end{enumerate}
                      
              \end{itemize}
           \item Using (Proof-Err) and (Proof-Succ) above and the definition of sum type {\sf t result} (G4.2) for the {\sc T-p-bind} is implied by the two cases in the post-condition of {\sc T-p-bind}.
         \end{enumerate}    
      \item Case {\sc T-fix} : $\Gamma \vdash \mu\, \mathsf{x}\, :\, (\mathsf{PE}^{\el}\, \{ \phi \}\, \nu\, :\, \S{t}\, \{ \phi' \}).\, p\, :\, \mathsf{PE}^{\el}\, \{ \phi \}\, \nu\, :\, \S{t}\, \{ \phi' \} $
      \begin{enumerate}
            \item By Inversion Lemma we have $\Gamma, \mathsf{x}\, :\, (\mathsf{PE}^{\el}\, \{ \phi \}\, \nu\, :\, \S{t}\, \{ \phi' \})\, \vdash\ p\, : \mathsf{PE}^{\el}\, \{ \phi \}\, \nu\, :\, \S{t}\, \{ \phi' \}$ \dots (IL1)
            \item Using types for {\sf x} in $\Gamma$ and $\mu\, \mathsf{x : \dots}.\, p$, and the Substitution Lemma~\ref{lem:substitution}, the substitution $\mathsf{x:\sigma}. p/\S{x}]p$ is well-formed. 
            \item Using IH, we are Given $\exists$ a heap $\mathcal{H'}$ and a value $v$ such that 
            $(\mathcal{H}; [\mu \mathsf{x:\sigma}.p/\S{x}]p) \Downarrow (\mathcal{H'}; v)$. lets call this argument \dots (E1)
            \item Thus, the preconditions for rule {\sc p-fix} hold and it can be applied, giving us $(\mathcal{H}; \mu \mathsf{x:\sigma}. p) \Downarrow (\mathcal{H'}; v)$ (giving us G4)
        
            \item Using IH on the judgement from the (IL1) we get $\Gamma, {\sf x} : (\mathsf{PE}^{\el}\, \{ \phi \}\, \nu\, :\, \S{t}\, \{ \phi' \})\, \phi (\mathcal{H}) \models$ 
            $\phi'$ ($\mathcal{H}$, $\S{v}$, $\mathcal{H'}$) and $\Gamma \vdash \nu$ : {\sf t} 
           
            \item Using Subtitution Lemma $\Gamma \vdash [\mathsf{x:\sigma}.p/\S{x}]p : [\mathsf{x:\sigma}.p/\S{x}] \mathsf{PE}^{\el}\, \{ \phi \}\, \nu\, :\, \S{t}\, \{ \phi' \} $ \dots (J1)
            \item Using Inversion Lemma on the original judgement for {\sf T-p-fix}, we ${\sf x} \notin FV(\phi, \phi')$.
            \item Thus, $[\mathsf{x:\sigma}.p/\S{x}] \mathsf{PE}^{\el}\, \{ \phi \}\, \nu\, :\, \S{t}\, \{ \phi' \} $ reduces to $\mathsf{PE}^{\el}\, \{ \phi \}\, \nu\, :\, \S{t}\, \{ \phi' \}$ using definition of substitution in Types.
            \item Thus from this and (J1) we have  $\Gamma \vdash [\mathsf{x:\sigma}.p/\S{x}]p : \mathsf{PE}^{\el}\, \{ \phi \}\, \nu\, :\, \S{t}\, \{ \phi' \} $ \dots (J2)
            \item Using (E1) and (J2) and the IH, $\Gamma, \phi (\mathcal{H}) \models \phi' (\mathcal{H}, v, \mathcal{H'})$ and $\Gamma \vdash v : \S{t}$
            \item This prooves the Goals G4.1 and G4.2
            
          \end{enumerate}
      \item Case {\sc T-app} : $\Gamma \vdash \eip_f\ \xp_a\, : [\xp_a/x]\S{PE}^{\el} \{\phi\}\ \nu\ : \S{t}\ \{\phi'\}$
          \begin{enumerate}
            \item Using Inversion Lemma, we have $\Gamma \vdash \eip_f : (x\, :\, \{ \nu : \S{t} \mid \phi_x\}) \rightarrow \S{PE}^{\el} \{\phi\}\ \nu\ : \S{t}\ \{\phi'\}\ $
            and  $\Gamma \vdash \xp_a : \{ \nu : \S{t} \mid \phi_x\} $
            \item Using Cannonical Form Lemma for arrow type, we must have $\eip_f$ = $\lambda (\S{x} : \{ \nu : \S{t} \mid \phi_x\}).\eip $
            \item Using Soundness Lemma for pure term typing $\Gamma \vdash \xp_a : \{ \nu : \S{t} \mid \phi_x\} $ we have $(\mathcal{H}; \xp_a) \Downarrow (\mathcal{H}; {\sf v})$
            and $\Gamma \vdash {\sf v} : \{ \nu : \S{t} \mid \phi_x\}$
            \item Using the Typing rule {\sc T-fun}, we get $\Gamma, (x\, :\, \{ \nu : \S{t} \mid \phi_x\}) \vdash $  $\eip$ : $\S{PE}^{\el} \{\phi\}\ \nu\ : \S{t}\ \{\phi'\} $  
            \item Using Induction Hypothesis on the above Type for $\eip$ we must have if $\Gamma,  (x\, :\, \{ \nu : \S{t} \mid \phi_x\}) \models \{\phi\} (\mathcal{H})$ then 
            $\exists \mathcal{H'}$, such that  $(\mathcal{H}; \eip \Downarrow (\mathcal{H'}; {\sf v'})$ 
            \item and from G4.1 and G4.2 $\Gamma,  (x\, :\, \{ \nu : \S{t} \mid \phi_x\}) , \{\phi\} (\mathcal{H}) \models \{\phi'\} (\mathcal{H}, \nu', \mathcal{H'})$ and 
            $\Gamma,  (x\, :\, \{ \nu : \S{t} \mid \phi_x\}) \vdash \nu' : \S{t}$.
            \item Applying Substition Lemma on the above two points, $\exists$ $\mathcal{H'}$, such that  $(\mathcal{H}; [\xp_a/x]\eip \Downarrow (\mathcal{H'}; {\sf v'})$
            \item and $\Gamma,  (\xp_a\, :\, \{ \nu : \S{t} \mid \phi_x\}) , [\xp_a/x]\{\phi\} (\mathcal{H}) \models [\xp_a/x]\{\phi'\} (\mathcal{H}, \nu', \mathcal{H'})$ and 
            $\Gamma,  (\xp_a\, :\, \{ \nu : \S{t} \mid \phi_x\}) \vdash \nu' : \S{t}$.
            \item The above gives G4.1 and G4.2  
          \end{enumerate}
      \item Case {\sc T-retrun} : $\Gamma \vdash \mathbf{return}\, \eip :
                  \mathsf{PE}^{\S{pure}} \{ \forall \S{h}. \S{true} \}\, \nu :\, \S{t} \
                  \{ \forall \S{h}, \nu, \S{h'}. \S{h' = h} \wedge \phi \}$
            \begin{enumerate}
              \item Using Inversion Lemma, we have $\Gamma \vdash \eip\ : \{ \nu : \S{t} \mid \phi \}$
              \item Using Soundness lemma for pure terms on the above judgement, we have $(\mathcal{H}; \eip) \Downarrow (\mathcal{H}; {\sf v})$
              and $\Gamma \vdash {\sf v} : \{ \nu : \S{t} \mid \phi\}$ \dots (Proof-pure)
              \item Thus we can apply {\sc P-return} giving $(\mathcal{H}; \mathbf{return}\, \eip) \Downarrow (\mathcal{H}; {\sf v})$ giving us (G4)
              \item Using above evaluation we have $\mathcal{H'} = \mathcal{H}$, this and using (Proof-pure) we get $\Gamma, \S{true} \models$ $\mathcal{H'} = \mathcal{H} \wedge \phi (\nu)$, giving
              us G4.2.
              \item G4.1 follows directly from (Proof-pure)

            \end{enumerate}
      \item Case {\sc T-match} : $\Gamma \vdash {\bf match} \ \vp \ {\bf with} \ \mathsf{D_i}\ 
                                \overline{\mathsf{\alpha}_{k}} \overline{\mathsf{x}_{j}} \rightarrow e_i : \S{PE}^{\S{\el}} \{ \forall\ \S{h}. 
                                \bigwedge_{i}^{} (\vp\ = \mathsf{D_i} \ \overline{\mathsf{\alpha}_{k}} \overline{\mathsf{x}_{j}}) =>  
                                \phi_i\}\ \nu\ : \S{t}\ \{ 
                                \forall\ \S{h}, \nu', \S{h'}.                           
                                \bigvee_{i}^{} \phi_{i'}\}\ $ 
             \begin{enumerate}
              \item The soundness argument is presented for soem {\sf i} and then generalized for each $\mathsf{D_i}$.    
              \item Using Inversion Lemma $\Gamma \vdash \vp : \tau_0$
               \item By Soundness Pure Lemma ($\mathcal{H}$; \eip) $\Downarrow$ ($\mathcal{H}$; \S{v}) and $\Gamma \vdash \S{v} : \tau_0$
               \item By Inversion Lemma again $\Gamma_{i} = \Gamma, \overline{\alpha_k}, \overline{\mathsf{x}_{j} :\tau_j}$ and 
               $\Gamma_i \vdash \mathsf{D_i} \ \overline{\mathsf{\alpha}_{k}} \overline{\mathsf{x}_{j}} : \tau_0$
                \item Using the pre-condition of the {\sc T-match} and the given conditions for the theorem, we extract the pre-consition component for the 
                $\mathsf{D_i}$, thus we have $(\vp\ = \mathsf{D_i} \ \overline{\mathsf{\alpha}_{k}} \overline{\mathsf{x}_{j}})$.
                \item $\Gamma_{iext}$ = $\Gamma_i$, $(\vp\ = \mathsf{D_i} \ \overline{\mathsf{\alpha}_{k}} \overline{\mathsf{x}_{j}})$
                \item Using Inversion Lemma once again we have $\Gamma_{i} \vdash e_i : \S{PE}^{\S{\el}} \{\phi_i\}\ \nu\ : \S{t}\ \{\phi_{i'}\}\ $ \dots {\sf (J1)}
                \item Using IH for the above judgement we get $\exists \mathcal{H'}$ such that $(\mathcal{H}; \S{e_i}) \Downarrow (\mathcal{H'}; \S{v_i})$ 
                \item Using above conclusions, we have the required pre-conditions for the application of the evaluation rule {\sf P-match}, thus we get $\exists \mathcal{H'}$ such that 
                $(\mathcal{H};{\bf match} \ \vp \ {\bf with} \ \mathsf{D_i}\ \overline{\mathsf{\alpha}_{k}} \overline{\mathsf{x}_{j}} \rightarrow e_i) \Downarrow (\mathcal{H'}; \S{v_i})$ 
                \item Now using IH against {\sf (J1)} again, we get $\Gamma_{iext}, \phi_i (\mathcal{H}) \models \phi_i' (\mathcal{H}, \S{v_i} \mathcal{H'})$ and $\Gamma_i \vdash \S{v_i} : \S{t}$
                \item Finally the above only holds for the given assumption $(\vp\ = \mathsf{D_i} \ \overline{\mathsf{\alpha}_{k}} \overline{\mathsf{x}_{j}})$ in $\Gamma_{iext}$, thus 
                we can move the assumption to the pre-condition we get, $\Gamma_{i}, ((\vp\ = \mathsf{D_i} \ \overline{\mathsf{\alpha}_{k}} \overline{\mathsf{x}_{j}}) => \phi_i (\mathcal{H})) \models \phi_i' (\mathcal{H}, \S{v_i} \mathcal{H'})$
                \item The above give (G4.2) for some $\mathsf{D_i}$.
                \item Now generalizing this for each {\sf i} we get : $\bigcup_{i}^{} \Gamma_i, \bigwedge_{i}^{} 
                                                                  ((\vp\ = \mathsf{D_i} \ \overline{\mathsf{\alpha}_{k}} \overline{\mathsf{x}_{j}}) => \phi_i(\mathcal{H})) 
                                                                  \models  \bigvee_{i}^{} \phi_{i}' (\mathcal{H}, \S{v_i} \mathcal{H'})$ 

                \item The above gives us (G4.2) and (G4.1) holds directly as $\Gamma_i \vdash \S{v_i} : \S{t}$ for each {\sf i}
             \end{enumerate}                   
      \item Case {\sc T-deref} : $\Gamma \vdash \S{deref} \ \rp : \{ \forall\ \S{h}. \S{dom (h,\rp)} \}\ \nu'\ : \S{t}\ \{ \forall\ \S{h}, \nu', \S{h'}. \S{sel(h,\rp)}\ = \nu' \wedge \S{h = h'} \} $
             \begin{enumerate}
               \item Using IL $\Gamma \vdash \rp\ : \S{PE}^{\S{state}} \{\phi_1\}\ \nu\ : \S{t\ ref}\ \{\phi_2\}\ $
               \item By IH on the above judgement we have $\Gamma \vdash \rp\ : \S{t\ ref}\  $ \dots (IH1)
                \item Using the pre-condition for {\sc T-deref} $\Gamma \models \S{dom (h,\rp)}$.
                \item Using the Given Heap Soundness over the logical entailments $\Gamma \models \S{dom (h,\rp)}$ implies $\exists \S{v}$. $\mathcal{H} (\rp) = \S{vp}$. \dots (Heap-map)
                \item Using the above argument, {\sc T-deref} is applicable $(\mathcal{H}; \S{deref} \ \rp) \Downarrow (\mathcal{H}; \S{vp})$.
                \item Using the given Heap Soundness and (Heap-map) we get $\S{sel(\mathcal{H},\rp)} = \S{\vp}$ and using the above evaluation step $\mathcal{H'}$ = $\mathcal{H}$.
                \item Thus $\Gamma, \S{dom (\mathcal{H},\rp)} \models \S{sel(\mathcal{H},\rp)} = \S{\vp} \wedge \mathcal{H'}$ = $\mathcal{H}$ giving us (G4.2)
                \item Again using the given Heap Soundness and (IH1), we get $\Gamma \vdash S{\vp} : \S{t}$ giving us (G4.1) 
             \end{enumerate} 
      \item Case {\sc T-assign} : $ \Gamma \vdash \rp\ \mathsf{:=}\ \eip\  : \{ \forall \S{h}.\S{dom(h,\rp)} \}\ \nu'\ :\ \S{t}\ \{ \forall\ \S{h},\nu',\S{h'}.\S{sel(h',\rp)}\ =\ \nu' \wedge\ \phi(\nu') \} $ 
              \begin{enumerate}
                 \item Using IL $\Gamma \vdash\ e\ :\ \{ \nu\ :\  \S{t} \mid\ \phi \}$. \dots (IH1)
                 
                 \item Applyin Soundness Lemma for pure terms on the above judgement 
                 we have $(\mathcal{H}; \mathsf{\eip}) \Downarrow (\mathcal{H}; \S{v})$ and $\Gamma \models \phi$
                 
                 (\S{v}) (Using the definition of $\Gamma \models \phi$) 
                 \item Using the above argument, the preconditions for {\sc P-assign} are valid thus we get $(\mathcal{H}; \rp\ := \S{v}) \Downarrow (\mathcal{H}[\rp \mapsto \S{v}]; \S{v})$.
                 \item The above reduction has to component, a) The location $\rp$ is updated and the remaining heap $\mathcal{H}$ remains the same. 
                 \item Using the given Heap Soundness we get, let $\mathcal{H'}$ = $(\mathcal{H}[\rp\ \mapsto \S{v})$ then $\S{sel(\mathcal{H'},\rp)}\ =\ \S{v}$ holds. \dots (HS1)
                 
                 \item Using the Definition of {\sc T-frame}, we also get the for all $\phi_r$, such that 
                 ${\bf Locs (\phi_r)} \cap ({\bf Locs (\phi)} \cup {\bf Locs (\phi')}) = \emptyset$ we have ($\phi_r \wedge \phi'$) $\mathcal{H'}$ \S{v} $\mathcal{H'}$. \dots (HS2)  
                 \item Using HS1 ad HS2 we get that the post-condition $\Gamma, \phi_r \wedge \phi' \models$ $\phi_r \wedge \phi'$).This along with IH1 above gives us G4.2
                 \item The above gives G4.2, while G4.1 holds directly from IH1.
              \end{enumerate}
      \item Case {\sc T-ref} : $\Gamma\ \vdash\ \B{let}\ \rp\ = \S{ref}\ \eip\ \B{in}\ \S{e}_b\ :\ \sigma$
             \begin{enumerate}
               \item Using Inversion Lemma we have $\Gamma\ \vdash\ \vp\ :\ \{\ \nu\ :\ \S{t} \mid\ \phi\ \}$ \dots ({\sf IL1})
               \item and $\Gamma \vdash \rp\ : \S{PE}^{\S{state}} \{ \forall\ \S{h}. \neg\ \S{dom(h,\rp)} \}\ 
               \nu'\ : \S{t\ ref}   \{ \forall\ \S{h,\nu',h'}.  \S{sel(h',\rp)}\ =\ \vp\ \wedge\\
                \phi(\vp) \wedge\ \S{dom(h',\rp)} \} \dots ({\sf IL2})$
               \item and $\Gamma, \rp\ : \S{PE}^{\S{state}} \{ \forall\ \S{h}. \neg\ \S{dom(h,\rp)} \}\ 
               \nu'\ : \S{t\ ref}   \{ \forall\ \S{h,\nu',h'}.  \S{sel(h',\rp)}\ =\ \vp\ \wedge\\
                \phi(\vp) \wedge\ \S{dom(h',\rp)} \} \vdash\ \S{e}_b\ : \S{PE}^{\S{\el}} \{\S{dom(h,\rp)}\}\ \nu\ : \S{t}\ \{\phi_b'\}\  $ (\dots {\sf IL3})
                \item By IH on (IL1) we have $\Gamma \models \phi (\vp)$
                \item By IH on (IL2) , if $\exists \mathcal{H}$ such that $\neg\ \S{dom(\mathcal{H},\rp)}$  then $\exists \mathcal{H}_i$, such that
                $(\mathcal{H}; \B{let}\ \rp\ = \S{ref}\ \eip\ ) \Downarrow (\mathcal{H}_i[\rp \mapsto v]; \rp)$.
                \item Given the pre-condition and the statement of the theorem $\Gamma \models \neg\ \S{dom(h,\rp)} \mathcal{H}$ thus the right hand side of the above implication 
                holds.
                \item Completing the IH argument on (IL2), $\Gamma, \neg\ \S{dom(\mathcal{H},\rp)} \mathcal{H} \models \S{dom(\mathcal{H}_i,\rp)} \wedge \phi (\vp)$. \dots ({\sf IH2})
                \item Using  IH on (IL3), if $\exists \mathcal{H}_j$, such that  $Gamma \models \S{dom(\mathcal{H}_j,\rp)}$ then $(\mathcal{H}[\rp \mapsto v]; \eip_b) \Downarrow (\mathcal{H'}; {\sf v'})$ 
                \item Using (IH2), we can substitute $\mathcal{H}_i$ for $\mathcal{H}_j$ in the above statement, thus $(\mathcal{H}[\rp \mapsto v]; \eip_b) \Downarrow (\mathcal{H'}; {\sf v'})$ 
                \item Completing the IH argument on (IL3), $\exists \mathcal{H_i}$ such that $Gamma,\S{dom(\mathcal{H}_i,\rp)} \models \phi_b' (\mathcal{H}_i, \S{v'}, \mathcal{H'}) \wedge \S{dom(\mathcal{H}_i,\rp)}$ 
                \item Using the given Heap Soundness 
                $\Gamma, \S{h_i : heap}, \S{dom(\mathcal{H}_i,\rp)} \models \phi_b' (\mathcal{H}_i, \S{v'}, \mathcal{H'}) \wedge \S{dom(\mathcal{H}_i,\rp)}$  and 
                $\Gamma, \S{h_i : heap} \vdash \S{v'} : S{t}$  \dots ({\sf IH3})
                \item Thus, the preconditions for {\sf P-ref} holds, applying it we have $(\mathcal{H}; \B{let}\ \rp\ = \S{ref}\ \eip\ \B{in}\ \S{e}_b \ ) \Downarrow (\mathcal{H'}; {\sf v'})$
                \item Using transitive reasoning over (IH2) and (IH3) we get  $\Gamma, \S{h_i : heap}, \neg\ \S{dom(\mathcal{H},\rp)} \models \phi_b' (\mathcal{H}_i, \S{v'}, \mathcal{H'}) \wedge \S{dom(\mathcal{H}_i,\rp)} $ 
                giving us (G4.2) for the {\sc T-ref} typing rule.
                \item (G4.1) follows directly from IH3.
			
             \end{enumerate}
      
  \end{itemize}
  
\end{proof}

\paragraph{Decidability}
Propositions in our specification language are first-order formulas in
the theory of Equality {\sf +} Uninterpreted Functions {\sf +} Linear Integer Arithmetic
(EUFLIA)~\cite{eufa}. 

The subtyping judgment in $\lambda_{sp}$ relies on the semantic
entailment judgment in this theory. Thus, decidability of type
checking in $\lambda_{sp}$ reduces to decidability of semantic
entailment in EUFLIA.  Although semantic entailment is undecidable for
full first-order logic, the following lemma argues that the
verification conditions generated by \name typing rules always
produces a logical formula in the Effectively Propositional
(EPR)~\cite{epr,smt-eps} fragment of this theory consisting of formulae with
prenex quantified propositions of the forms $\exists^{*}$
$\forall^{*}$ $\phi$.  Off-the-shelf SMT solvers (e.g., Z3) are
equipped with efficient decision procedures for EPR
logic~\cite{smt-eps}, thus making typechecking decidable in \name.

\begin{definition}
  We define two judgments:
  \begin{itemize}
    \item $\vdash$ $\Gamma$ {\sf EPR} asserting that all
    propositions in $\Gamma$ are of the form $\exists^{*}$ $\forall^{*}$
    $\phi$ where $\phi$ is a quantifier free formula in {\sf EUFLIA}.
    \item $\Gamma$ $\vdash$ $\phi$ {\sf EPR}, asserting that under a given $\Gamma$,
    semantic entailment of $\phi$ is always of the form $\exists^{*}$
    $\forall^{*}$ $\phi'$.  
  \end{itemize}
\end{definition}

\begin{lemma}[Grounding]{\label{lem:grounding}}
If $\Gamma$ $\vdash$ {\sf e} : $\tau$,
then $\vdash$ $\Gamma$ {\sf EPR} and if\ $\Gamma$ $\vDash$ $\phi$ then
$\Gamma$ $\vdash$ $\phi$ {\sf EPR}
\end{lemma}

\begin{proof}
  
  \begin{itemize}
    \item Proof for: \\
    If $\Gamma$ $\vdash$ {\sf e} : $\tau$, then $\vdash$ $\Gamma$ {\sf EPR} \\
    If $\Gamma$ $\vdash$ {\sf e} : $\tau$,
        then $\vdash$ $\Gamma$ EPR uses finite induction on typing rules which add a formula $\phi$ in $\Gamma$. intuitively we prove that For all rules 
        iff $\vdash \Gamma'$ {\sf EPR} for union of enviornments $\Gamma'$ in rule's antecedents then $\vdash \Gamma$ {\sf EPR} in the consequence. 
        
        \begin{itemize}
          
          \item Case {\sc T-p-bind} : This rule extends the environment with two variables {\sf x} and intermediate heap $\mathsf{h_i}$, since the language of specification has no 
          existential quantifier $\exists$, all formulas in this extended $\Gamma$ are of the form $\exists$ {\sf x}, $\mathsf{h_i}$. $\forall$ , hence $\vdash$ $\Gamma$
          
          \item Case {\sc T-match} : This rule also extends the environment with variables for constructor arguments, here again the argument for {\sf T-p-bind} holds.
          
          \item For all other rules $\Gamma$ $\subseteq$ $\Gamma'$ thus the argument trivially holds.
        
        \end{itemize}
    \item Proof for: \\
     If\ $\Gamma$ $\vDash$ $\phi$ then $\Gamma$ $\vDash$ $\phi$ {\sf EPR} \\
    if $\Gamma$ $\vDash$ $\phi$ then this must be created using subtyping rules 
    as these are the only rules which translate syntactic typing to semantic entailment in logic. 
    Now using the fact that: 
    \begin{enumerate}
      \item all specifications in \name are either quantifier free or can use universal quantifiers,  
      \item from first part of the proof we know that existentials reside in $\Gamma$, 
      we get that for any subtyping entailment of the form $\Gamma, \exists^{*} \forall^{*} \dots $ $\vDash$ $\phi$ $\implies$ $\phi'$, 
      ($\phi$ $\implies$ $\phi'$), is free from existentials.
      \item Using the definition of $\Gamma \vDash \phi$ the above translated to the formula of the form $\exists^{*}$ $\forall^{*}$ $\wedge$ $\phi$ $\implies$ $\phi'$ . 
      \item The above implication is in EPR.
    \end{enumerate}
     
\end{itemize}
  
\end{proof}

\begin{theorem}[Decidability \name]
  \label{thm:decidability}
  Typechecking in \name\ is decidable.
  \end{theorem}

\begin{proof}
Follows from Grounding lemma and decidability of {\sf EPR} fragment in EUFLIA.
\end{proof}

\section{Other Supplemental Items}
\subsection{Implementation}
An anonymized implementation repository and benchmarks  are available at: 
https://anonymous.4open.science/r/morpheus-DEF4/README.md

\subsection{Derived Combinators}
Following a non-exhaustive list of commonly used Derived combinators available in \name
\begin{lstlisting}[escapechar=\@,basicstyle=\small\sf,breaklines=true,language=ML]
 let any l = List.fold_left (fun acc pi -> acc <|> pi) l

 let map f p = (p >>= \x f x)
 
 let (>>) e1 e2 = e1 >>= \_ e2 
 
 let (<<) e1 e2 = e1 >>= \x. e2 >>= return x 
 
 let option e = (e >>= \r. return Some r) <|> (eps >>= \_ retrun None)  
 
 let star e = fix (\e_star : @$\tau$@. 
		map (\_ -> []) eps 
		<|>
                (e >>= \x. 
                e_star >>= \xs. return (x :: xs)) 
 
 let plus e = e >> (star e)
 
 let count n p = fix (\countnp : @$\tau$@. 
		if (n <= 0) then 
		 map (\_ -> []) eps 
		else
		  (p >>= \x. 
		  countnp (n-1) >>= \xs. return (x :: xs)) 
\end{lstlisting}

\subsection{Benchmark Grammars}
Following are the grammars for the Benchmark applications:
\begin{enumerate}
 \item {\bf PNG-chunk}
\begin{lstlisting}[escapechar=\@,basicstyle=\small\sf,breaklines=true]
 png : header . many chunk 
 chunk : length . typespec . content . Pair (length,content)
 length : number 
 typespec : char 
 content : char*
\end{lstlisting}
 \item {\bf PPM}
  
\begin{tabbing}
\small
 {\sf ppm} : {\sf ``P''} . {\sf versionnumber} . {\sf header} . {\sf data} \\
 {\sf versionnumber} :  digit \\
 {\sf header } : width = number . height = number . max = number \\
 {\sf data } : rows* [length (rows) = height]\\
 {\sf row } : rgb* [length (rgb) = width]  \\
 {\sf rgb } : r = number . g = number  . b= number [r < max, g < max, b < max] \\ 
 {\sf number} :  digit* \\
 
 \end{tabbing}
\item {\bf Haskell-case exp}

\begin{lstlisting}[escapechar=\@,basicstyle=\small\sf,breaklines=true]
caseexp : offside ('case' . exp . 'of') . offside (align alts) 
alts :  (alt) . (alt)* 
alt : pat ralt;
ralt : ('->' exp)
pat : exp
exp : varid
varid : [a-z, 0-9]*
\end{lstlisting}

\item {\bf Python-while-block}
\begin{lstlisting}[escapechar=\@,basicstyle=\small\sf,breaklines=true]
while_stmt: offsie 'while' . offside test . offside ':' . offside suite
suite: offside NEWLINE . offside stmt+
test: expr op expr
expr : identifier
stmt: small_stmt NEWLINE
small_stmt: expr op expr
op : > | < | = 
\end{lstlisting}

\item {\bf xauction}
\begin{lstlisting}[escapechar=\@,basicstyle=\small\sf,breaklines=true]
  listing :  sellerinfo .  auctioninfo
  sellerinfo :  sname . srating 
  auctioninfo  : bidderinfo+  
  bidderinfo  : bname . brating
  sname :  "<name>" name "</name>"
  srating :  "<rating>" number "</rating>"
  bname :  "<name>" name "</name>"
  brating : "<rating>" number "</rating>"

  name : [a-z].[a-z,0-9]*
  number : [0-9]+
\end{lstlisting}

\item {\bf xprotein} where {\sf proteins} is a global list of parsed proteins. 
\begin{lstlisting}[escapechar=\@,basicstyle=\small\sf,breaklines=true]

proteindatabase  = database  proteinentry+
database =  <database> uid </database>
proteinentry =  <ProteinEntry> header  protein 	skip* </ProteinEntry> 
header = <header> . uid .</header>
uid = number
protein = <protein> name . id . </protein> [@$\neg$@ (name @$\in$@ proteins)] {proteins.add name}
\end{lstlisting}

\item {\sf health}
\begin{itemize}
 \item Following is the custom-stateful regex patter-matcher
\begin{lstlisting}[escapechar=\@,basicstyle=\small\sf,breaklines=true]
Custom Regex pattern = 
	<skip> ([^,]*, {4}) 
		(<?round-off> cancer-deaths) 
			@$\lambda$@ x. (<skip>[,*,] {2}) 
				(<?check-less-than x> cancer-deaths-min) 
					@$\lambda$@ y. 
					(<?check-greater-than x> cancer-deaths-max) @$\lambda $@ z.(<Triple {x;y;z}> [*\n])
\end{lstlisting}
\item 
Following is the grammar capturing the above pattern:
\begin{lstlisting}[escapechar=\@,basicstyle=\small\sf,breaklines=true]
csvhealth : count 4 (skip) . x = cancer-deaths . (count 2 skip) . y = cancer-deaths-min [y < x] . z = cancer-deaths-max [z > x]   
skip : [a-z]*.','
cancer-deaths : number 
cancer-deaths-min : number 
cancer-deaths-max : number
\end{lstlisting}

\end{itemize}

\item {\bf streams}

\begin{lstlisting}[escapechar=\@,basicstyle=\small\sf,breaklines=true]
streamicc :  t = tagentry . chunk (t)
tagentry : signature .  offset  .  size
signature : number
offset : number 
size : number
chunk (t) : s = GetStream . s1 = Take (t.sz) s . SetStream s1 . Tag (t.signature) . s2 = Drop sz s . Setstream s2
  
GetStreamm : !inp
SetStream (s1) : inp := s1
Tag (choice) : tag-left [choice=0]| tag-right[choice=1]
tag-left :x = number [x = 0] 
tag-right : x = number [x =1]
\end{lstlisting}

\item {\bf c typedef}

\begin{lstlisting}[escapechar=\@,basicstyle=\small\sf,breaklines=true]
decl := "typedef" . 
		typeexpr . 
		id=rawident [@$\neg$@ id @$\in$@ (!identifiers)]
		{types.add id} 

typename := x = rawident [x @$\in$@ (!types)]{return x}
typeexp := "int" | "bool"
expr :=  id=rawident {identifiers.add id ; return id}
program := many decl . many expr
\end{lstlisting}

\end{enumerate}

\subsection{Specification Monad Morphism}

\begin{lstlisting}[escapechar=\@,basicstyle=\small\sf,language=ML]
@${\mathcal{T}_{\mathsf{pure}}}^{\mathsf{state}}$@ ({ @$\nu$@ : t | @$\phi$@}) = @$\mathsf{PE}^{\mathsf{state}}$@ {true}  @$\nu$@ : t { @$\phi$@ @$\wedge$@ h' = h}

@${\mathcal{T}_{\mathsf{pure}}}^{\mathsf{exc}}$@ ({ @$\nu$@ : t | @$\phi$@}) = @$\mathsf{PE}^{\mathsf{exc}}$@ {true}  @$\nu$@ : t result { x = Inl (@$\nu$@) @$\wedge$@ @$\phi$@[x/@$\nu$@] @$\wedge$@ h' = h}

@${\mathcal{T}_{\mathsf{state}}}^{\mathsf{stexc}}$@ (@$\mathsf{PE}^{\mathsf{state}}$@ {@$\phi$@}  @$\nu$@ : t { @$\phi'$@}) = @$\mathsf{PE}^{\mathsf{stexc}}$@ {@$\phi$@}  @$\nu$@ : t result { x = Inl (@$\nu$@) @$\wedge$@ @$\phi'$@[x/@$\nu$@]}


@${\mathcal{T}_{\mathsf{pure}}}^{\mathsf{nondet}}$@ (pure { v : t | @$\phi$@}) = 
	@$\mathsf{PE}^{\mathsf{nondet}}$@ {true}  v : t { @$\phi$@ @$\wedge$@ h' = h}


@${\mathcal{T}_{\mathsf{el}}}^{\mathsf{el \sqcup nondet}}$@ @$\mathsf{PE}^{\mathsf{el}}$@ {@$\phi$@}  v : t { @$\phi'$@} = 
	@$\mathsf{PE}^{\mathsf{el \sqcup nondet}}$@ {@$\phi$@}  v : t  { @$\phi'$@}

@${\mathcal{T}_{\mathsf{stnon}}}^{\mathsf{parser}}$@ @$\mathsf{PE}^{\mathsf{parser}}$@ {@$\phi$@}  v : t { @$\phi'$@} = 
	@$\mathsf{PE}^{\mathsf{parser}}$@ {@$\phi$@}  v : t result  { x = Inl (v) @$\wedge$@ @$\phi'$@}

\end{lstlisting}

  


\subsection{Fine-Grained Effects}
\label{sec:discussion}

\label{sec:mod-effect}

Because our language has multiple effects, we use a localized
effect typing system~\citep{param-monad-effect,lightweightmonadicML,multimonfstar}
to reason locally over computations with different effects. Thus, a
type-schema specification for a $\lambda_{sp}$ expression \eip\ has
an effect-label annotation $\el$ capturing the scope of \eip's
effect. Moreover, the pre- and post-conditions for \eip's
specification define relations between its own output and these effects.

For example, an \textit{exception-free} expression is not forced to
mention an exception effect, and a pure transition with no effect
should not mention how the state changes. However, given a single
specification monad, the annotation burden of differentiating
and enumerating these effects is problematic. To illustrate, 
consider the following simple \name program.
\begin{lstlisting}[escapechar=\@,basicstyle=\small\sf,language=ML]
   char 'A' >>= @$\lambda$@ x. char 'B' >>= @$\lambda$@ y. return [x] ++ [y]
\end{lstlisting}
that monadically sequences two character parsers
storing their outputs in two lists and then invoking a pure list
append function (++) to append them. In the absence of
effects, the expected type for append can be defined using a qualifier {\sf len} for list's length as follows\footnote{In the remainder of the paper, we elide
  explicit quantification of \S{h}, \S{h'} and $\nu$ in pre- and post-conditions in specifications to ease
  readability.}:
\begin{lstlisting}[escapechar=\@,basicstyle=\small\sf,language=ML]
   ++ : l1 : @$\alpha$@ list  -> l2 : @$\alpha$@ list -> {v : (@$\alpha$@ list) | len @$\nu$@ = len l1 + len l2}
\end{lstlisting}
However, since the characters parsers have state and exception effects, the
type for the subexpression {\sf\small (char 'A' >>= $\lambda$ x. char
  'B')} synthesized using the typing rule for the bind combinator
would be:
\begin{lstlisting}[escapechar=\@,basicstyle=\small\sf,language=ML]
   (char 'A' >>= $\lambda$ x. char 'B') : @$\mathsf{PE^{state \sqcup exc}}$@ {true}  @$\nu$@ : char result 
                    { ( @$\nu$@ = Inl (v1) =>  x = 'A' @$\wedge$@ v1 = 'B' @$\wedge$@ len (sel (h', inp)) = len (sel h inp) @$-$@ 2  
                                       @$\wedge$@ (@$\nu$@ = Inr (Err) => len (sel (h', inp)) = len (sel h inp) @$-$@ 1) }
\end{lstlisting}
This type makes it impossible to bind this subexpression with the
later subexpression {\sf\small ($\lambda$ y. return [x] ++ [y])}, as
their effect labels do not match; see typing rule ({\sc T-p-bind}). Thus the above expression becomes {\it ill-typed} and cannot be written in \name. 
To allow this very trivial binding, we need to manually strengthen the
type of the append function, so that the type of the application term
{\sf\small ([x] ++ [y])} synthesized using the {\sc T-app} rule
matches the effect label of the subexpression ((char 'A' >>= $\lambda$ x. char 'B')):
\begin{lstlisting}[escapechar=\@,basicstyle=\small\sf,language=ML]
   (++) : l1 : @$\alpha$@ list  -> l2 : @$\alpha$@ list -> @$\mathsf{PE^{state \sqcup exc}}$@ {true}    @$\nu$@ : (@$\alpha$@ list) result       
                                         { (Inl x = (@$\nu$@) => sel h' inp = sel h inp @$\wedge$@ sel h' ist = sel h ist @$\wedge$@ 
                                           len x= len l1 + len l2)  @$\wedge$@ (@$\nu$@ = Inr (Err) => h' = h }
\end{lstlisting}


Unfortunately, manually lifting effect labels for each expression in
this way is impractical. To solve this, we {\it weaken} the typing
rule for monadic bind and allow for {\it effect-local} reasoning to
obviate the need for annotating effect-behavior outside of an
expression's effect scope. To support such reasoning, we need some
additional machinery. First, we require an {\it effect-label} lattice
with an ordering relation ($\leq$). The lattice is presented as a
Hasse diagram in Figure~\ref{fig:lattice} and defines the join
\begin{wrapfigure}{r}{.40\textwidth}
\begin{center}
    \centering
     \includegraphics[width=1.0\linewidth]{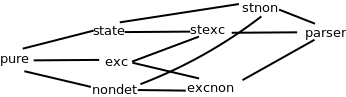}
    
\end{center}
\caption{The effect lattice, a powerset lattice over elements state, exc and nondet. Ordered left-to-right, the bottom element is pure, while the top element is parser. stnon is a shorthand for ({state $\sqcup$ nondet}), stexc = ({state $\sqcup$ exc}), excnon = ({exc $\sqcup$ nondet}) }
\label{fig:lattice}
\end{wrapfigure}
operation ($\sqcup$) on effect-labels used in our typing rules. The
least element of the lattice is the pure effect, while the top element
is the parser effect that subsumes all possible effects expressible in
the language. Second, we define a finite set of specification
monad morphisms (similar to the notion of Haskell monad
transformers~\cite{monadx}) between any two specification monads
parameterized with different effect labels. We represent such
morphisms using a function ${\mathcal{T}_{\el_1}}^{\el_2}$ that
transforms a specification monad parameterized with $\el_1$, to a
specification monad parameterized with $\el_2$ given $\el_1 \leq
\el_2$.

%
%
%
%
For example, the ${\mathcal{T}_{\mathsf{pure}}}^{\mathsf{state}}$ lifts a pure
specification to a specification capturing a {\sf state} effect
(i.e. a specification monad parameterized with effect-label {\sf
  state}) with a trivial pre-condition and a post-condition that
captures refinements in the pure specification and establishes
equivalence of the heap in the pre- and post-states~\footnote{The full list of these
morphisms is available in the supplementary material}.
\begin{lstlisting}[escapechar=\@,basicstyle=\small\sf,language=ML]
   @${\mathcal{T}_{\mathsf{pure}}}^{\mathsf{state}}$@ ({ @$\nu$@ : t | @$\phi$@}) = @$\mathsf{PE}^{\mathsf{state}}$@ {true}  @$\nu$@ : t { @$\phi$@ ($\nu$) @$\wedge$@ h' = h}
\end{lstlisting}
Finally, we also define two new typing rules in
Figure~\ref{fig:lifted-typing}; a lifting rule {\sc T-l-id} asserting
that if an expression has a type $\tau$ =
$\mathsf{PE}^{\mathsf{\el_1}}$\{$\phi$\} v : t \{$\phi'$\} under some
$\Gamma$, and $\el_1 \leq \el_2$ then the expression also has a lifted
type
${\mathcal{T}_{\mathsf{\el_1}}}^{\mathsf{\el_2}}$($\tau$). 
Rule {\sc T-l-bind} defines a {\it weakening} of  {\sc
  T-p-bind} (Figure 6 in the main paper) that lifts the type of
each of the arguments to (>>=) into a join effect ({\el}), and then
defines the binding semantics over these lifted types in a fashion similar to the
original rule.
\begin{figure}
\begin{minipage}[t]{.3\textwidth}
\inference[{\sc T-l-id}]{
      \tau = \mathsf{PE}^{\el} \ \{ \phi \}\, \nu\, :\, \S{t}  \{ \phi' \} & 
      \Gamma \vdash\, \eip\, :\, \tau  & 
      \tau' = \mathcal{T}_{\mathsf{\el}}^{\h{\mathsf{\el}}} (\tau)}
	    {  \Gamma \vdash\, \eip\, :\, \tau' }
\end{minipage}\\[10pt]
\hfill
\bigskip	     
\begin{minipage}[t]{.5\textwidth}
\inference[{\sc T-l-bind}]{
      \tau_1 = \mathsf{PE}^{\el_1} \ \{ \phi_1 \}\, \nu\, :\, \S{t}  \{ \phi_{1'} \} & 
      \tau_2 = \mathsf{PE}^{\el_2} \
              \{ \phi_2 \} \ \nu' : \S{t'} \ \{ \phi_{2'} \} &
       \hat{\mathsf{\el}} = \mathsf{\el_1} \sqcup \mathsf{\el_2} \\      
      \Gamma\ \vdash\ p\, :\, \tau_1  &
      \Gamma\ \vdash\ \eip\, :\, (\xp : \tau) \rightarrow \tau_2 \\
      \mathcal{T}_{\mathsf{\el_1}}^{\h{\mathsf{\el}}} (\tau_1) = \mathsf{PE}^{\h{\mathsf{\el}}} \ \{ \h{\phi_1} \}\, \nu\, :\, \h{\S{t}}  \{ \h{\phi_{1'}} \} &  \mathcal{T}_{\mathsf{\el_2}}^{\h{\mathsf{\el}}} (\tau_2) = \mathsf{PE}^{\h{\mathsf{\el}}} \ \{ \h{\phi_2} \}\, \nu'\, :\, \h{\S{t'}}  \{ \h{\phi_{1'}} \}  \\
      \Gamma' = \Gamma, \xp : \tau, \mathsf{h_i : heap} \quad\ \mathsf{h_i}\ \mbox{fresh}}
	    { \begin{array}{@{}c@{}}
		  \Gamma' \vdash p\ \textnormal{>>=}\ \eip : 
		  \mathsf{PE}^{\h{\el}} \ \{ \forall \mathsf{h}.\, \phi_1 \, \mathsf{h}\,  \wedge 
                                               \h{\phi_{1'}} (\mathsf{h}, x, \mathsf{h_i})  => \h{\phi_2} \ \mathsf{h_i} \} \\ 
  			           \mkern20mu\nu'\,:\, \h{\S{t'}}\ \S{result} \\  
				  \mkern125mu\{ \forall \mathsf{h}, \nu', \mathsf{h'}. 
				  ({x} \neq \S{Err} => \nu' = \S{Inl}\ \S{y} \wedge \h{\phi_{1'}} (\mathsf{h}, x, \mathsf{h_i})\\
                                   \mkern385mu\wedge\ \h{\phi_{2'}} ({\mathsf{h_i}, \S{y}, \mathsf{h'}}))\ \wedge\\
                                   \mkern235mu({x} = \S{Err} => \nu' = \S{Inr}\ {\S{Err}} \wedge \h{\phi_{1'}} (\mathsf{h}, x, \mathsf{h_i}))                           \} 
	     \end{array}\hfill}
\end{minipage}\\[10pt]	     
\caption{Typing semantics for lifting local effects}
\label{fig:lifted-typing}
\end{figure}

Revisiting our append example above, the given expression can be now correctly typed. Rather than manually strengthening the type for
append, \name uses these typing rules to do lifting, and synthesize a type for the overall expression (simplified for elucidation). 
\begin{lstlisting}[escapechar=\@,basicstyle=\small\sf,language=ML]
   char 'A' >>= @$\lambda$@ x. char 'B' >>= @$\lambda$@ y. return [x] ++ [y] : 
   @$\mathsf{PE^{state \sqcup exc}}$@ {true} @$\nu$@ : (char list) result 
             { ( @$\nu$@ = Inl (v1) =>  x = 'A' @$\wedge$@ y = 'B' @$\wedge$@ len (sel (h', inp)) = len (sel h inp) @$-$@ 2 @$\wedge$@ len (v1) = 2 
    @$\wedge$@ (@$\nu$@ = Inr (Err) => len (sel (h', inp)) = len (sel h inp) @$-$@ 1) @$\vee$@ len (sel (h', inp)) = len (sel h inp) @$-$@ 2)}
\end{lstlisting}

\end{document}


\title{\name: Automated Type-Based Verification of Parser Combinator
  Programs
  }
\subtitle{Supplemental Material with the main submission
  }
\author{Anonymous}
\authornote{Both authors contributed equally to this research.}
\email{trovato@corporation.com}
\orcid{1234-5678-9012}
\author{Anonymous}
\authornotemark[1]
\email{webmaster@marysville-ohio.com}
\affiliation{%
  \institution{Institute for Clarity in Documentation}
  \streetaddress{P.O. Box 1212}
  \city{Dublin}
  \state{Ohio}
  \country{USA}
  \postcode{43017-6221}
}

\renewcommand{\shortauthors}{Trovato and Tobin, et al.}

\begin{abstract}
  
\end{abstract}

\ccsdesc[500]{Computer systems organization~Embedded systems}
\ccsdesc[300]{Computer systems organization~Redundancy}
\ccsdesc{Computer systems organization~Robotics}
\ccsdesc[100]{Networks~Network reliability}


\maketitle

\section{$L$ : The Underlying Language of Transducer Labels}

\begin{table}[ht]
\centering 
\small \begin{tabular}{c c l} 
&x, y, z, r, $\nu$ & $\in$ Program Variables, c $\in$ constants \\
& $\mathsf{r}$ $\in$ Heap Variables & \\
&& first order functions and expressions  \\
fun (f)& := & $\lambda$ x : t .  e \\ 
value (v) & :=  & C $\overline{{\sf x}}$ $\mid$ c \\ 
exp (e) & := &  v $\mid$ {\sf x}\\ 
&& $\mid$  f v $\mid$ let x = e in e  $\mid$ fold f v v \\
&& $\mid$ match v with C $\overline{\mathsf{e}}$ -> e else e \\
&& $\mid$ !r $\mid$ r := v  $\mid$ ref v 
\end{tabular}%
\caption{Core syntax for $L$-expressions} 
\label{table:Ssigma} 
\end{table} 

%

$L$ has standard runtime semantics, closely resembling a functional language with references like ML without higher-order functions and excludes references as values. Thus references cannot be passed or returned as function argument or return value.

\subsection{Typing Semantics for $L$ expressions}
Fig.~\ref{fig:Ssemantics-lexp} presents typing semantics for $L$-expressions. Each typing judgment is of the form 
$\Sigma \vdash$ {\sf e}  : $\tau$, saying, in a typing environment $\Sigma$, an $L$ expression has a type $\tau$.

\begin{figure}[h]
\begin{flushleft}
\textbf{$L$ expression Typing }\,\fbox{\small
        $\Sigma \vdash$ {\sf e} : $\tau$ 
           
}
\end{flushleft}
\begin{minipage}{0.3\textwidth}
\small \begin{center}
\inference[T-ref]{\Sigma \vdash {\sf v} : t} 
				 {
				\Sigma \vdash :  {\sf ref \ v} :  {\sf ref} \  t
			         }
\end{center}
\end{minipage}
\begin{minipage}{0.5\textwidth}
\small \begin{center}
\inference[T-app]{\Sigma \vdash {\sf f} : ({\sf x} : t \rightarrow \textnormal{state} \ \{ \phi \} \ t1  \{ \phi' \} & \Sigma \vdash {\sf v1} : t } 
				{\Sigma \vdash {\sf f} \ {\sf v}  : \textnormal{state} \ \{ \phi[{\sf v1}/{\sf x}] \} \ t1  \{ \phi'[{\sf v1}/{\sf x}] \} }
\end{center}
\end{minipage}

\bigskip	
\begin{minipage}{0.4\textwidth}
\small \begin{center}
\inference[T-match]{\Sigma \vdash {\sf v} : \tau & \Sigma \vdash {\sf C} {\overline{ {\sf e}}} : \tau \\
		    \Sigma_{true} = \Sigma, {\sf v} :\tau, {\sf v} = {\sf C} {\overline{{\sf e}}} & 
		    \Sigma_{false} = \Sigma, {\sf v}:\tau, {\sf v} \neq {\sf C} {\overline{{\sf e}}} \\
		     \Sigma_{true} \vdash {\sf e1} : \tau_m & \Sigma_{false} \vdash {\sf e2} : \tau_m}
				{\Sigma \vdash \textnormal{ match v with C} {\overline{{\sf e}}} \rightarrow \textnormal{e1 else e2} : \tau_m }
\end{center}
\end{minipage}
\begin{minipage}{0.5\textwidth}
\small \begin{center}
\inference[T-let]{\Sigma \vdash e1 : \textnormal{state} \ \{ \phi1 \} \ t1  \{ \phi1' \} \\ 
				 \Sigma, x : t1 \vdash e2 : \textnormal{state} \ \{ \phi2 \} \ t2 \ \{ \phi2' \}}
				{
				\Sigma \vdash \textnormal{let x = e1 in e2} :  
				  \textnormal{state} \ \{ \mathsf{bind_{pre}} \} \  t2  \{\mathsf{bind_{post}} \}       
			         }
\end{center}
\end{minipage}
\bigskip
\begin{minipage}{0.5\textwidth}

\small \begin{center}
\inference[T-fun]{\Sigma, {\sf x} : \tau \vdash e : \textnormal{state} \ \{ \phi \} \ t1  \{ \phi' \}} 
				{
				\Sigma \vdash \textnormal{$\lambda$ x : t}. {\sf e} : (x : \tau) \rightarrow \textnormal{state} \ \{ \phi \} \ t1  \{ \phi' \} }
\end{center}
\end{minipage}
\begin{minipage}{0.4\textwidth}
\small \begin{center}
\inference[T-deref]{\Sigma \vdash {\sf r} : {\sf ref} t} 
				 {
				\Sigma \vdash :  {\sf ! \ r} :  t
			         }
\end{center}
\end{minipage}
\bigskip
\begin{minipage}{0.4\textwidth}

\small \begin{center}
\inference[T-assign]{\Sigma \vdash {\sf r} : {\sf ref} t & \Sigma \vdash {\sf v} : t} 
				 {
				\Sigma \vdash :  {\sf r := v} : {\sf unit}
			         }
\end{center}
\end{minipage}
\begin{minipage}{0.5\textwidth}
\small \begin{center}
\inference[T-fold]{\Sigma \vdash {\sf f} : ( x: t  \rightarrow ({\sf state} \{ \phi \} \ y : t' \ \{ \phi'\})  
				  \\ \Sigma \vdash {\sf ac} : t' & \Sigma \vdash {\sf l} : t &
				   \Sigma \vdash {\sf Inv} \\
				   \Sigma \vdash \mathsf{Inv_{pre}} \wedge \mathsf{Inv_{ind}} \wedge \mathsf{Inv_{post}}  } 
				 {
				\Sigma \vdash :  {\sf fold \ f \ ac \ l} :  ({\sf state} \{ {\sf Inv \ h} \} \ {\sf v} : t' \ \{ {\sf Inv \ h'}\})
			         }
\end{center}
\end{minipage}

\caption{Typing Semantics for $L$ expressions}
\label{fig:Ssemantics-lexp}
\end{figure}

\pagebreak
\section{\name\ expression Dynamic Semantics}

\begin{figure}[h]
\begin{flushleft}
\fbox{\small
        $(\mathcal{H}; \mathsf{e}) \Rrightarrow (\mathcal{H'}; \mathsf{v}) $ 
           
} 
\bigskip
\end{flushleft}
\begin{minipage}{0.4\textwidth}
\small \begin{center}
\inference[{\sf CS}-$\iota$]{ (\mathcal{H}; {\sf AP}) \Rightarrow^{*} (\mathcal{H'}, {\sf v})}  
				  {(\mathcal{H}; \iota \mathsf{AP}) \Rrightarrow (\mathcal{H'}; {\sf v}) }

\end{center}
\end{minipage}
\hfill
\begin{minipage}{0.5\textwidth}
\small \begin{center}
\inference[{\sf CS-bind}]{ (\mathcal{H}; \mathsf{CP_1}) \Rrightarrow (\mathcal{H'}; {\sf v1}) \\
				   (\mathcal{H'}; \mathsf{CP_2}[\ \mathsf{v1} / \mathsf{x}]\ \Rrightarrow (\mathcal{H''}; {\sf v2})}  
				  {(\mathcal{H}; (\mathsf{CP_1} \textnormal{>>=} (\lambda {\sf x : t}. \mathsf{CP_2}))) \Rrightarrow (\mathcal{H''}; {\sf v2})}

\end{center}
\end{minipage}
\hfill
\bigskip
\begin{minipage}{0.5\textwidth}
\small \begin{center}
\inference[{\sf CS-seq}]{ (\mathcal{H}; \mathsf{CP_1}) \Rrightarrow (\mathcal{H'} ;{\sf v1}) \\
				  (\mathcal{H'}[\ {\sf inp} \mapsto {\sf v1}]\ ;\mathsf{CP_2}) \Rrightarrow (\mathcal{H''}; {\sf v2})} 
				  {(\mathcal{H};(\mathsf{CP_1} \textnormal{o} \mathsf{CP_2})) \Rrightarrow (\mathcal{H''}; {\sf v2})}

\end{center}
\end{minipage}
\begin{minipage}{0.4\textwidth}
\small \begin{center}
\inference[{\sf CS-Err}]{ (\mathcal{H}; \mathsf{CP}) \Rrightarrow (\mathcal{H'}; {\sf Err})  &
			 \odot = \{ \textnormal{o}, \textnormal{>>=} \}}
				  {(\mathcal{H}; (\mathsf{CP} \ \odot \ \mathsf{CP_i})) \Rrightarrow (\mathcal{H'}; {\sf Err})) }

\end{center}
\end{minipage}
\bigskip
\begin{minipage}{0.5\textwidth}
\small \begin{center}
\inference[{\sf CS-L}]{ (\mathcal{H}; \mathsf{CP_1}) \Rrightarrow (\mathcal{H'}; {\sf v1}) }
				  {(\mathcal{H}; (\mathsf{CP_1} \mid \mathsf{CP_2})) \Rrightarrow (\mathcal{H'}; {\sf v1})) }

\end{center}
\end{minipage}
\bigskip
\begin{minipage}{0.4\textwidth}
\small \begin{center}
\inference[{\sf CS-R}]{ (\mathcal{H}; \mathsf{CP_2}) \Rrightarrow (\mathcal{H''}; {\sf v2}) }
				  {(\mathcal{H}; (\mathsf{CP_1} \mid \mathsf{CP_2})) \Rrightarrow (\mathcal{H''}; {\sf v2})) }

\end{center}
\end{minipage}
\begin{minipage}{0.5\textwidth}
\small \begin{center}
\inference[{\sf CS-loop}]{ (\mathcal{H}; \mathsf{CP}) \Rrightarrow (\mathcal{H'}; {\sf v}) & {\sf v} \neq \mathsf{Err} \\
		(\mathcal{H'} ; {\sf f \ v \ b}) \Downarrow (\mathcal{H'}; {\sf b'}) \\ 
		(\mathcal{H'} ; {\sf g \ b'}) \Downarrow (\mathcal{H'}; {\sf false}) \\
		(\mathcal{H'}; (\mathsf{foldT \ {\sf CP} \ f \ g \ b'})) \Rrightarrow (\mathcal{H''}; {\sf v'})}  
		{(\mathcal{H}; (\mathsf{foldT \ {\sf CP} \ f \ g \ b})) \Rrightarrow (\mathcal{H''};{\sf v'})}

\end{center}
\bigskip
\end{minipage}
\begin{minipage}{0.4\textwidth}
\small \begin{center}
\inference[{\sf CS-app}]{ (\mathcal{H}; \mathsf{CP}[\ \mathsf{v} / \mathsf{x}]) \ \Rrightarrow (\mathcal{H'}; {\sf v1})}  
				  {(\mathcal{H}; (\lambda {\sf x : t}. \mathsf{CP}) ({\sf v})) \Rrightarrow (\mathcal{H'}; {\sf v1})}

\end{center}
\end{minipage}

\bigskip
\bigskip
\begin{minipage}{0.5\textwidth}
\small \begin{center}
\inference[{\sf CS-break}]{ (\mathcal{H}; \mathsf{CP}) \Rrightarrow (\mathcal{H'}; {\sf v}) & {\sf v} \neq \mathsf{Err} \\
		      (\mathcal{H'} ; {\sf f \ v \ b}) \Downarrow (\mathcal{H'}; {\sf b'}) \\ 
		      (\mathcal{H'} ; {\sf g \ b'}) \Downarrow (\mathcal{H'}; {\sf true})}       
		    {(\mathcal{H}; (\mathsf{foldT \ {\sf CP} \ f \ g \ b})) \Rrightarrow (\mathcal{H'};{\sf b'})}

\end{center}
\end{minipage}
\begin{minipage}{0.4\textwidth}
\small \begin{center}
\inference[{\sf CS-backtrack}]{ (\mathcal{H}; \mathsf{CP}) \Rrightarrow (\mathcal{H'};{\sf v}) & {\sf v} = \mathsf{Err} }
				{(\mathcal{H}; (\mathsf{foldT \ {\sf CP} \ f \ g \ b})) \Rrightarrow (\mathcal{H'}; {\sf b})}

\end{center}
\end{minipage}
\bigskip
\begin{minipage}{0.4\textwidth}
\small \begin{center}
\inference[{\sf CS-Let}]{ (\mathcal{H}; \mathsf{CP_2}[\ \mathsf{CP_1} / \mathsf{x}]) \ \Rrightarrow (\mathcal{H'}; {\sf v})}  
				  {(\mathcal{H}; {\sf let} \ {\sf x} \ = \ \mathsf{CP_1} \ {\sf in} \ \mathsf{CP_2}  \Rrightarrow (\mathcal{H'}; {\sf v})}

\end{center}
\end{minipage}

\caption{Evaluation rules for Morpheus Combinators, $\Rightarrow^{*}$ represents multi-step evaluation for atomic-parsers (Figure~\ref{fig:Ssemantics-pst}), 
$\Downarrow$ is evaluation relation over $L$ terms}
\label{fig:Ssemantics-cpst}
\end{figure}

\begin{figure*}[t]
\begin{flushleft}
\bigskip
{\bf Expression Typing}\fbox{\small
$\Gamma \vdash$ {\sf e} : $\tau$
           
}
\end{flushleft}
\begin{minipage}{0.40\textwidth}
\small \begin{center}
\inference[{\sf T-Identity}]{\Gamma \vdash {\sf AP} : \textnormal{eff} \ \{ \phi \} \ \tau  \{ \phi' \}}
			{\Gamma \vdash (\iota {\sf AP}) : \textnormal{eff} \{ \phi \} \ \tau \ \{ \phi' \} }
\end{center}
\end{minipage}
\hfill
\begin{minipage}{0.50\textwidth}
\small \begin{center}
\inference[{\sf T-Seq}]{\Gamma \vdash \mathsf{CP_1} : \textnormal{eff} \ \{ \phi1 \} \ \tau_1  \{ \phi1' \} \\ 
		  \Gamma, {\sf x} : {\tau_1} \vdash \mathsf{CP_2}  : \textnormal{eff} \ \{ \phi2 \} \ \tau_2 \ \{ \phi2' \} }
				  { \begin{array}{@{}c@{}}  
				   \Gamma, {\sf x} : {\tau_1}, {\sf h_i} : {\sf heap} \vdash (\mathsf{CP_1} \textnormal{o} \mathsf{CP_2}) : \\
				  \textnormal{eff} \ \{ \forall {\sf h, x, h_i}. \phi1 \ {\sf h}  \wedge (\phi1' ({\sf h, x, h_i})  => [\phi2]({\sf x}/{\sf inp})) {\sf h_i} \} 
				  \tau_2 \\
				  \{ \forall {\sf h, y, h'}. \phi1' ({\sf h, x, h_i}) \wedge \phi2' ({\sf h_i, y, h'}) \}
				     \end{array}			  
				  } 
\end{center}
\end{minipage}
\hfill
\bigskip
\begin{minipage}{0.40\textwidth}
\small \begin{center}
\inference[{\sf T-Let}]{\Gamma \vdash \mathsf{CP_1} : \tau_1 & \Gamma, {\sf x } : \tau_1 \vdash \tau_2 \ }
			{\Gamma \vdash \textnormal{let x = } \mathsf{CP_1} \ {\sf in} \ \mathsf{CP_2} : \tau_2 }
\end{center}
\end{minipage}
\begin{minipage}{0.50\textwidth}
\small \begin{center}
\inference[{\sf T-Abs}]{\Gamma, {\sf x } : {\sf t} \vdash {\sf CP} : \tau}   
	       		{\Gamma \vdash (\lambda {\sf x} : {\sf t}. \mathsf{CP})  : ({\sf x : t}) \rightarrow \tau } 
\end{center}
\end{minipage}
\bigskip
\begin{minipage}{0.50\textwidth}
\small \begin{center}
\inference[{\sf T-Choice}]{\Gamma \vdash \mathsf{CP_1} : \textnormal{eff} \ \{ \phi1 \} \ \tau  \{ \phi1' \} \\ 
				  \Gamma \vdash \mathsf{CP_2} : \textnormal{eff} \ \{ \phi2 \} \ \tau \ \{ \phi2' \} \\ 
				  \textnormal{eff'} = {\sf eff}  \sqcup {\sf nondet}} 
				 { \begin{array}{@{}c@{}}  
				   \Gamma \vdash (\mathsf{CP_1} \textnormal{<|>}  \mathsf{CP_2}) : \\ 
				   \textnormal{eff'} \ \{ (\phi1 \wedge \phi2) \} \ \tau \ \{ (\phi1'  \lor \phi2') \} 
				     \end{array}			  
				 } 
\end{center}  
\end{minipage}
\bigskip
\begin{minipage}{0.40\textwidth}
\small \begin{center}
\inference[{\sf T-Fold}**]{  \Gamma \vdash  \mathsf{CP} : \textnormal{eff} \ \{ \phi  \} \ \tau \ \{ \phi' \} & \Gamma \vdash {\sf acc} : \tau_1 \\
		      \Gamma \vdash {\sf f} : \tau_1 \rightarrow \tau \rightarrow {\sf pure} \ \{ \nu : {\sf shape} (\tau_1) | \phi_{\sf f} \}   \\
	             \Gamma \vdash {\sf g} :  \tau_1 \rightarrow {\sf pure} \ {\sf bool} \\   
	             \Gamma_{\mathsf{e}} =  \mathsf{Inv_{start}} \wedge \mathsf{Inv_{ind}} \wedge \mathsf{Inv_{break}} }
		    {\Gamma, \Gamma_{\mathsf{e}} \vdash (\mathsf{foldT} \ {\sf CP} \ {\sf f} \ {\sf g} \ {\sf acc}) :  \textnormal{eff} \ \{\sf P \} \ \tau_1 \ \{ \sf Q \} }
\end{center}  
\hfill
\end{minipage}
\small \begin{center}
\inference[{\sf T-Bind}]{\Gamma \vdash \mathsf{CP_1} : \textnormal{eff} \ \{ \phi1 \} \ \tau_1  \{ \phi1' \}  &
				  \Gamma, {\sf x} : {\tau_1} \vdash \mathsf{CP_2}  : \textnormal{eff} \ \{ \phi2 \} \ \tau_2 \ \{ \phi2' \}  & \Gamma' = \Gamma, {\sf x} : {\tau_1}, {\sf h_i : heap}}
			    { \begin{array}{@{}c@{}}  
				  \Gamma' \vdash (\mathsf{CP_1} \textnormal{>>=} \lambda {\sf x} : \tau_1. \mathsf{CP_2}) : 
				  \textnormal{eff} \ \{ \forall {\sf h}. \phi1 \ {\sf h}  \wedge \phi1' ({\sf h, x, h_i})  => \phi2 \ {\sf h_i} \}
				  \ \tau_2 \  
				  \{ \forall {\sf h, y, h'}. \phi1' ({\sf h, x, h_i}) \wedge \phi2' ({\sf h_i, y, h'}) \}
				     \end{array}			  
				 } 
\end{center}
\bigskip
\small \begin{center}
\inference[{\sf T-Sub}]{\Gamma \vdash \mathsf{e_1} : \tau_1 & \Gamma \vdash \tau_1 <: \tau_2 \ }
			{\Gamma \vdash \mathsf{e_1} : \tau_2 }
\end{center}

\hfill
\begin{flushleft}
\begin{tabbing}
\small {**}
$\mathsf{Inv_{start}}$ \ {\sf =}  P \ $\mathsf{h_0}$ \ => ({\sf Inv $\mathsf{h_0}$ acc})  \\ 
\small $\mathsf{Inv_{ind}}$ \ {\sf =}  $\forall$ {\sf h, x, y, y' h'}. ({\sf Inv} \ {\sf h} \ {\sf y}) $\wedge$ (g y = false) {\sf =>} \ {\sf $\phi$ h} \ {\sf =>} (( $\phi'$  {\sf  h \ x \ h'} $\wedge$ $\phi_{\sf f}$ {\sf y' y x}) {\sf =>} ({\sf Inv} \ {\sf h' \ y'}))  \\
\small $\mathsf{Inv_{break}}$ \ {\sf =}  $\forall$ {\sf h, y}. (g y = true) $\wedge$ {\sf Inv} \ h \ y {\sf =>} Q {\sf h \ y}\\ 

\end{tabbing}

\end{flushleft}

\begin{flushleft}
\bigskip
{\bf Subtyping}\fbox{\small
$\Gamma \vdash$ $\tau_1$ <: $\tau_2$
           
}
\end{flushleft}
\begin{minipage}{0.40\textwidth}
\small \begin{center}
 \inference[{\sf T-Sub-Base}]{ \Gamma \vdash \{ \nu : {\sf t} \mid \phi_1 \} & \Gamma \vdash \{ \nu : {\sf t} \mid \phi_2 \} \\
			  \Gamma \vDash \phi_1 => \phi_2 }
				{ 
				\Gamma \vdash \{ \nu : {\sf t} \mid \phi_1 \} <: \{ \nu : {\sf t} \mid \phi_2
				  } 
 
\end{center}  
\end{minipage}
\begin{minipage}{0.50\textwidth}
\small \begin{center}
 \inference[{\sf T-Sub-Arrow}]{ \Gamma \vdash \tau_{21} < \tau_{11} & \Gamma \vdash \tau_{12} <: \tau_{22}}
			  { \Gamma \vdash ({\sf x }: \tau_{11}) \rightarrow \tau_{12}  <: ({\sf x} : \tau_{21}) \rightarrow \tau_{22} } 
 
\end{center}  
\end{minipage}

\bigskip
\begin{minipage}{1.0\textwidth}
\small \begin{center}
 \inference[{\sf T-Sub-Comp}]{ \Gamma \vDash \phi_2 => \phi_1  &    \Gamma \vdash \tau_1 <: \tau_2 &
				    \Gamma \vdash \mathsf{eff_1} \sqsubseteq \mathsf{eff_2} &   \Gamma, \phi_2 \vDash (\phi_{1'} => \phi_{2'})}
				{\begin{array}{@{}c@{}}  
				    \Gamma  \vdash  
				      \mathsf{eff_1} \ \{ \phi_1 \} \ \tau_1 \ \{ \phi_{1'} \} 
					 <:
				      \mathsf{eff_{2}} \ \{ \phi_2 \} \ \tau_2 \ \{ \phi_{2'} \} 
				    \end{array}			  
				  } 
  

\end{center}  
\end{minipage}
\caption{Static Semantics for Morpheus Combinators.}
\label{fig:Styping-transducer2}
\end{figure*}

\clearpage
  
\section{PST : Dynamic and Typing Semantics}
\subsection{PST and Path Semantics}
\begin{figure}[h]
\begin{flushleft}
\fbox{\small
        $(\mathcal{H}, q); c \Rightarrow (\mathcal{H}, q') \mid (\mathcal{H}, q); c \Rightarrow (\mathcal{H}, \mathsf{v}) $ 
           
}       
\end{flushleft}
\bigskip
\begin{minipage}{0.4\textwidth}
\small \begin{center}
\hspace{5ex} \inference[S-P-$\mathcal{I}$]{ (\mathcal{H}; \mathcal{I}) \Downarrow (\mathcal{H'}; {\sf v})}  
				  {(\mathcal{H}, q_{0}); \mathcal{I} \Rightarrow (\mathcal{H'}, q_{0}) }

\end{center}
\end{minipage}
\bigskip
\begin{minipage}{0.4\textwidth}
\small \begin{center}
\inference[{\sf S-P-Ind}]{  (\mathcal{H}, p); \mathsf{p_{k-1}}^{(p, q)} \Rightarrow (\mathcal{H'}, q) & 
				      \delta = (q \xrightarrow[]{ \mathsf{\phi} \mathbf{/} \mathsf{f} } r) \\
				      \llbracket \mathsf{\phi} \rrbracket \mathcal{H'} = true  	&(q \xrightarrow[]{ \mathsf{\phi'} / \mathsf{f'} } r') \in $R$ \setminus \delta  \\  
					  \llbracket \mathsf{\phi'} \rrbracket \mathcal{H'} = false &   (\mathcal{H'}; \mathsf{f}) \Downarrow (\mathcal{H''}; \_)
					  }  
				  {(\mathcal{H}, p);\mathsf{p_k}^{(p, r)} \Rightarrow (\mathcal{H''}, r) }

\end{center}
\end{minipage}
\bigskip
\begin{minipage}{0.4\textwidth}
\small \begin{center}
\hspace{5ex} \inference[S-P-Base]{ \delta = (p \xrightarrow[]{ \mathsf{\phi} \mathbf{/} \mathsf{f} } q) \\
					\llbracket \mathsf{\phi} \rrbracket \mathcal{H} = true &  (p \xrightarrow[]{ \mathsf{\phi'} / \mathsf{f'} } q') \in $R$ \setminus \delta  \\  
					  \llbracket \mathsf{\phi'} \rrbracket \mathcal{H} = false \\
					  (\mathcal{H}; \mathsf{f}) \Downarrow (\mathcal{H'}; \_)
					  }  
				  {(\mathcal{H}, p);\delta \Rightarrow (\mathcal{H'}, q) }

\end{center}
\end{minipage}
\begin{minipage}{0.4\textwidth}
\small \begin{center}
\inference[S-P-Err]{ \delta = (p \xrightarrow[]{ \mathsf{\phi} / \mathsf{f} } err) & ( p \xrightarrow[]{ \mathsf{\phi'} / \mathsf{f'} } q') \in R \setminus \delta \\
					\llbracket \mathsf{\phi} \rrbracket \mathcal{H} = true & ( p \xrightarrow[]{ \mathsf{\phi'} / \mathsf{f'} } q') \  
					  \llbracket \mathsf{\phi'} \rrbracket \mathcal{H} = false \\
					  (\mathcal{H}; \mathsf{f}) \Downarrow (\mathcal{H'}; \_)
					  }  
				  {(\mathcal{H}, p); \delta \Rightarrow (\mathcal{H'}, {\sf Err}) }

\end{center}
\end{minipage}

\bigskip

\begin{minipage}{0.4\textwidth}
\small \begin{center}
\inference[S-P-Loop]{ \delta = (p \xrightarrow[]{ \mathsf{\phi} \mathbf{/} \mathsf{f} } p) \\
					\llbracket \mathsf{\phi} \rrbracket \mathcal{H} = true \\  
					(p \xrightarrow[]{ \mathsf{\phi'} / \mathsf{f'} } q') \in $R$ \setminus \delta  \\  
					  \llbracket \mathsf{\phi'} \rrbracket \mathcal{H} = false \\
					  (\mathcal{H}; \mathsf{f}) \Downarrow (\mathcal{H'}; \_)
					  }  
				  {(\mathcal{H}, p);\delta \Rightarrow (\mathcal{H'}, p) }

\end{center}
\end{minipage}
\begin{minipage}{0.5\textwidth}
\small \begin{center}
\inference[S-P-Final]{ \delta = (p \xrightarrow[]{ \mathsf{\phi} / \mathsf{f} } q) & (q  \in  F)& \mathcal{A} (q) = \mathsf{e_{q}}) \\
					( p \xrightarrow[]{ \mathsf{\phi'} / \mathsf{f'} } q') \in R \setminus \delta \\
					\llbracket \mathsf{\phi} \rrbracket \mathcal{H} = true & ( p \xrightarrow[]{ \mathsf{\phi'} / \mathsf{f'} } q') \  
					  \llbracket \mathsf{\phi'} \rrbracket \mathcal{H} = false \\
					  (\mathcal{H}; \mathsf{f}) \Downarrow (\mathcal{H'}; \_) & (\mathcal{H'}; \mathsf{e_q}) \Downarrow (\mathcal{H''}; {\sf v})
					  }  
				  {(\mathcal{H}, p);\delta \Rightarrow (\mathcal{H''},  \mathsf{v}) }

\end{center}
\end{minipage}

%
%
%

\caption{Evaluation rules for PST}
\label{fig:Ssemantics-pst}
\end{figure}
\pagebreak

\subsection{Complete Typing Semantics for PST}
\begin{figure}[h]
\begin{flushleft}

\textbf{Paths Typing}\,\fbox{\small
        $\Gamma \vdash$ $\mathsf{p_k}^{(p, q)}$ : $\tau$ 
           
}

\end{flushleft}

\small\begin{center}
\inference[{\sf T-P-Base}] {\delta = (p, \mathsf{\phi_g}, \mathsf{f}, q) & q \neq err & q \notin F &
				    \Sigma (\mathsf{f}) =  (({\sf i} : {\sf t})  \rightarrow \tau) \\
				    \tau = \textnormal{eff} \ \{\phi \} \tau' \{ \phi' \} &  
				    \Gamma_{{\sf ext}} =  \Gamma, i, (\mathsf{\phi_g} \ {\sf i} \ {\sf h})}  
				  {\Gamma_{{\sf ext}} \vdash     
				    \delta : \textnormal{eff} \ \{\phi \} \tau' \{ \phi' \} }

\end{center} 
\bigskip
\small \begin{center}

\inference[{\sf T-P-Ind}]{ \Gamma_1 \vdash \mathsf{p_{k-1}}^{(p,q)} :  (\textnormal{eff} \ \{\phi_{1}\} \  \tau_1 \ \{ \phi_{1'}\}) \ \ \
 				    {\sf term} (\mathsf{p_{k-1}}) = {\sf start} (\mathsf{p_{1}}) \\
				    \Gamma_2 \vdash \mathsf{p_{1}}^{(q,r)} : (\textnormal{eff} \ \{\phi_{2}\} \ \tau_2 \ \{ \phi_{2'}\}) &  
				     & \Gamma = \Gamma_1 $@$ \Gamma_2
				   } 
				  { \Gamma  \vdash  \mathsf{p_{k}}^{(p,r)} : \textnormal{eff} \ \{ \mathsf{path_{pre}} \} \ \tau_2 \ \{\mathsf{path_{post}} \} }

\end{center}
\bigskip

\hfill
\small \begin{center}
\inference[{\sf T-P-Error}]{ \delta = (p, \mathsf{\phi_g}, \mathsf{f}, {\sf err}) & \Sigma (\mathsf{f}) =  (({\sf i} : {\sf t})  \rightarrow \tau) \\
				    \tau = \textnormal{eff} \ \{\phi \} \ \mathsf{\tau'} \ \{ \phi' \} &  
				    \Gamma_{{\sf ext}} =  \Gamma, i, (\mathsf{\phi_g} \ {\sf i} \ {\sf h})\\
				    \textnormal{eff'} = \textnormal{eff} \sqcup \textnormal{exc} & 
				    \phi'^{*} = \forall \ {\sf h} \ {\sf v} \ {\sf h'}. ({\sf v  =  Err}) \wedge \phi'
				    }  
				  {\Gamma_{{\sf ext}} \vdash     
				     \mathsf{p_{1}}^{(p,{\sf err})} : \textnormal{eff'} \ \{\phi \} \ {\sf result \ \tau'} \ \{ \phi'^{*} \} }

\end{center}

\bigskip
\hfill
\small \begin{center}
\inference[{\sf T-P-Final}]{ \delta = (p, \mathsf{\phi_g}, \mathsf{f}, {\sf q}) & 
				    ({\sf q} \in F \setminus {\sf err}) & \Sigma (\mathsf{f}) =  (({\sf i} : {\sf t})  \rightarrow \tau) \\
				    \tau = \textnormal{eff} \ \{\phi \} \ \tau_1 \ \{ \phi' \} &  
				    \Gamma_{{\sf ext}} =  \Gamma, i, (\mathsf{\phi_g} \ {\sf i} \ {\sf h})\\
				    \mathcal{A}(\sf q) = \mathsf{e_q} & \Sigma (\mathsf{e_q}) =  \textnormal{eff} \ \{\phi1 \} \ \tau_2 \ \{ \phi1' \}				    
				    }  
				  {\Gamma_{{\sf ext}} \vdash     
				     \mathsf{p_{1}}^{(p,{\sf q } \in F)} : \textnormal{eff} \ \{ \mathsf{act_{pre}} \} \ \tau_2 \ \{ \mathsf{act_{post}} \}}

\end{center} 

\bigskip
\small \begin{center}
%
%

\end{center}



\small \begin{center}
 \inference[{\sf T-P-Loop}]{ \delta = (p, \mathsf{\phi_g}, {\sf f}, p) & \Sigma (\mathsf{f}) =  (({\sf i} : {\sf t})  \rightarrow \tau) \\
				    \tau = \textnormal{eff} \ \{\phi \} \ {\sf t} \ \{ \phi' \} &  \Gamma_{{\sf ext}} =  \Gamma, i, (\mathsf{\phi_g} \ {\sf i} \ {\sf h})\\
				    \Gamma_{{\sf ext}} \vdash \forall {\sf h \_ h'}. ({\sf Inv \ h} \wedge \mathsf{\phi_g} {\sf h} => \phi {\sf h}) => (\phi'\ {\sf h} \_ {\sf h'} => {\sf Inv \ h'})} 
				   { \Gamma_{{\sf ext}}  \vdash  
				      \mathsf{p_{1}}^{(p,p)} : \textnormal{eff} \ \{{\sf Inv} \ {\sf h} \} \ {\sf t} \ \{ {\sf Inv} \ {\sf h'} \}   }
				    
\end{center}

\bigskip
\begin{flushleft}
\textbf{PST Path Summation}\,\fbox{\small
        $\Gamma \vdash$ PST :  $\tau$ 
           
}
 
\end{flushleft}

\small \begin{center}
 \inference[T-P-Sum]{ paths ({\sf AP}) = \{  \mathsf{p_{k1}}^{(q0, (q \in F))}, \mathsf{p_{k2}}^{(q0, (r \in F))} \} \\
		      \Gamma_{1} \vdash \mathsf{p_{k1}}^{(q0, (q \in F))} : \textnormal{eff} \ \{ \phi1 \} \ \tau  \{ \phi1' \} \\ 
		       \Gamma_{2} \vdash \mathsf{p_{k2}}^{(q0, (r \in F))} : \textnormal{eff} \ \{ \phi2 \} \ \tau  \{ \phi2' \} 
			   }
		{ \Gamma_1 $@$ \Gamma_2 \vdash  
				      {\sf AP} : \textnormal{eff} \{ (\phi1 \wedge \phi2) \} \ \tau \ \{ (\phi1'  \lor \phi2') \} }
				    
\end{center}

%
%

\caption{Typing rules over PST paths}
\label{fig:typing-wff-pst}
\end{figure}

\pagebreak
\clearpage

\section{Details of Soundness Theorems and Proofs}

\subsection{Well-Typed Heap and Configuration-Typing}
We extend the typing judgements to a configuration typing $\Gamma \vdash$ ($\mathcal{H}$;{\sf e}) : $\tau$ using an extended notion of store typing as follows: 
\begin{definition}[Heap Interpretation Function]
A heap $\mathcal{H}$ in our evaluation rule is a list of references {\sf r1, r2,...} to values {\sf v}. 
To relate it to the logical heaps we use in our specification language {\sf h, h', etc.} we define the following interpretation function:
\[ \| . \|  = empty  \]
\[\| \mathcal{H}, ({\sf r} \mapsto {\sf v}) \| = {\sf update} \| \mathcal{H} \| \ {\sf  r} \ {\sf v}  \]
\end{definition}

Using this definition of interpretation we defined the notion of semantic entailment ($\models$) over a heap $\mathcal{H}$ in an environment $\Gamma$.
\begin{definition}[Semantic entailment]
$\Gamma \models \phi$ ($\mathcal{H}$) iff $\Gamma \vdash$ $\| \mathcal{H} \|$ $\implies$ $\phi$.  
\end{definition}

We also define a well-typed (extended notion of well-typed store) heap by using the standard {\it well-typed} store definition:
\begin{definition}[Well-Typed Heap]
 A heap $\mathcal{H}$ is well-typed with respect to a typing context $\Gamma$, given by the judgement $\Gamma \vdash \mathcal{H}$, iff for all r, such that  
 (r $\mapsto$ {\sf v}) $\in \mathcal{H}$ and $\| \mathcal{H} \|$ = h and $\Gamma \vdash$ (dom h r) $\wedge$ (sel h r v).
\end{definition}

\begin{definition}[Configuration Typing]
A configuration typing $\Gamma \vdash$ ($\mathcal{H}$;{\sf e}) : $\tau$ is defined by 
  $\Gamma \vdash$ $\mathcal{H}$ and $\Gamma \vdash$ {\sf e} : $\tau$.
\end{definition}

Using the above definitions, we define the soundness theorems for $L$-expressions, \name\ and PST Typing rules as follows:
We present a detailed proof for \name\ expression typing soundness as well as PST typing soundness and only formally state soundness for $L$-expression typing

\begin{lemma}[Preservation : $L$-expression]{\label{thm:preservation}}
If  $\Sigma \vdash$ ($\mathcal{H}$; {\sf e}) : $\tau$ and ($\mathcal{H}$ ;{\sf e}) $\Downarrow$ ($\mathcal{H'}$; {\sf e'}) 
then, for some $\Sigma$', such that $\Sigma \subseteq \Sigma'$, $\Sigma' \vdash$ ($\mathcal{H'}$; {\sf e'}) : $\tau$ 
\end{lemma}
\begin{proof}
Proof. The proof is based on induction over standard evaluation rules for $L$ and inversion of the corresponding typing rules along with a standard substitution or reduction lemma.
\end{proof}

\begin{lemma}[Progress : $L$ expression]{\label{thm:progress}}
If $\Sigma \vdash$ ($\mathcal{H}$;{\sf e}) : $\tau$, then either a) {\sf e} is a $L$-value or a $L$-function, 
 OR b) there exists {\sf e'}, such that ($\mathcal{H}$; {\sf e}) $\Downarrow$ ($\mathcal{H'}$; {\sf e'}).
\end{lemma}
\begin{proof}
 The proof based is standard induction over typing derivations for $L$-expressions.
\end{proof}

\begin{theorem}[Soundness : $L$ expression]{\label{thm:l-soundness}}
Given, $\Sigma \vdash$ {\sf e} : {\sf eff} \{$\phi$\} $\tau$ \{$\phi'$\}. 
 and  given there exists some $\mathcal{H}$, such that $\Gamma \models \phi (\mathcal{H}$). \\
 Then 
 \begin{itemize}
 \item $\exists \mathcal{H'}$. ($\mathcal{H}$; {\sf e}) $\Rrightarrow$ ($\mathcal{H'}$; {\sf v})
 \item $\Sigma \vdash $ {\sf v} : $\tau$
 \item $\Sigma, \phi$ ($\mathcal{H}$) $\models$ $\phi'$ ($\mathcal{H}$) {\sf v} ($\mathcal{H'}$)
 \end{itemize}
 
 \end{theorem}
\begin{proof}
 Using progress and preservation for $L$-expressions
\end{proof}

\begin{theorem}[Soundness \name\ Typing]{\label{lem:morpheussound}}
 Given, $\Gamma \vdash$ {\sf e} : {\sf eff} \{$\phi$\} $\tau$ \{$\phi'$\}. 
 and  given there exists some $\mathcal{H}$, such that $\Gamma \models \phi (\mathcal{H}$). \\
 Then 
 \begin{itemize}
 \item $\exists \mathcal{H'}$. ($\mathcal{H}$; {\sf e}) $\Rrightarrow$ ($\mathcal{H'}$; {\sf v})
 \item $\Gamma \vdash $ {\sf v} : $\tau$
 \item $\Gamma, \phi$ ($\mathcal{H}$) $\models$ $\phi'$ ($\mathcal{H}$) {\sf v} ($\mathcal{H'}$)
\end{itemize}
\end{theorem}
\begin{proof}
 The proof is inductively defined over each \name\ expression typing rule, giving the following cases:
 \begin{itemize}
 \item Case : T-Identity
 \begin{enumerate}
  \item Given $\Gamma \vdash$ $\iota$ AP : $\tau$
  \item Inverting the typing rule implies $\Gamma \vdash$ AP : $\tau$.
  \item Using the soundness result for the PST typing, $\exists$ $\mathcal{H'}$, such that 
  ($\mathcal{H}$; {\sf e}) $\Rightarrow^{*}$ ($\mathcal{H'}$; {\sf v}) and $\Gamma \vdash $ {\sf v} : $\tau$
  
  \item From (3) Evaluation rule CS-$\iota$ is applicable, $(\mathcal{H}; {\sf AP}) \Rrightarrow (\mathcal{H'}, {\sf v})$.
  \item Using (3) and (4) soundness is implied
  \end{enumerate}
 \item Case : T-Bind  
 \begin{enumerate}
  \item Given $\Gamma \vdash$ $\mathsf{CP_1} \textnormal{>>=} (\lambda {\sf x : t}. \mathsf{CP_2})$ : eff $\phi_{pre} \ \tau_2 \ \phi_{post}$.
  \item By Inverting the Typing rule (T-Bind), we get 
    $\exists \Gamma_0$, $\Gamma_0 \vdash$ $\mathsf{CP_1}$ : eff $\phi_1$ $\tau_1$ $\phi_1'$ and $\Gamma_0$, x : $\tau_1$ $\vdash$ $\mathsf{CP_2}$ : $\phi_1$ $\tau_1$ $\phi_1'$.
  \item Using IH, we can assume typing soundness for $\mathsf{CP_1}$  and $\mathsf{CP_2}$, this gives:
    \begin{itemize}
     \item $\exists \mathcal{H'}$, $(\mathcal{H}; \mathsf{CP_1}) \Rrightarrow (\mathcal{H'}, {\sf v1})$. 
     \item $\exists \mathcal{H''}$ $(\mathcal{H'}; \mathsf{CP_2}[\ \mathsf{v1} / \mathsf{x}]) \Rrightarrow (\mathcal{H''}, {\sf v})$.
    \end{itemize}
  \item (3) implies that rule (CS-bind) is applicable, 
  giving ($\mathcal{H}$; $\mathsf{CP_1} \textnormal{>>=} (\lambda {\sf x : t}. \mathsf{CP_2})$ $\Rrightarrow$ $(\mathcal{H''}; {\sf v2})$.
  \item (4) implies first part of the goal i.e. $\exists \mathcal{H'}$. ($\mathcal{H}$; {\sf e}) $\Rrightarrow$ ($\mathcal{H'}$; {\sf v}).
  \item Again using IH, $\Gamma$, {\sf x} : $\tau_1$ $\vdash$ {\sf v2} : $\tau_2$, implying the second part of the goal.
  \item Using IH and Well-typedness of the heap, on $\mathsf{CP_1}$ and $\mathsf{CP_2}$ we get 
    \begin{itemize}
     \item $\Gamma$, $\phi_1$ ($\mathcal{H}$) $\vDash$ $\phi_1' (\mathcal{H} \ {\sf v1}  \ \mathcal{H'})$
     \item $\Gamma$, x : $\tau_1$, $\phi_2$ ($\mathcal{H'}$) $\vDash$ $\phi_1' \mathcal{H'} {\sf v2} \mathcal{H''}$
   \end{itemize}
  \item Using T-Bind rule, constructed $\phi_{pre}$ = \{ $\forall {\sf h}. \phi1 \ {\sf h}  \wedge \phi1' ({\sf h, x, h_i})  => \phi2 \ {\sf h_i}$ \}
  \item And $\phi_{post}$ = \{ $\forall {\sf h, y, h'}. \phi1' ({\sf h, x, h_i}) \wedge \phi2' ({\sf h_i, y, h'})$ \}
  \item The third part of the goal becomes: $\Gamma$,  x : $\tau_1$, $\mathsf{h_i}$ : heap $\vDash$ $\phi_{pre}$ =>  $\phi_{post}$
  \item Using Well-Typedness of the heap, the above becomes:
  $\Gamma$,  x : $\tau_1$ $\vDash$ $\phi_1$ $\mathcal{H}$ $\wedge$ ($\phi_1'$ ($\mathcal{H}$) {\sf x} ($\mathcal{H'}$)  => $\phi_2$ \ ($\mathcal{H'}$))
    =>  $\phi_{post}$
  \item Using 7 
  $\Gamma$,  x : $\tau_1$ $\vDash$ ($\phi_1$ $\mathcal{H}$ $\wedge$ ($\phi_1'$ ($\mathcal{H}$) {\sf x} ($\mathcal{H'}$)  => $\phi_2$ \ ($\mathcal{H'}$))
  $\wedge$ $\phi_1' (\mathcal{H} \ {\sf v1} \ \mathcal{H'})$ $wedge$
  $\phi_1' \mathcal{H'} {\sf v2} \mathcal{H''}$) => $\phi_{post}$
  \item Which is trivially true.
 \end{enumerate}

 \item Case : T-seq : Similar to the T-Bind case
 \item Case : T-Fold 
 \begin{enumerate}
  \item Given $\Gamma_{e}$ $\vdash$ foldT ... : eff \{P\} $\tau_1$ \{Q\}.
  \item Using T-Fold inversion: $\exists$ $\Gamma$, $\Gamma$ $\vdash$ CP : eff $\phi$ $\tau$ $\phi'$ and $\Gamma \vdash$ f : \_, g : \_ 
  \item Using IH for CP, $\exists \mathcal{H'}$. $(\mathcal{H}; \mathsf{CP}) \Rrightarrow (\mathcal{H'}; {\sf v})$.
  \item Using case split on {\sf v}
  \begin{itemize}
   \item Case v $\neq$ {\sf Err}
	\begin{enumerate}
	 \item Using progress for $L$-expressions, we get 
	  \begin{itemize}
	   \item $(\mathcal{H'} ; {\sf f \ v \ b}) \Downarrow (\mathcal{H'}; {\sf b'})$
	   \item $(\mathcal{H'} ; {\sf g \ b'}) \Downarrow (\mathcal{H'}; {\sf brk})$
	  \end{itemize}
	  \item Again, case splitting on {\sf brk}
	  \item Case {\sf brk} = true 
	     \begin{itemize}
	      \item The (CS-break) rule is applicable. Thus
	      $(\mathcal{H}; (\mathsf{foldT \ {\sf CP} \ f \ g \ b})) \Rrightarrow (\mathcal{H'};{\sf b'})$
	      \item This implies the first part of the goal.
	      \item Using (2), $\Gamma \vdash$  {\sf b} : $\tau_1$, thus implying the second part of the goal.
	      \item Premise for the third part of the goal becomes:
	      $\Gamma, \mathsf{Inv_{start}}$ $\wedge$ $\mathsf{Inv_{break}}$ $\wedge$ {\sf P} ($\mathcal{H}$) 
	      \item Using definitions of $\mathsf{Inv_{start}}$ and $\mathsf{Inv_{break}}$, this reduces to:  
	      P $\mathcal{H}$ => (Inv $\mathcal{H}$ b) $\wedge$ (g b' = true $\wedge$ (Inv $\mathcal{H'}$ b')) => Q $\mathcal{H'}$ b'
	      \item The above premise implies the required goal condition Q ($\mathcal{H'}$) {\sf b'}.
	     \end{itemize}

	 \item Case {\sf brk} = false
	     \begin{itemize}
	     \item Using Another induction on the length of the input, Using IH and finiteness of the input gives:
	      $\exists \mathcal{H'}$, such that  $(\mathcal{H'}; (\mathsf{foldT \ {\sf CP} \ f \ g \ b'})) \Rrightarrow (\mathcal{H''}; {\sf v'})$
	      \item Thus (CS-loop) rule is applicable, giving 
	      $(\mathcal{H}; (\mathsf{foldT \ {\sf CP} \ f \ g \ b})) \Rrightarrow (\mathcal{H''};{\sf v'})$
	      \item This implies the first part of the goal.
	      \item Using IH and finiteness of the the input, we get $\Gamma$ ${\sf v'}$ : $\tau_1$, implying the second part of the goal
	      \item Premise for the third part of the goal becomes:
	      $\Gamma, \mathsf{Inv_{start}}$ $\wedge$ $\mathsf{Inv_{break}}$ $\wedge$ $\mathsf{Inv_{ind}}$ $\wedge$ {\sf P} ($\mathcal{H}$) 
	      \item By definition of $\mathsf{Inv_{start}}$, $\mathsf{Inv_{break}}$ and  $\mathsf{Inv_{ind}}$ this reduces to  
	      P $\mathcal{H}$ => (Inv $\mathcal{H}$ b)  
	      $\wedge$ (Inv $\mathcal{H'}$ b' $\wedge$ ( g b' = false)) => $\phi$ $\mathcal{H'}$ => 
	      ($\phi'$ $\mathcal{H''}$ {\sf v'} $\mathcal{H'}$  => Inv $\mathcal{H''}$ {\sf v'})
	      $\wedge$ (( g v' = false) $\wedge$ Inv $\mathcal{H''}$ v' => Q $\mathcal{H''}$ v')
	      \item This implies the required goal condition Q ($\mathcal{H''}$) {\sf v'}.
	     \end{itemize}

	\end{enumerate}

   \item Case v = {\sf Err}
	\begin{enumerate}
	 \item The CS-backtrack is applicable, thus 
	 ($\mathcal{H}$; {\sf foldT \_}) $\Rrightarrow$ $(\mathcal{H'}; {\sf b})$
	 \item This implies the first part of the goal.
	 \item Using (2), $\Gamma$ {\sf b} : $\tau_1$, thus implying the second part of the goal.
	 \item The premise for the third goal becomes:
	  $\Gamma, \mathsf{Inv_{start}}$ $\wedge$ $\mathsf{Inv_{break}}$ $\wedge$ {\sf P} ($\mathcal{H}$) 
	 \item  By definition of $\mathsf{Inv_{start}}$ and $\mathsf{Inv_{break}}$, this reduces to  
	  P $\mathcal{H}$ => (Inv $\mathcal{H}$ b) $\wedge$ (g b = true $\wedge$ (Inv $\mathcal{H'}$ b)) => Q $\mathcal{H'}$ b 
	 \item This implies the required goal condition Q ($\mathcal{H'}$) {\sf b}.
	\end{enumerate}

  \end{itemize}
  
 \end{enumerate}

 \item Case : T-choice : 
 \begin{enumerate}
  \item Given $\Gamma \vdash (\mathsf{CP_1} \textnormal{<|>}  \mathsf{CP_2}) : \textnormal{eff'} \ \{ (\phi1 \wedge \phi2) \} \ \tau \ \{ (\phi1'  \lor \phi2')$ \}. 
  \item Using T-Choice inversion: $\Gamma$ $\vdash$ $\mathsf{CP_1}$ : eff $\phi1$ $\tau$ $\phi1'$ and $\Gamma$ $\vdash$ $\mathsf{CP_2}$ : eff $\phi2$ $\tau$ $\phi2'$.
  \item By IH, $(\mathcal{H}; \mathsf{CP_1}) \Rrightarrow (\mathcal{H'}; {\sf v1})$ and $(\mathcal{H}; \mathsf{CP_2}) \Rrightarrow (\mathcal{H''}; {\sf v2})$.
  \item Thus CS-L and CS-R rules are applicable, giving: two possible output heaps and values:
  $(\mathcal{H}; (\mathsf{CP_1} \mid \mathsf{CP_2})) \Rrightarrow (\mathcal{H'}; {\sf v1})$ and 
  $(\mathcal{H}; (\mathsf{CP_1} \mid \mathsf{CP_2})) \Rrightarrow (\mathcal{H''}; {\sf v2})$
  \item This implies the first part of the goal.
  \item Using IH, $\Gamma \vdash$ v1 : $\tau$ and $\Gamma \vdash$ v2 : $\tau$, implying the second part of the goal.
  \item Again using IH:
  \begin{itemize}
   \item $\Gamma, \phi1(\mathcal{H}$ $\vDash$ $\phi1'$ $\mathcal{H}$ v1 $\mathcal{H'}$ 
   \item $\Gamma, \phi2(\mathcal{H}$ $\vDash$ $\phi2'$ $\mathcal{H}$ v2 $\mathcal{H''}$ 
 \end{itemize}
  \item Taking conjunction of the pre-conditions we get 
  $\Gamma, \phi1(\mathcal{H}$) $\wedge$ $\phi2 (\mathcal{H})$ $\vDash$ $\phi1'$ $\mathcal{H}$ v1 $\mathcal{H'}$ $\wedge$ $\phi2'$ $\mathcal{H}$ v2 $\mathcal{H''}$.
  \item Since the post-condition in the above judgement implies the inferred post-condition, i.e.
  $\phi1'$ $\mathcal{H}$ v1 $\mathcal{H'}$ $\wedge$ $\phi2'$ $\mathcal{H}$ v2 $\mathcal{H''}$ => ($\phi1'$ $\mathcal{H}$ v1 $\mathcal{H'}$) $\lor$ ($\phi2'$ $\mathcal{H}$ v2 $\mathcal{H''}$).
  \item A non-Deterministic pick for the result implies :
  $\phi1'$ $\mathcal{H}$ v1 $\mathcal{H'}$ $\wedge$ $\phi2'$ $\mathcal{H}$ v2 $\mathcal{H''}$ => ($\phi1'$ $\lor$ $\phi2'$) $\mathcal{H}$ v $\mathcal{H}_{out}$).
  \item This  implies the required third goal.
  
 \end{enumerate}
  \item Trivial cases T-App, T-Let, T-Abs
  \end{itemize}

\end{proof}

\begin{lemma}[Soundness Subtyping]{\label{lem:subtyping-soundness}}
 Given $\Gamma \vdash {\sf e}$ : $\tau_1$ and $\exists \tau_1$, 
 such that $\Gamma$ $\vdash$ $\tau_1$ <: $\tau_2$, then $\Gamma$ $\vdash$ {\sf e} : $\tau_2$.
 \end{lemma}
\begin{proof}
  The proof inductively reasons over three subtyping relations:
  \begin{enumerate}
  \item Given $\Gamma \vdash {\sf e}$ : $\tau_1$ and $\exists \tau_1$, such that $\Gamma$ $\vdash$ $\tau_1$ <: $\tau_2$, then $\Gamma$ $\vdash$ {\sf e} : $\tau_2$.
  \item To prove  $\Gamma$ $\vdash$ {\sf e} : $\tau_2$ is sound, i.e. 
   \item We split on the shape of $\tau_1$
  \begin{itemize}
   \item Base-refinement ($\tau_1$ = \{ v : t | $\phi$\}) 
   \begin{enumerate}
    \item Using T-Sub-Base $\tau_2$ is of the form  \{ v : t | $\phi'$\})
    \item Trivial case, as no effect associated.
    \item From the given and the definition of expression typing $\Gamma \vDash \phi {\sf v}$.
    \item To prove soundness we are required to prove $\Gamma \vDash \phi' {\sf v}$.
    \item Inverting the T-Sub-Base gives, $\Gamma \vDash \phi \implies \phi'$, which gives the required condition.
       \end{enumerate}

    \item Computation-type-refinement ($\tau_1$ = eff1 \{$\phi_1$\} v : $\tau_{11}$ \{$\phi_1'$\})
    \begin{enumerate}
      \item Using T-Sub-Comp $\tau_2$ is of the form  eff2 \{$\phi_2$\} v : $\tau_{22}$ \{$\phi_2'$\})
      \item From the given and the definition of expression typing $\Gamma \vDash \phi {\sf v}$.
      \item To prove soundness we are required to prove, the following 3 goals from the soundness theorem:
      \begin{itemize}
       \item $\exists \mathcal{H'}$. ($\mathcal{H}$; {\sf e}) $\Rrightarrow$ ($\mathcal{H'}$; {\sf v})
       \item $\Gamma \vdash $ {\sf v} : $\tau_{12}$
       \item $\Gamma, \phi_2$ ($\mathcal{H}$) $\models$ $\phi_2'$ ($\mathcal{H}$) {\sf v} ($\mathcal{H'}$)
      \end{itemize}
      \item Using the given typing of {\sf e}, the first goal follows.
      \item From given $\Gamma \vdash $ {\sf v} : $\tau_{11}$, Using T-Sub-Comp, $\tau_{11}$ <:$\tau_{12}$.
      \item Thus using T-Sub rule, $\Gamma \vdash $ {\sf v} : $\tau_{12}$. hence the second goal follows.
      \item From given, $\Gamma, \phi_1$ ($\mathcal{H}$) $\models$ $\phi_1'$ ($\mathcal{H}$) {\sf v} ($\mathcal{H'}$),
      \item Using T-Sub-Comp again, $\Gamma \vDash$ $\phi_2$ => $\phi_1$ and $\Gamma, \phi_2 \vDash$ $\phi_1'$ => $\phi_2'$.
      \item Using the above two facts, $\Gamma, \phi_2$ ($\mathcal{H}$) $\models$ $\phi_2'$ ($\mathcal{H}$) {\sf v} ($\mathcal{H'}$)
      \item Thus, the third goal holds.
      \end{enumerate}
 
   \item Arrow-refinement (x : $\tau_{11}$  $\rightarrow$ $\tau_{12}$)
    \begin{enumerate}
      \item Using T-Arrow-Base $\tau_2$ is of the form  (x : $\tau_{11}$  $\rightarrow$ $\tau_{12}$)
      \item From the given and the definition of expression typing $\Gamma \vDash \phi {\sf v}$.
      \item To prove soundness we are required to prove $\Gamma \vDash \phi' {\sf v}$.
      \item Follows from the case analysis and induction over the type of the body.
    \end{enumerate}
  \end{itemize}
  \end{enumerate}

\end{proof}

\begin{lemma}[Soundness Path Typing]{\label{lem:pathsound}}
 Given, $\Gamma \vdash$ $\mathsf{p_{k}}^{(p,q)}$ : {\sf eff} \{$\phi_k$\} $\mathsf{t_{k}}$ \{$\phi'_k$\}. 
 Also Given there exists some $\mathcal{H}$, such that $\Gamma \models \phi_k (\mathcal{H}$). \\
 Then there exists $\mathcal{H'}$, such that:
 \begin{itemize}
  \item Either: ($\mathcal{H}$, {\sf p});$\mathsf{p_{k}}^{(p,q)}$ $\Rightarrow$ ($\mathcal{H'}$; {\sf v}),  $\Gamma \vdash $ {\sf v} : $\mathsf{t_k}$ and 
  $\Gamma, \phi_k$ ($\mathcal{H}$) $\models$ $\phi'_k$ ($\mathcal{H}$) {\sf v} ($\mathcal{H'}$).
  \item Or: ($\mathcal{H}$, {\sf p});$\mathsf{p_{k}}^{(p,q)}$ $\Rightarrow$ ($\mathcal{H'}$; {\sf q}) and 
  $\Gamma, \phi_k$ ($\mathcal{H}$) $\models$ $\phi'_k$ ($\mathcal{H}$) ($\mathcal{H'}$).
  
 \end{itemize}
\end{lemma}
\begin{proof}
 The proof is inductively defined on the length of the path, giving two cases:
 \begin{enumerate}
  \item {\bf Base Case (k = 1)} : We do a case analysis on typing rule for path of length 1:
  \begin{itemize}
   \item Case : T-P-Base: Inverting (T-P-Base) implies: \\
    \begin{enumerate}
      \item $\exists \delta$ = (p, $\mathsf{\phi_g}$, $\mathsf{f}$, {\sf q})
      \item $\exists \Gamma_0$, $\Gamma$ = $\Gamma_0$, $\phi_g$ {\sf i} {h} , where {\sf h} = $\| \mathcal{H}\|$ 
      \item  Using the progress for $\Sigma$ typing and 
      \item Deterministic semantics of the Transducers. 
      \item Using (a)-(d), The evaluation rule (S-P-Base) applies, thus: 
      \item There exists a heap $\mathcal{H'}$ such that ($\mathcal{H}$, {\sf p}); $\delta$ $\Rightarrow$ ($\mathcal{H'}$; {\sf q}).
      \item Using sound typing result (Theorem ~\ref{thm:progress} and ~\ref{thm:preservation}) for $L$-expression we get ($\Sigma$, $\phi$ $\models$ $\phi'$ $\mathcal{H}$ $\mathcal{H'}$.) Using (T-P-Base), The pre and post type annotations, $\phi_k$ = $\phi$  and $\phi'_k$ = $\phi'$. Given, $\Gamma$ is parameterized over $\Sigma$, we have  $\Gamma, \phi_k$ ($\mathcal{H}$) $\models$ $\phi'_k$ ($\mathcal{H}$) ($\mathcal{H'}$).
      \item f and g $\implies$ soundness
    \end{enumerate}
     \item Case : T-P-Error : Inverting (T-P-Error) implies: \\
     \begin{enumerate}
      \item $\exists \delta$ = (p, $\mathsf{\phi_g}$, $\mathsf{f}$, {\sf err}), 
      \item $\exists \Gamma_0$, $\Gamma$ = $\Gamma_0$, $\phi_g$ {\sf i} {h} , where {\sf h} = $\| \mathcal{H}\|$. 
      \item Using the progress for $\Sigma$ typing and 
      \item Deterministic semantics of the Transducers. 
      \item Using (a)-(d), The evaluation rule (S-P-Err) applies, thus: 
      \item  There exists a heap $\mathcal{H'}$ such that ($\mathcal{H}$, {\sf p}); $\delta$ $\Rightarrow$ ($\mathcal{H'}$; {\sf Err}).
      \item Using sound typing result (Theorems ~\ref{thm:progress} and ~\ref{thm:preservation}) for $L$-expression we get ($\Sigma$, $\phi$ $\models$ $\phi'$ $\mathcal{H}$ $\mathcal{H'}$ 
      \item Using (T-P-Error), The pre and post type annotations, $\phi_k$ = $\phi$  and $\phi'_k$ = ($\phi' \wedge$ {\sf v = Err}). Given, $\Gamma$ is parameterized over $\Sigma$, we have  $\Gamma, \phi_k$ ($\mathcal{H}$) $\models$ $\phi'_k$ ($\mathcal{H}$) ($\mathcal{H'}$).
     \item Using (T-P-Error), $\mathsf{t_k}$ = {\sf result t} and $\phi'_k$ = ($\phi' \wedge$ {\sf v = Err}), thus {\sf v = Err} : $\mathsf{t_k}$.
     \item f, h and i $\implies$ soundness.
   
     \end{enumerate}
  
   \item Case : T-P-Act Similar to above two cases, using (S-P-Act) and soundness over the action of the final state. 
   \item Case : T-P-Loop : This case uses the definition of the {\sf Inv} and the Invariant checking predicate to prove the soundness.
   Inverting (T-P-Loop) implies:\\
    \begin{enumerate}
      \item $\exists \delta$ = (p, $\mathsf{\phi_g}$, $\mathsf{f}$, {\sf p})
      \item $\Gamma$ = $\Gamma_0$, $\phi_g$ {\sf i} {h} , where {\sf h} = $\| \mathcal{H}\|$ 
      \item  Using the progress for $\Sigma$ typing and 
      \item  Well-formed semantics of the Transducers. 
      \item  Using (a)-(d), The evaluation rule (S-P-Loop) applies, thus: 
      \item There exists a heap $\mathcal{H'}$ such that ($\mathcal{H}$, {\sf p}); $\delta$ $\Rightarrow$ ($\mathcal{H'}$; {\sf p}).
      \item Using sound typing result (Theorem ~\ref{thm:progress} and ~\ref{thm:preservation}) for $L$-expression we get ($\Sigma$, $\phi$ $\models$ $\phi'$ $\mathcal{H}$ $\mathcal{H'}$.) Using (T-P-Loop), and the definition of $\mathsf{Inv_{ind}}$ 
      $\mathsf{\Gamma_{ext}}$ $\models$ {\sf Inv} ($\mathcal{H}$) ($\mathcal{H'}$).
      \item f and g $\implies$ soundness.
    \end{enumerate}
   
    \end{itemize}
  \item {\bf Inductive Case (k > 1)} : The only typing rule applicable in this case: 
  \begin{itemize}
   \item T-P-Ind : Inverting the (T-P-Ind) rule implies:\\
   
   \begin{enumerate}
   \item  $\exists$ a path of length (k-1), ($\mathsf{p_{k-1}}^{(p,q)}$), such that for some $\Gamma_1 \subseteq \Gamma \vdash$ ($\mathsf{p_{k-1}}^{(p,q)}$) : {\sf eff} \{$\phi_{k-1}$\} $\mathsf{t_{k-1}}$ \{$\phi'_{k-1}$\} 
   \item  $\exists$ a path of length (1), ($\mathsf{p_{1}}^{(q,r)}$), such that for some $\Gamma_2 \subseteq \Gamma \vdash$ ($\mathsf{p_{1}}^{(q,r)}$) :{\sf eff} \{$\phi_1$\} $\mathsf{t_{1}}$ \{$\phi'_1$\} 
   \item  terminal control state for ($\mathsf{p_{k-1}}^{(p,q)}$) = start control state ($\mathsf{p_{1}}^{(q,r)}$).
   \item  Using Induction Hypothesis on soundness of ($\mathsf{p_{k-1}}^{(p,q)}$) implies ($\mathcal{H}$, {\sf p});($\mathsf{p_{k-1}}^{(p,q)}$)  $\Rightarrow$ ($\mathcal{H}_{int1}$; {\sf q}) 
   \item ($\mathsf{p_{1}}^{(q,r)}$), implies ($\mathcal{H}_{int2}$, {\sf p});($\mathsf{p_{1}}^{(q,r)}$)  $\Rightarrow$ ($\mathcal{H'}$; {\sf r}).
   \item  Using (T-P-Ind) $\Gamma$ $\vdash$ $\phi$ the guard for the ($\mathsf{p_{1}}^{(q,r)}$).
   \item  Main induction step : Using (T-P-Ind), $\mathsf{path_{pre}}$ implies ($\phi'_{k-1}$ ($\mathcal{H}_{int1}$)  $\implies$ ($\phi_{1}$ $\mathcal{H}_{int2}$), thus, ($\mathcal{H}_{int2}$; {\sf q}) $\Rightarrow$ ($\mathcal{H'}$; {\sf r}), thus (S-P-Ind) is applicable and 
   ($\mathcal{H}$, {\sf p});($\mathsf{p_{k}}^{(p,r)}$)  $\Rightarrow$ ($\mathcal{H'}$; {\sf r}).
   \item  Using Induction Hypothesis, $\Gamma_1, \phi_{k-1}$ ($\mathcal{H}$) $\models$ $\phi'_{k-1}$ ($\mathcal{H}$) ($\mathcal{H}_{int1}$) and \\
   $\Gamma_2, \phi_{1}$ ($\mathcal{H}_{int2}$) $\models$ $\phi'_{1}$ ($\mathcal{H}_{int2}$) ($\mathcal{H'}$) and 
   Thus, $\Gamma_1 @ \Gamma_2$ $\vdash$ $\mathsf{path_{pred}} \models $ $\mathsf{path_{post}}$ ($\mathcal{H}$) ($\mathcal{H}_{int1}$) ($\mathcal{H}_{int2}$) ($\mathcal{H'}$), giving us our required proof of soundness for $\mathsf{p_{k}}^{(p,r)}$.
   \item g and h $\implies$ soundness.
   \end{enumerate}
  \end{itemize}

 \end{enumerate}

\end{proof}

\begin{theorem}[Soundness PST]
Given a correctness specification $\tau$ = {\sf eff} \{$\phi$\} {\sf v:  t} \{$\phi'$\} and a {\sf PST T}, such that 
under some $\Gamma$, $\Gamma \vdash$ {\sf T} : $\tau$, then if there there exists $\mathcal{H}$, such that $\Gamma \models \phi (\mathcal{H})$, then there exists a $\mathcal{H'}$, such that, 
\begin{itemize}
 \item ($\mathcal{H}$, $q_0$); {\sf T} $\Rightarrow$ ($\mathcal{H'}$; {\sf v})
 \item $\Gamma \vdash $ {\sf v} : {\sf t}
 \item $\Gamma, \phi$ ($\mathcal{H}$) $\models$ $\phi'$ ($\mathcal{H}$) {\sf v} ($\mathcal{H'}$)
\end{itemize}
\end{theorem}
\begin{proof}
The proof uses T-P-Sum rule, Inverting the rule implies:
\begin{enumerate}
 \item $\exists$ two paths $\mathsf{p_{k1}}^{(q0, (q \in F))}$, and $\mathsf{p_{k1}}^{(q0, (r \in F))}$.
 \item Using the soundness rule for paths (Lemma~\ref{lem:pathsound}), either (S-P-L) or (S-P-R) apply.
 \item Thus, $\exists$ either $\mathcal{H'}$, {\sf vleft} or $\mathcal{H''}$, {\sf vright} such that.
 \item ($\mathcal{H}$, $q_0$); {\sf PST} $\Rightarrow$ ($\mathcal{H'}$; {\sf vleft}). Or 
 \item ($\mathcal{H}$, $q_0$); {\sf PST} $\Rightarrow$ ($\mathcal{H''}$; {\sf vright}). 
 \item Using (T-P-Sum) and Soundness of path typing (Lemma~\ref{lem:pathsound}) $\Gamma \vdash$ {\sf vleft} : {\sf t} and {\sf vright} : {\sf t}, Thus, 
 $\Gamma \vdash$ {\sf v} : {\sf t}.
 \item Using (T-P-Sum) and Soundness of path typing (Lemma~\ref{lem:pathsound}), $\Gamma, \phi1$ ($\mathcal{H}$) $\models$ $\phi1'$ ($\mathcal{H}$) {\sf vleft} ($\mathcal{H'}$) and  $\Gamma, \phi2$ ($\mathcal{H}$) $\models$ $\phi2'$ ($\mathcal{H}$) {\sf vright} ($\mathcal{H''}$).
 \item Thus $\Gamma, \phi1 \wedge \phi2$ ($\mathcal{H}$) $\models$ ($\phi1'$ ($\mathcal{H}$) {\sf vleft} ($\mathcal{H'}$)) $\lor$ ($\phi1'$ ($\mathcal{H}$) {\sf vright} ($\mathcal{H''}$)).
 \item Using 4, 5, 6, 8 $\implies$ soundness condition.
 
\end{enumerate}
\end{proof}

%

\section{Benchmark Grammars}
Following are the grammars for the Benchmark applications:
\begin{enumerate}
 \item {\bf PNG}
\begin{lstlisting}[escapechar=\@,basicstyle=\small\sf,breaklines=true]
 png : header . many chunk 
 chunk : length . typespec . content . Pair (length,content)
 length : number 
 typespec : char 
 content : char*
\end{lstlisting}
 \item {\bf PPM}
  
\begin{tabbing}
\small
 {\sf ppm} : {\sf ``P''} . {\sf versionnumber} . {\sf header} . {\sf data} \\
 {\sf versionnumber} :  digit \\
 {\sf header } : width = number . height = number . max = number \\
 {\sf data } : rows* [length (rows) = height]\\
 {\sf row } : rgb* [length (rgb) = width]  \\
 {\sf rgb } : r = number . g = number  . b= number [r < max, g < max, b < max] \\ 
 {\sf number} :  digit* \\
 
 \end{tabbing}
\item {\bf Simple-Haskell}

\begin{lstlisting}[escapechar=\@,basicstyle=\small\sf,breaklines=true]
caseexp : offside ('case' . exp . 'of') . offside (align alts) 
alts :  (alt) . (alt)* 
alt : pat ralt;
ralt : ('->' exp)
pat : exp
exp : varid
varid : [a-z, 0-9]*
\end{lstlisting}

\item {\bf Simple-Python}
\begin{lstlisting}[escapechar=\@,basicstyle=\small\sf,breaklines=true]
while_stmt: offsie 'while' . offside test . offside ':' . offside suite
suite: offside NEWLINE . offside stmt+
test: expr op expr
expr : identifier
stmt: small_stmt NEWLINE
small_stmt: expr op expr
op : > | < | = 
\end{lstlisting}

\item {\bf Xauction}
\begin{lstlisting}[escapechar=\@,basicstyle=\small\sf,breaklines=true]
  listing :  sellerinfo .  auctioninfo
  sellerinfo :  sname . srating 
  auctioninfo  : bidderinfo+  
  bidderinfo  : bname . brating
  sname :  "<name>" name "</name>"
  srating :  "<rating>" number "</rating>"
  bname :  "<name>" name "</name>"
  brating : "<rating>" number "</rating>"

  name : [a-z].[a-z,0-9]*
  number : [0-9]+
\end{lstlisting}

\item {\bf Xprotein} where {\sf proteins} is a global list of parsed proteins. 
\begin{lstlisting}[escapechar=\@,basicstyle=\small\sf,breaklines=true]

proteindatabase  = database  proteinentry+
database =  <database> uid </database>
proteinentry =  <ProteinEntry> header  protein 	skip* </ProteinEntry> 
header = <header> . uid .</header>
uid = number
protein = <protein> name . id . </protein> [@$\neg$@ (name @$\in$@ proteins)] {proteins.add name}
\end{lstlisting}

\item {\sf Health}
\begin{itemize}
 \item Following is the custom-stateful regex patter-matcher
\begin{lstlisting}[escapechar=\@,basicstyle=\small\sf,breaklines=true]
Custom Regex pattern = 
	<skip> ([^,]*, {4}) 
		(<?round-off> cancer-deaths) 
			@$\lambda$@ x. (<skip>[,*,] {2}) 
				(<?check-less-than x> cancer-deaths-min) 
					@$\lambda$@ y. 
					(<?check-greater-than x> cancer-deaths-max) @$\lambda $@ z.(<Triple {x;y;z}> [*\n])
\end{lstlisting}
\item 
Following is the grammar capturing the above pattern:
\begin{lstlisting}[escapechar=\@,basicstyle=\small\sf,breaklines=true]
csvhealth : count 4 (skip) . x = cancer-deaths . (count 2 skip) . y = cancer-deaths-min [y < x] . z = cancer-deaths-max [z > x]   
skip : [a-z]*.','
cancer-deaths : number 
cancer-deaths-min : number 
cancer-deaths-max : number
\end{lstlisting}

\end{itemize}

\item {\bf Streams}

\begin{lstlisting}[escapechar=\@,basicstyle=\small\sf,breaklines=true]
streamicc :  t = tagentry . chunk (t)
tagentry : signature .  offset  .  size
signature : number
offset : number 
size : number
chunk (t) : s = GetStream . s1 = Take (t.sz) s . SetStream s1 . Tag (t.signature) . s2 = Drop sz s . Setstream s2
  
GetStreamm : !inp
SetStream (s1) : inp := s1
Tag (choice) : tag-left [choice=0]| tag-right[choice=1]
tag-left :x = number [x = 0] 
tag-right : x = number [x =1]
\end{lstlisting}

\item {\bf Typedef}

\begin{lstlisting}[escapechar=\@,basicstyle=\small\sf,breaklines=true]
decl := "typedef" . 
		typeexpr . 
		id=rawident [@$\neg$@ id @$\in$@ (!identifiers)]
		{types.add id} 

typename := x = rawident [x @$\in$@ (!types)]{return x}
typeexp := "int" | "bool"
expr :=  id=rawident {identifiers.add id ; return id}
program := many decl . many expr
\end{lstlisting}

\end{enumerate}

%
%
%

%

%
%
%